\documentclass[12pt, a4paper,reqno,english]{smfart}
\usepackage[utf8]{inputenc}
\usepackage{amsmath, amssymb, amsthm, physics, mathrsfs, bm}
\usepackage{booktabs,longtable,array,enumitem,xcolor}
\usepackage{graphicx, color}
\usepackage[colorlinks=true, linkcolor=blue, citecolor=blue, urlcolor=blue]{hyperref}
\usepackage{geometry}
\usepackage{soul}
\soulregister\ref7
\usepackage{cite}
\usepackage{listings}
\definecolor{codegreen}{rgb}{0,0.6,0}
\definecolor{codegray}{rgb}{0.5,0.5,0.5}
\definecolor{codepurple}{rgb}{0.58,0,0.82}
\definecolor{backcolour}{rgb}{0.07,0.07,0.06}
\numberwithin{equation}{section}
\definecolor{bubblegum}{rgb}{0.99, 0.76, 0.8}
\definecolor{ietocean}{RGB}{0, 30, 140}
\definecolor{ietcoast}{RGB}{0, 150, 173}
\definecolor{ietlagoon}{RGB}{0, 216, 180}
\definecolor{veriforest}{RGB}{0, 80, 40}
\definecolor{revred}{RGB}{153, 0, 0}
\hypersetup{
        colorlinks,%
        citecolor={ietlagoon},%
        filecolor={ietocean},%
        linkcolor={veriforest},%
        urlcolor={revred},
        bookmarksnumbered=true,     
        bookmarksopen=true,         
        bookmarksopenlevel=1,       
        pdfstartview=Fit,           
        pdfpagemode=UseOutlines,
        pdfpagelayout=TwoPageRight
}
  
\newtheorem{theorem}{Theorem}[section]
\newtheorem{proposition}[theorem]{Proposition}
\newtheorem{lemma}[theorem]{Lemma}
\newtheorem{corollary}[theorem]{Corollary}

\theoremstyle{definition}
\newtheorem{definition}[theorem]{Definition}

\newtheorem{convention}[theorem]{Convention}
\theoremstyle{remark}
\newtheorem{remark}[theorem]{Remark}

\newcommand{\Q}{\mathbb Q}
\newcommand{\Z}{\mathbb Z}
\newcommand{\C}{\mathbb C}
\newcommand{\PP}{\mathbb P}

\renewcommand{\ev}{\operatorname{ev}}
\newcommand{\Tot}{\operatorname{Tot}}
\newcommand{\vir}{\mathrm{vir}}
\newcommand{\vdim}{\operatorname{vdim}}

\newcommand{\Li}{\operatorname{Li}}

\newcommand{\Log}{\operatorname{Log}}

\newcommand{\Ann}{\operatorname{Ann}}

\newcommand{\Dcal}{\mathcal D}

\newcommand{\M}{\overline{\mathcal M}}

\newcommand{\thQ}{\theta_Q}
\newcommand{\thetaQ}{\theta_Q}
\newcommand{\thetay}{\theta_y}
\renewcommand{\dd}{\,\mathrm d}

\newif\ifdraftnotes \draftnotestrue
\definecolor{bluebell}{rgb}{0.69, 0.69, 0.92}

\definecolor{codegreen}{rgb}{0,0.6,0}
\definecolor{codegray}{rgb}{0.5,0.5,0.5}
\definecolor{codepurple}{rgb}{0.58,0,0.82}
\definecolor{backcolor}{rgb}{0.95,0.95,0.92}

\lstdefinestyle{pythonstyle}{
    backgroundcolor=\color{backcolor},   
    commentstyle=\color{codegreen},
    keywordstyle=\color{magenta},
    numberstyle=\tiny\color{codegray},
    stringstyle=\color{codepurple},
    basicstyle=\ttfamily\small,
    breakatwhitespace=false,         
    breaklines=true,                 
    captionpos=b,                    
    keepspaces=true,                 
    numbers=left,                    
    numbersep=5pt,                  
    showspaces=false,                
    showstringspaces=false,
    showtabs=false,                  
    tabsize=4
}

\makeatletter
\newif\ifIXlive \IXlivefalse
\IfFileExists{sunset-I-general.aux}{\IXlivetrue}{}
\newcommand{\IXfreeze}[2]{\expandafter\gdef\csname IX@#1\endcsname{#2}}

\newcommand{\IXnum}[1]{%
  \ifIXlive
    \IfRefUndefinedExpandable{PaperI-#1}%
      {\refused{PaperI-#1}\textbf{??}}%
      {\getrefnumber{PaperI-#1}\IX@checkfrozen{#1}}%
  \else
    \@ifundefined{IX@#1}%
      {\textbf{??}\PackageWarning{sunsetII}{No frozen Paper I number for `#1'}}%
      {\csname IX@#1\endcsname}%
  \fi}

\protected\def\IX@checkfrozen#1{%
  \@ifundefined{IX@#1}%
    {\PackageWarning{sunsetII}{Paper I label `#1' has no frozen number;
      regenerate the frozen block}}%
    {\edef\IX@live{\getrefnumber{PaperI-#1}}%
     \edef\IX@frozen{\csname IX@#1\endcsname}%
     \ifx\IX@live\IX@frozen\else
       \PackageWarning{sunsetII}{Paper I number of `#1' is \IX@live\space
         in sunset-I-general.aux but frozen as \IX@frozen;
         regenerate the frozen block}%
     \fi}}

\newcommand{\Iref}[2][]{\cite[#1\IXnum{#2}]{PaperI}}
\newcommand{\Ithmref}[2][]{\cite[Theorem~#1\IXnum{#2}]{PaperI}}
\newcommand{\Icorref}[2][]{\cite[Corollary~#1\IXnum{#2}]{PaperI}}
\newcommand{\Ipropref}[2][]{\cite[Proposition~#1\IXnum{#2}]{PaperI}}
\newcommand{\Ilemref}[2][]{\cite[Lemma~#1\IXnum{#2}]{PaperI}}
\newcommand{\Iremref}[2][]{\cite[Remark~#1\IXnum{#2}]{PaperI}}
\newcommand{\Ieqref}[1]{\cite[Eq.~(\IXnum{#1})]{PaperI}}

\makeatother

\ifIXlive
  \usepackage{refcount}
  \usepackage{xr}
\fi

\IXfreeze{cor:Qradius}{11.5}
\IXfreeze{cor:nfive}{10.3}
\IXfreeze{cor:nfour}{10.2}
\IXfreeze{cor:primaryintegral}{14.10}
\IXfreeze{cor:realcontinuation}{11.2}
\IXfreeze{def:operator}{8.1}
\IXfreeze{eq:Cdef}{5.2}
\IXfreeze{eq:Kdef}{6.8}
\IXfreeze{eq:Schwinger}{9.1}
\IXfreeze{eq:allloop}{10.2}
\IXfreeze{eq:descdef}{5.1}
\IXfreeze{eq:physicalGamma}{9.8}
\IXfreeze{eq:physicalPF}{9.9}
\IXfreeze{ext:endpoint}{19.2}
\IXfreeze{ext:endpoints}{19}
\IXfreeze{ext:sec:primary}{18}
\IXfreeze{ext:thm:primary}{18.1}
\IXfreeze{lem:basechange}{5.3}
\IXfreeze{lem:tree}{5.1}
\IXfreeze{prop:Gamma}{9.6}
\IXfreeze{prop:ambient}{4.2}
\IXfreeze{prop:analyticdomain}{9.1}
\IXfreeze{prop:geometry}{4.1}
\IXfreeze{prop:sTaylor}{9.2}
\IXfreeze{prop:virtualcompare}{6.1}
\IXfreeze{rem:pseudothresholds}{9.4}
\IXfreeze{sec:dictionary}{16}
\IXfreeze{sec:local}{5}
\IXfreeze{sec:mirror}{6}
\IXfreeze{sec:zetaambiguity}{16.2}
\IXfreeze{thm:CYcover}{14.3}
\IXfreeze{thm:IP}{14.11}
\IXfreeze{thm:PF}{8.2}
\IXfreeze{thm:assembly}{10.1}
\IXfreeze{thm:chou}{14.8}
\IXfreeze{thm:extraction}{6.9}
\IXfreeze{thm:factorization}{5.4}
\IXfreeze{thm:fullambient}{15.1}
\IXfreeze{thm:laplace}{9.7}
\IXfreeze{thm:minimalperiod}{15.2}
\IXfreeze{thm:mirror}{6.6}
\IXfreeze{thm:torus}{11.1}
\IXfreeze{thm:virtual}{5.5}

\title[Three- and four-loop sunset integrals]{The three- and four-loop
equal-mass sunset integrals\\ and local Calabi--Yau fourfolds and fivefolds}
\author{Hartmut Maennel}
\address{Google DeepMind, Zurich, Switzerland}
\author{Pierre Vanhove}
\address{Institut de Physique Théorique, Université Paris Saclay,
  CNRS, CEA, F-91191 Gif-sur-Yvette, France\\
  Laboratoire d’Annecy-le-Vieux de Physique Th\'eorique (LAPTh), CNRS and Universit\'e
Savoie Mont-Blanc, Annecy, 74940, France}
\date{\today --- preprint number LAPTH-052/26}

\begin{document}

\begin{abstract}
We determine the equal-mass two-dimensional three- and four-loop sunset
Feynman integrals $I^{(3)}_\circleddash$ and $I^{(4)}_\circleddash$ at large
Euclidean momentum in terms of the genus-zero Gromov--Witten theory of the
local Calabi--Yau fourfold and fivefold $X_L=\Tot(K_{F_L})$, where
$F_L\subset(\PP^1)^{L+1}$, $L=3,4$, is a smooth hypersurface of multidegree
$(1,\ldots,1)$. At two loops the sunset integral is governed by the
prepotential of a local Calabi--Yau threefold. For $L\geq3$ the moduli spaces
of unmarked stable maps have positive virtual dimension and there is no prepotential, so the Feynman integral is governed instead by the one-point
potential $\mathcal F_{\gamma_L}$ with insertion $\gamma_L=H^{L-1}/(L-1)!$, by
its gravitational descendants, and by $\zeta(3)$ boundary terms coming from the
Gamma class. At both loop orders the coefficient of $R^{L-1}$ in the Frobenius
expansion in the mirror coordinate $R$ is $-\theta_Q\mathcal F_{\gamma_L}$ with the Novikov variable $Q=\exp(R)$. At
three loops the whole Feynman integral is an explicit expression in
$\mathcal F_{\gamma_3}$ alone, because genus-zero recursion determines all
descendant series from it; at four loops two diagonal series suffice, and the
$\zeta(3)$ term is itself corrected by $\theta_Q\mathcal F_{\gamma_4}$. At
three loops we also derive the analytic normalization independently, from the
Picard--Fuchs equation and the polar part of a Mellin transform, prove
integrality of the multiple-cover transforms for arbitrary integral insertions
together with an explicit quantum-$K$ formula, and compute the exact period
connection polynomials. Finally we give endpoint values along the Euclidean
axis, bounds on the Taylor radii in the Novikov variable $Q$, a computer-assisted
proof of the four-loop endpoint nonanalyticity, and a comparison of the
normalizations at two, three and four loops.
\end{abstract}
\maketitle

\tableofcontents

\section{Introduction}\label{II:sec:intro}

The equal-mass two-dimensional sunset integrals are among the simplest Feynman
integrals that are not polylogarithmic. For $L\geq2$ the $L$-loop member, with
all $L+1$ masses equal to one, has the Schwinger representation
\begin{equation}\label{II:eq:general-schwinger}
 I^{(L)}_\circleddash(s)=
 \int_{\mathbb R_+^{L+1}}
 e^{-\sum_{i=1}^{L+1} x_i+sW(x)}\,W(x)
 \prod_{i=1}^{L+1}\frac{\dd x_i}{x_i},\qquad
 W(x)=\left(\sum_{i=1}^{L+1} x_i^{-1}\right)^{-1},
\end{equation}
which converges for $\operatorname{Re}s<(L+1)^2$ and continues analytically to
$\C\setminus[(L+1)^2,\infty)$. This paper is about the three- and four-loop
integrals $I^{(3)}_\circleddash$ and $I^{(4)}_\circleddash$. We express them
exactly, at large Euclidean momentum $s=-T\to-\infty$, through the genus-zero
Gromov--Witten invariants of a local Calabi--Yau fourfold and of a local
Calabi--Yau fivefold, and we study the enumerative, arithmetic and analytic consequences of this identification.

\subsection*{From two loops to \texorpdfstring{$L\geq3$}{L>=3} loops}

At two loops the geometry attached to the sunset integral is a family of
elliptic curves. Bloch, Kerr and Vanhove~\cite{BKV} related the two-loop
integral to the genus-zero prepotential of the local Calabi--Yau threefold
$\Tot(K_{F_2})$ over the del Pezzo surface $F_2$ of degree six, through a
mirror-symmetry and Legendre-transform construction. At $L$ loops the
corresponding geometry is the local Calabi--Yau $(L+1)$-fold
\[
 X_L=\Tot(K_{F_L}),\quad F_L\subset Y_L=(\PP^1)^{L+1}\ \text{a smooth hypersurface of
 multidegree }(1,\ldots,1).
\]
Here $F_L$ is Fano with $K_{F_L}=\mathcal O_{F_L}(-H)$, where $H=p_1+\cdots+p_{L+1}$ and
$p_i$ is the pullback of the point class of the $i$th factor. The Laurent
polynomial $$\phi(x)=\left(\sum_{i=1}^{L+1}x_i\right)\left(\sum_{i=1}^{L+1}x_i^{-1}\right)$$ is the Landau--Ginzburg
superpotential mirror to $F_L$, and the exponent in~\eqref{II:eq:general-schwinger}
factors as $-W(x)(\phi(x)-s)$. The sunset integral is therefore a Laplace
transform of the oscillatory integral of this mirror
(see~\Ipropref{prop:analyticdomain}), and it is Iritani's Gamma-class
identity for that oscillatory integral~\cite{Iritani} that brings the
Gromov--Witten theory of $X_L$ into play.

\medskip

For $L\geq3$ the moduli spaces $\overline{\mathcal M}_{0,\ell}(X_L,\beta)$ have positive
virtual dimension $\vdim = L-2+\ell$, so unmarked invariants require insertions
(e.g., $\gamma_L = H^{L-1}/(L-1)!$) to define numerical observables. At two loops, by contrast, the prepotential suffices; this is the setting of~\cite{BKV}.

The first numerical observable is the one-point invariant with the insertion
\[
 \gamma_L=\frac{H^{L-1}}{(L-1)!}=\sum_{|S|=L-1}\,\prod_{i\in S}p_i
 \in H^{2(L-1)}(F_L,\Z),
\]
whose complex codimension $L-1$ equals
$\vdim\overline{\mathcal M}_{0,1}(X_L,\beta)$. Its generating series is the
one-point potential
\begin{equation}\label{II:eq:FgammaL}
 \mathcal F_{\gamma_L}(Q)=\sum_{m\ge1}Q^m\sum_{H\cdot\beta=m}
 \int_{[\M_{0,1}(X_L,\beta)]^{\vir}}\ev^*(\varpi^*\gamma_L),
\end{equation}
where $\varpi\colon X_L\to F_L$ is the projection and $Q$ is the Novikov variable.
For $L\geq3$ the class $\gamma_L$ is not a divisor, and no divisor equation
reduces $\mathcal F_{\gamma_L}$ to unmarked invariants. Accordingly its
multiple-cover transform has the one-marking weight $r^{-2}$ instead of the
weight $r^{-3}$ of the unmarked prepotential used at two loops;
Section~\ref{app:loops} shows that this is a change of observable, not of the
genus-zero cover law.

\medskip
The coefficients of the lower powers of the mirror
coordinate involve the diagonal one-point descendant series
\begin{equation}\label{II:eq:Dcal}
 \Dcal_a^\Delta(Q)=\sum_{m\ge1}Q^m\sum_{H\cdot\beta=m}
 \int_{[\M_{0,1}(X_L,\beta)]^{\vir}}\psi^a\,\ev^*(\varpi^*\delta_a),\qquad
 \delta_a=\frac{H^{L-1-a}}{(L-1-a)!},
\end{equation}
for $0\leq a\leq L-1$, so that $\Dcal_0^\Delta=\mathcal F_{\gamma_L}$. At two
loops the divisor and dilaton equations reduce both such series to the
prepotential. For $L\geq3$ they do not, and one of the questions answered below
is to what extent the descendants are nevertheless determined by
$\mathcal F_{\gamma_L}$.

\medskip
The boundary behaviour at large momentum is governed
by the Gamma class
$\widehat\Gamma_{X_L}|_{F_L}=\exp\bigl(\sum_{r\geq1}\frac{2\zeta(2r+1)}{2r+1}H^{2r+1}\bigr)$
(see~\Ipropref{prop:Gamma}). On the surface $F_2$ one has $H^3=0$, so no
zeta value enters and the two-loop boundary constants vanish
(Section~\ref{app:L2boundary}). For $L=3,4$ the term $\frac23\zeta(3)H^3$
survives. It produces the boundary term $-16\zeta(3)=-\frac23\zeta(3)\int_{F_3}H^3$
at three loops and, at four loops, the term
$-80\zeta(3)=-\frac23\zeta(3)\int_{F_4}H^4$ multiplying a logarithmic solution that
is itself corrected by $\mathcal F_{\gamma_4}$; the value $\zeta(5)$ first
enters at five loops.

Two further differences are technical. For $L=3$ the hypersurface $F_3$ has $b_3(F_3)=2$ (Remark~\ref{l3:rem:odd-cohomology}) and is therefore not
toric, so the toric mirror theorem used at two loops does not apply to it
directly; the local theory is computed instead on the ambient space $Y_L$ with an Euler and an inverse Euler twist~\cite{PaperI}. For
$L=4$ the middle cohomology of $F_4$ contains an $18$-dimensional nonambient
summand, invisible to the differential operator; the one-point evaluation
classes are nevertheless ambient (see~\Ithmref{thm:fullambient}).

Conversely, two hypotheses used throughout for $L\geq3$ fail at $L=2$, which is
why the two-loop case is not a specialization of the results below. 

First, the weak Lefschetz isomorphism $H_2(F_L,\Z)\cong H_2(Y_L,\Z)$ requires $L\geq3$. 
For the del Pezzo surface $F_2$ the curve lattice has rank four and
the ambient one rank three, so an ambient degree does not determine a curve
class. 

Second, at the threshold $s=(L+1)^2$ the integrand of the simplex
representation of~\eqref{II:eq:general-schwinger} behaves like
$\|u-u_*\|^{-2}$ in real dimension $L$. This is integrable for $L\geq3$, and for
$L=3,4$ the Feynman integral is continuous, but not $C^1$, at threshold
(see~\Ipropref{prop:sTaylor}); for $L=2$ the integral diverges
logarithmically at $s=9$. Section~\ref{app:loops} treats the two-loop case
separately and compares it with three and four loops.

\subsection*{The Feynman integral in terms of \texorpdfstring{$\mathcal F_{\gamma_L}$}{F}}

Put $y=-1/s$, so that $y=1/T>0$ on the Euclidean axis, and write
$\theta_y=y\,d/dy$ and $\theta_Q=Q\,d/dQ$. Let $A(y)=1+O(y)$ be the holomorphic
period, $R$ the logarithmic mirror coordinate, normalized by $\theta_yR=A$ and
$R=\log y+O(y)$, and $Q=e^R$; put $\Pi_L=(-1)^LyA$. With the uniform assembly
of~\Ithmref{thm:assembly} as input, the three- and
four-loop Feynman integrals are expressed through $\mathcal F_{\gamma_L}$ as
follows.
We use the notation of~\cite{PaperI} as follows. For
\(\gamma_L=H^{L-1}/(L-1)!\), \(\mathcal F_{\gamma_L}\) denotes the genus-zero one-point potential associated with the local Calabi--Yau \(X_L=\operatorname{Tot}(K_{F_L})\). The diagonal descendant series are denoted by \(\mathcal D_a^\Delta\), with

$$
\mathcal D_0^\Delta=\mathcal F_{\gamma_L}.
$$

We write \(\theta_Q=Q\,d/dQ\), \(Q=e^R\), and \(\mathcal T\) for the diagonal
two-point primary series. All other notation imported from~\cite{PaperI} is
recalled at its first use below.

\medskip

At three loops (Theorem~\ref{II:thm:L3} and Corollary~\ref{II:cor:L3F}), with
$\mathcal F=\mathcal F_{\gamma_3}$ and
$\theta_Q^{-1}\sum_{m\ge1}c_mQ^m=\sum_{m\ge1}(c_m/m)Q^m$,
\begin{multline}\label{II:eq:intro-L3}
 I^{(3)}_\circleddash(s)=-yA\Bigl[4R^3-(\theta_Q\mathcal F)R^2
 +\frac{(\theta_Q\mathcal F)^2}{12}\,R\\
 +2\theta_Q^{-1}\mathcal F
 +\frac1{12}\theta_Q^{-1}\bigl((\theta_Q\mathcal F)^2\bigr)
 -\frac16\,\mathcal F\,\theta_Q\mathcal F-16\zeta(3)\Bigr].
\end{multline}
The coefficient of $R^2$ is $-\theta_Q\mathcal F_{\gamma_3}$, the fourfold
analogue of the two-loop theorem of~\cite{BKV}. It is proved in
Section~\ref{app:analytic3} independently of the oscillatory Gamma identity on
which~\cite{PaperI} relies, and Section~\ref{ext:reconstruction} gives a
second, differential proof of the complete formula~\eqref{II:eq:intro-L3},
which avoids the signed analytic assembly of~\cite{PaperI}. The two lower
coefficients, which in the general assembly involve the descendant series
$\Dcal_1^\Delta$ and $\Dcal_2^\Delta$, are determined by $\mathcal F_{\gamma_3}$
through the genus-zero topological recursion (Theorem~\ref{ext:thm:Bdiag}).
Thus at three loops the whole Feynman integral is determined by a single
generating series of genus-zero invariants, as at two loops, with
$\mathcal F_{\gamma_3}$ in the role of the prepotential. Its multiple-cover
expansion $\mathcal F_{\gamma_3}=\sum_{m\ge1}n_m\Li_2(Q^m)$ has integer
coefficients $n_m=24,-72,264,-1488,9696,\ldots$
(Theorem~\ref{l3:thm:bps-integrality}), so that
$\exp(\theta_Q\mathcal F_{\gamma_3})=\prod_{m\ge1}(1-Q^m)^{-mn_m}$ is an
integral Euler product.

At four loops (Theorem~\ref{II:thm:L4} and Corollary~\ref{II:cor:L4F}),
\begin{equation}\label{II:eq:intro-L4}
I^{(4)}_{\circleddash}(s)=yA\Bigl[5R^4-(\theta_Q\mathcal F_{\gamma_4})R^3
+h_2R^2+h_1R+h_0
-80\zeta(3)\Bigl(R-\frac{\theta_Q\mathcal F_{\gamma_4}}{20}\Bigr)\Bigr],
\end{equation}
with $h_2=-3\mathcal F_{\gamma_4}-\theta_Q\Dcal_1^\Delta$, and with $h_1$ and
$h_0$ determined by $\mathcal F_{\gamma_4}$ and $\Dcal_1^\Delta$ through the
recursion of Theorem~\ref{ext:thm:fiveclosure}. Thus the four-loop Feynman
integral is determined by two diagonal genus-zero series, $\mathcal
F_{\gamma_4}$ and $\Dcal_1^\Delta$, equivalently by $\mathcal F_{\gamma_4}$ and
the diagonal two-point primary series $\mathcal T$. We do not know whether $\mathcal F_{\gamma_4}$ alone suffices.

At both loop orders the leading term is $(L+1)R^L$ and the first instanton
coefficient is $-\theta_Q\mathcal F_{\gamma_L}$:
\[
 I^{(L)}_\circleddash(s)=\Pi_L\Bigl[(L+1)R^L-(\theta_Q\mathcal F_{\gamma_L})R^{L-1}
 +\sum_{j=0}^{L-2}c_{L,j}(Q)R^j\Bigr],
\]
with $c_{L,j}$ convergent power series in $Q$ whose constant terms are the
$\zeta(3)$ terms above. The same shape holds at two loops for the coarsened
invariants of Section~\ref{app:loops}, where $\gamma_2=H$ and, by the divisor
equation, $\theta_Q\mathcal F_H$ is the second logarithmic derivative of the
prepotential. In this sense $\mathcal F_{\gamma_L}$ takes over the role of the prepotential for $L\geq3$.

\bigskip
In this work we derive specific results at $L= 3$ and $L= 4$ independently of
the result of~\cite{PaperI}.\footnote{The all-loop identity and its proof are
given in the companion paper~\cite{PaperI}. We use the following results from
that paper without reproving them: the zero-section factorization and virtual
reduction; the twisted mirror theorem and the extraction of the descendant
series; the operator $P_L$ and its all-degree certificate; the signed Gamma
formula, based on Iritani's oscillatory identity~\cite{Iritani}; and the uniform
assembly theorem~\Ithmref{thm:assembly}. These results provide the general
all-loop framework. The three- and four-loop formulas stated in
Theorems~\ref{II:thm:L3} and~\ref{II:thm:L4} coincide with the corresponding
specializations of Paper~I (see~\Icorref{cor:nfour} and~\Icorref{cor:nfive}), but are
derived here independently of the all-loop identity. In particular, the
loop-specific derivations and checks below do not rely on the proofs of the
imported results or on the all-loop identity itself.} 

\medskip
Section~\ref{II:sec:statements} states the three- and four-loop identities and
their expressions through $\mathcal F_{\gamma_L}$, with a dictionary between the
two loop orders and numerical checks. 

\medskip
Section~\ref{app:analytic3} gives an independent analytic derivation at three
loops. It determines the three Frobenius boundary constants of
$I^{(3)}_\circleddash$ from the inhomogeneous Picard--Fuchs equation and the
polar part of a Mellin transform. We show that two vanish and the third is $-16\zeta(3)$
(Proposition~\ref{l3:prop:leading}). A coefficient bridge built on the
quantum-Serre identity and a logarithmic B-model recurrence gives
$d_2(m)=-mN_m(\gamma_3)$ for the coefficients of $Q^mR^2$, independently of all
three constants (Theorems~\ref{l3:thm:coefficient-bridge}
and~\ref{l3:thm:main}); a direct count of lines confirms the normalization
$N_1(\gamma_3)=24$ (Proposition~\ref{l3:prop:m1-check}). The section also
identifies the three descendant potentials, $\Dcal_0^\Delta=-12K_2$,
$\Dcal_1^\Delta=-24K_3$ and $\Dcal_2^\Delta=-24K_4$
(Theorem~\ref{l3:thm:full-descendant-assembly}), and the exterior domains on
which the series converge.

\medskip
In section~\ref{app:arith3}, we derive an integrality and quantum $K$-theory at three loops.
We prove that for every integral class
$\eta\in H^4(F_3,\Z)$ the $r^{-2}$ transforms of the one-point invariants are
integers (Theorem~\ref{l3:thm:bps-integrality}). The proof uses the
projective-bundle compactification $\mathbb P(\mathcal O_{F_3}\oplus K_{F_3})$, which is
semipositive, and the integrality theorem of Ionel and Parker in the algebraic
form of Chou~\cite{IP,Chou}. The section also gives an explicit quantum-$K$
formula for the transforms of $\mathcal F_{\gamma_3}$
(Theorem~\ref{l3:thm:qk-strengthening}), the resulting supercongruences, and
exact values through degree $19$. The parity argument for the symplectic
semipositivity of this compactification is specific to complex dimension four;
primary integrality for every $L\geq3$ is proved in~\cite{PaperI}
(see~\Icorref{cor:primaryintegral}) with a compact Calabi--Yau cover instead.

\medskip
Section~\ref{app:loops} compares the curve
lattices, the normalizations and the multiple-cover conventions at $L=2,3,4$,
separates the change of observable from the restrictions specific to two loops,
derives the vanishing of the two-loop boundary constants by the Mellin method,
records the four-loop specialization with explicit series, and relates the
local invariants to logarithmic invariants of the pair $(F_L,D)$, $D\in|H|$,
through the local--relative correspondence of van~Garrel, Graber and
Ruddat~\cite{vGGR}.

\medskip
In section~\ref{ext:reconstruction}, we give a reconstruction from primary
invariants. We apply the genus-zero topological recursion on a compact model of
$X_L$. At three loops all one-point descendants are determined by the primary
evaluation classes, and on the diagonal by $\mathcal F_{\gamma_3}$ alone
(Theorem~\ref{ext:thm:Bdiag}); this gives~\eqref{II:eq:intro-L3}, together with
the differential proof mentioned above, which rests on the symmetric-square
structure of the Picard--Fuchs operator. At four loops the primary evaluation
classes in full multidegree determine every one-point descendant
(Theorem~\ref{ext:thm:inverse}); after restriction to the one-parameter
diagonal, the reconstruction closes on two series
(Theorem~\ref{ext:thm:fiveclosure}).

\medskip
Corollary~\ref{II:cor:L4F} shows that two diagonal inputs suffice at four loops. This does not contradict the full multidegree reconstruction theorem, which uses the complete collection of primary evaluation classes; after restriction to the one-parameter diagonal, that information is no longer contained in the single series $\mathcal F_{\gamma_4}$. Whether $\mathcal F_{\gamma_4}$ alone determines $I^{(4)}_{\circleddash}$, as $\mathcal F_{\gamma_3}$ determines $I^{(3)}_{\circleddash}$, remains open.

\medskip
Section~\ref{ext:lowend} gives the endpoint values of the mirror coordinate and of all descendant potentials
along the Euclidean axis. At three loops the endpoint value of $R$ is
$R_*=-7\zeta(3)/\pi^2$, minus twice the Mahler measure of $1+z_1+z_2+z_3$; at
four loops the endpoint values reduce to two constants. That section also gives the four-loop threshold expansion. Section~\ref{ext:connection3} computes the
exact three-loop period connection polynomials (Theorem~\ref{ext:periods}). As
a consequence $K_2,K_3,K_4$ have a one-sided $C^1$ extension, but an unbounded
second derivative, at the endpoint $Q_*=e^{R_*}$, and their $Q$-Taylor radii
lie between $1/12$ and $Q_*$. Sections~\ref{ext:cert:setup}--\ref{ext:cert:run}
prove the analogous four-loop statement, with radii between $1/20$ and
$Q_\infty=0.3366117386\ldots$, by a computer-assisted argument in ball
arithmetic; the complete code is in Appendix~\ref{ext:code}.

\subsection*{The differential operators}

The Feynman integrals satisfy the inhomogeneous Picard--Fuchs equations
\begin{equation}\label{II:eq:PF}
 P_L\Bigl(\frac{I^{(L)}_\circleddash(-1/y)}{y}\Bigr)=(-1)^L(L+1)!
\end{equation}
of~\Ieqref{eq:physicalPF}, where $P_L$ is the order-$L$ operator produced by
the algorithm of Definition~\Iref{def:operator}. At the two loop orders treated
here,
\[
 P_3=\theta_y^3
                   +64y^2(\theta_y+1)^3
 +y(20\theta_y^3+30\theta_y^2+18\theta_y+4),
\]
and 
\begin{align}
P_4&=\theta_y^4+y^3Q_3(\theta_y)+y^2Q_2(\theta_y)+yQ_1(\theta_y),\\
 Q_1&=35\theta_y^4+70\theta_y^3+63\theta_y^2+28\theta_y+5,\cr
 Q_2&=259\theta_y^4+1036\theta_y^3+1580\theta_y^2+1088\theta_y+285,\cr
 \nonumber Q_3&=225\theta_y^4+1350\theta_y^3+2925\theta_y^2+2700\theta_y+900.
\end{align}
After the substitution $y=-t$ these are the equal-mass sunset Picard--Fuchs
operators recorded in the Feynman-integral literature; they have been derived
in the mathematical~\cite{ZhouSurvey} and physics
literature~\cite{Vanhove:2014wqa,Vanhove:2018mto,LairezVanhove,KlemmEtAl,BFKNS}.
Their holomorphic solution $A(y)=\sum_{D\geq0}(-1)^Dc_Dy^D$, with
$c_D=\sum_{|\mathbf d|=D}\binom D{\mathbf d}^{\!2}$, is the alternating generating
series of the $(L+1)$-variable generalized Apéry numbers; for $L=3$ these are
$1,4,28,256,2716,\ldots$. What the companion paper~\cite{PaperI} supplies is the uniform
construction for all $L$ together with an all-degree certificate, which is what
allows the operator to be used without fitting it to initial coefficients.
That $P_L$ is irreducible, hence the minimal operator annihilating the period,
follows from the irreducibility of the symmetric power of the Bessel connection
and two Fourier transforms with checked extensions at zero. 
We refer to~\Ipropref{thm:minimalperiod} of the companion paper.

\section*{Acknowledgements and use of AI}    
This work started as an application of the AI co-mathematician tool
\cite{CoMath} developed by Google DeepMind, in which a human works with AI agents to develop a theory.
Later, the output of the AI co-mathematician tool was also reviewed by different
AI tools (first Gemini, later also Claude and ChatGPT) and the reviews were used to improve the theorems, sometimes in further collaboration with the AI co-mathematician.\\
The final paper is human written, using parts of the AI generated proofs, and
also proofread by AI.\\
The work of HM was done at Google DeepMind.\\
The work of PV was funded by the Agence Nationale de la Recherche (ANR) under
grant Observables (ANR-24-CE31-7996).
\section{The three- and four-loop identities}\label{II:sec:statements}

We state the three- and four-loop identities, which are the specializations of
the all-loop assembly~\Ieqref{eq:allloop}, together with their expressions
through the one-point potentials $\mathcal F_{\gamma_L}$. The notation is that
of the introduction, namely $y=-1/s$, $R$ is the local mirror coordinate,
$A=\theta_yR$, $Q=e^R$, $\Pi_L=(-1)^LyA$, and $\Dcal_a^\Delta$ is the diagonal
one-point descendant series~\eqref{II:eq:Dcal}, so that
$\Dcal_0^\Delta=\mathcal F_{\gamma_L}$. The $K_k$ are the normalized
coefficients of the ambient hypergeometric series~\Ieqref{eq:Kdef}; for $L=3$
they are recalled in~\eqref{l3:eq:full-ambient-K}.

\begin{theorem}[Three loops; {\rm \Icorref{cor:nfour}}]
\label{II:thm:L3}
For $L=3$ put $\Pi_3=-yA$ and $\delta_0=H^2/2$. Then
\[
 I^{(3)}_\circleddash(s)=\Pi_3\bigl[4R^3+f_2R^2+f_1R+f_0-16\zeta(3)\bigr],
\]
with
\[
 f_2=-\thQ\Dcal_0^\Delta,\qquad
 f_1=-2\Dcal_0^\Delta-\thQ\Dcal_1^\Delta,\qquad
 f_0=-\Dcal_1^\Delta-\thQ\Dcal_2^\Delta,
\]
and $\Dcal_0^\Delta=-12K_2$, $\Dcal_1^\Delta=-24K_3$, $\Dcal_2^\Delta=-24K_4$.
In particular the coefficient of $R^2$ is $-\thQ\mathcal F_{\gamma_3}$, where
$\mathcal F_{\gamma_3}=\Dcal_0^\Delta$ is the one-point
potential~\eqref{II:eq:FgammaL} of the local Calabi--Yau fourfold $X_3$ with insertion
$\gamma_3=H^2/2$.
\end{theorem}

\begin{theorem}[Four loops; {\rm \Icorref{cor:nfive}}]
\label{II:thm:L4}
For $L=4$ put $\Pi_4=yA$. Then
\[
 I^{(4)}_\circleddash(s)=\Pi_4\left[5R^4+\sum_{j=0}^3h_jR^j
      -80\zeta(3)\left(R-\frac1{20}\thQ\Dcal_0^\Delta\right)\right],
\]
where
\[
\begin{array}{ll}
 h_3=-\thQ\Dcal_0^\Delta,&
 h_2=-3\Dcal_0^\Delta-\thQ\Dcal_1^\Delta,\\
 h_1=-2\Dcal_1^\Delta-\thQ\Dcal_2^\Delta,&
 h_0=-\Dcal_2^\Delta-\thQ\Dcal_3^\Delta,
\end{array}
\]
and $(\Dcal_0^\Delta,\Dcal_1^\Delta,\Dcal_2^\Delta,\Dcal_3^\Delta)
=(-20K_2,-60K_3,-120K_4,-120K_5)$.
In particular the coefficient of $R^3$ is $h_3=-\thQ\mathcal F_{\gamma_4}$, where
$\mathcal F_{\gamma_4}=\Dcal_0^\Delta$ is the one-point
potential~\eqref{II:eq:FgammaL} of the local Calabi--Yau fivefold $X_4$ with insertion
$\gamma_4=H^3/6$.
\end{theorem}

The descendant series in Theorems~\ref{II:thm:L3} and~\ref{II:thm:L4} are not
independent of the primary series $\mathcal F_{\gamma_L}$, because the genus-zero
topological recursion relations, applied in Section~\ref{ext:reconstruction},
express them through primary invariants. On the diagonal this gives the
following two statements.

\begin{corollary}[Three loops in terms of $\mathcal F_{\gamma_3}$]\label{II:cor:L3F}
Write $\mathcal F=\mathcal F_{\gamma_3}$, and let $\thQ^{-1}$ act on power series
without constant term by $\thQ^{-1}\sum_{m\ge1}c_mQ^m=\sum_{m\ge1}(c_m/m)Q^m$.
Then
\[
 \Dcal_2^\Delta=\frac{\mathcal F^2}{12},\qquad
 \Dcal_1^\Delta=-2\thQ^{-1}\mathcal F
 -\frac1{12}\thQ^{-1}\bigl((\thQ\mathcal F)^2\bigr).
\]
Consequently $f_1=(\thQ\mathcal F)^2/12$, and the three-loop Feynman integral is
the expression~\eqref{II:eq:intro-L3} in $\mathcal F$ alone. Equivalently, with
$u=\thQ\mathcal F/12$,
\[
 I^{(3)}_\circleddash(s)=-yA\bigl[4(R-u)^3+4u^3-\Dcal_1^\Delta-2u\mathcal F
 -16\zeta(3)\bigr].
\]
\end{corollary}
\begin{proof}
The identities for $\Dcal_2^\Delta$ and $\Dcal_1^\Delta$ are
Theorem~\ref{ext:thm:Bdiag}, the second one integrated with zero constant term.
Substituting them into the coefficients $f_1$ and $f_0$ of
Theorem~\ref{II:thm:L3} gives~\eqref{II:eq:intro-L3}, and expanding $(R-u)^3$
gives the last form.
\end{proof}
\begin{corollary}[Four loops in terms of two diagonal series]\label{II:cor:L4F}
In Theorem~\ref{II:thm:L4} one has
$h_3=-\thQ\mathcal F_{\gamma_4}$ and
$h_2=-3\mathcal F_{\gamma_4}-\thQ\Dcal_1^\Delta$, while
$\Dcal_2^\Delta$ and $\Dcal_3^\Delta$, and hence $h_1$ and $h_0$, are the
unique power series without constant term that satisfy~\eqref{ext:eq:diagfive2}
and~\eqref{ext:eq:diagfive3}, where $\Dcal_a$ stands for $\Dcal_a^\Delta$ and
$\Dcal_0^\Delta=\mathcal F_{\gamma_4}$. Thus the four-loop Feynman integral is
determined by the pair
$
\left(\mathcal F_{\gamma_4},\Dcal_1^\Delta\right),
$
or equivalently by
$
\left(\mathcal F_{\gamma_4},\mathcal T\right),
$
where $\mathcal T$ is the diagonal two-point primary series of
Theorem~\ref{ext:thm:fiveclosure}.
\end{corollary}
\begin{proof}
The formulas for $h_3$ and $h_2$ are those of Theorem~\ref{II:thm:L4}, and the
statement on $\Dcal_2^\Delta$ and $\Dcal_3^\Delta$ is
Theorem~\ref{ext:thm:fiveclosure}. The first identity of that theorem,
$(1-\thQ^2\mathcal F_{\gamma_4}/20)\mathcal T
=2\thQ^2\Dcal_1^\Delta+12\thQ\mathcal F_{\gamma_4}$, determines either of
$\mathcal T$ and $\Dcal_1^\Delta$ from the other, because
$1-\thQ^2\mathcal F_{\gamma_4}/20$ is an invertible power series and $\thQ^2$ is
bijective on power series without constant term.
Thus, for fixed \(\mathcal F_{\gamma_4}\), the two pairs \((\mathcal F_{\gamma_4},\Dcal_1^\Delta)\) and \((\mathcal F_{\gamma_4},\mathcal T)\) contain the same information.
\end{proof}

This should be contrasted with Corollary~\ref{II:cor:L3F}, where the single
primary series $\mathcal F_{\gamma_3}$ determines all descendant data entering
the Feynman integral. Four loops is therefore the first case in which our
diagonal reconstruction does not reduce to a single primary potential, and we
do not know whether $\mathcal F_{\gamma_4}$ alone determines
$I^{(4)}_{\circleddash}$. In full multidegree, by contrast, the
primary evaluation classes alone determine every one-point descendant
(see Theorem~\ref{ext:thm:inverse}).

\begin{remark}[Why $L=3,4$ is not the general pattern]
At $L=5$ a $\zeta(5)$ term first appears, in addition to the $\zeta(3)$ term
multiplying a descendant-corrected logarithmic solution, as at four loops, and at $L=6$ the product $\zeta(3)^2$ appears. 
Formula~\Ieqref{eq:allloop} fixes both the zeta polynomial 
and the descendant
correction multiplying each of its terms.
\end{remark}

The dictionary between the $L=3$ and $L=4$ loop orders is as follows.

\begin{center}
\begin{tabular}{lcc}
\toprule
Quantity & $L=3$ & $L=4$\\\midrule
$n$ & $4$ & $5$\\
$\dim F_L$, $\dim X_L$ & $3,4$ & $4,5$\\
One-point local virtual dimension & $2$ & $3$\\
$\int_{F_L} H^L$ & $24$ & $120$\\
Primary insertion $\gamma_L$ & $H^2/2$ & $H^3/6$\\
Descendants needed & $a=0,1,2$ & $a=0,1,2,3$\\
Period prefactor & $-yA$ & $yA$\\
Degree-one diagonal primary & $24$ & $60$\\
Degree-one diagonal first descendant & $-48$ & $-120$\\
Degree-two diagonal primary & $-66$ & $-225$\\
Degree-two transformed primary & $-72$ & $-240$\\
Possible nonambient even cohomology & none & middle degree $4$\\
\bottomrule
\end{tabular}
\end{center}

\begin{remark}[Numerical verification]
The Bessel-moment representation of the Feynman integral was evaluated
numerically and compared with the right-hand sides of~\Ieqref{eq:physicalGamma}
and~\Ieqref{eq:allloop}. For $L=3$ at $T=50$ and $T=100$, and
for $L=4$ at $T=100$, agreement was observed to 35-40 digits:
\begin{center}
\begin{tabular}{llll}
\toprule
$L$ & $T$ & $I^{(L)}_\circleddash(-T)$ & agreement\\\midrule
$3$ & $50$  & $5.17515590864132484\ldots$ & $35$ digits\\
$3$ & $100$ & $4.08241320270405960\ldots$ & $40$ digits\\
$4$ & $100$ & $26.9629554067580390\ldots$ & $40$ digits\\
\bottomrule
\end{tabular}
\end{center}
This tests the descendant series, the mirror map and the $\zeta(3)$ term together. 
The diagonal primaries
and their transforms were computed to $m=8$ at $L=3$ and to $m=6$ at $L=4$:
\[
 n_m(L{=}3)=24,\,-72,\,264,\,-1488,\,9696,\,-69240,\,528072,\,-4238880,
\]
\[
 n_m(L{=}4)=60,\,-240,\,1400,\,-12480,\,131220,\,-1529520\,.
\]
They are all integers as stated in~\Icorref{cor:primaryintegral},
whereas the $N_m$ are rational numbers.
\end{remark}

\section{An independent analytic derivation at three loops}\label{app:analytic3}
The proof of~\Ithmref{thm:assembly} runs through the
oscillatory Gamma-class identity of Iritani's paper~\cite[Section 2, equation (10)]{Iritani},
which is its analytic input. This section independently determines the
$L=3$ analytic boundary normalization from the Picard--Fuchs equation and
the polar part of a Mellin transform. This determination does not use that
identity; the first proof of the full descendant assembly later invokes
the companion paper~\cite{PaperI}. Section~\ref{ext:reconstruction} supplies an alternative proof
using the primary bridge, recursion, and the symmetric-square equation. The two routes fix the same
three Frobenius boundary constants and give the same expansion. The second
route also isolates what does and does not depend on those constants. The
coefficient identity $d_2(m)=-mN_m(\gamma)$ is independent of all three,
whereas the explicit boundary term $-16\zeta(3)\pi(s)$ is not.

Throughout this section $L=3$, so that the number of factors is $n=L+1=4$.

\subsection{Three-loop conventions}\label{app:l3notation}

The geometry is $Y_3=(\PP^1)^4$, a smooth $F_3\in|\mathcal O_{Y_3}(1,1,1,1)|$ with
$j:F_3\hookrightarrow Y_3$, and $X_3=\Tot(K_{F_3})\xrightarrow{\varpi}F_3$ with zero
section $\iota$. By~\Ipropref{prop:geometry}, $K_{F_3}=\mathcal O_{F_3}(-H)$,
$\int_{F_3}H^3=24$ and $H_2(F_3,\Z)\cong\Z^4$; we write $|\beta|=H\cdot\beta$. Put
$\gamma:=H^2/2=\sum_{i<j}p_ip_j\in H^4(F_3,\Z)$, which is the insertion $\gamma_3=\delta_0$ of the introduction,
and
\begin{align}\label{l3:eq:potential-intro}
\mathcal F_\gamma(Q)&=\sum_{m\ge1}N_m(\gamma)Q^m,\\
 N_m(\gamma)&=\sum_{|\beta|=m}N_\beta(\gamma),\cr
 \nonumber N_\beta(\gamma)&=\int_{[\M_{0,1}(X_3,\beta)]^\vir}\ev^*(\varpi^*\gamma),\cr
\end{align}
so $\mathcal F_\gamma=\mathcal F_{\gamma_3}=\Dcal_0^\Delta$; in this section and the
next we keep the shorter notation $\gamma$ and $\mathcal F_\gamma$. In the dictionary between this section and~\cite{PaperI},
\begin{align*}
 y&=-1/s,\qquad g(y)=S(y)=-G(y,y,y,y),\cr R_{\mathrm{body}}(y)&=T(y),\qquad
 \Pi_3(s)=-yA(y)=\pi(s),
\end{align*}
and the $K_2(Q)$ used here is the $K_2(y(Q))$ of~\Ieqref{eq:Kdef}. The
symbol $k$ denotes the number of markings in the statements quoted below;
the number of factors is fixed at four. The Schwinger representation is~\eqref{II:eq:general-schwinger} with $L=3$, which is~\Ieqref{eq:Schwinger} with $n=4$.

The next four statements are the three-loop forms of results proved
in~\cite{PaperI}. They are recorded here, with their notation, because the arguments of
this section are written in it; their proofs are the specializations
indicated, and are not repeated.

\begin{proposition}[Zero-section factorization]\label{l3:prop:zero-section}
Let $S$ be an arbitrary $\C$-scheme, let $\varrho_S:\mathcal C_S\to S$ be a flat, proper,
finitely presented family of connected nodal curves of arithmetic genus zero,
and let $g:\mathcal C_S\to X_3$ be a morphism whose restriction to every
geometric fiber has class $\beta\ne0$.  Then $g$ factors uniquely through the
zero section $\iota:F_3\hookrightarrow X_3$.  The factorization preserves stability
and commutes with arbitrary base change.  Consequently
\[
 \overline{\mathcal M}_{0,k}(F_3,\beta)\xrightarrow{\ \sim\ }
 \overline{\mathcal M}_{0,k}(X_3,\beta)
\]
is a canonical isomorphism of stacks, and the latter is a proper
Deligne--Mumford stack.
\end{proposition}
\begin{proof}
This is~\Ithmref{thm:factorization} with $n=4$ and $K_{F_3}=\mathcal O_{F_3}(-H)$; that statement is functorial over arbitrary bases and preserves stability.
\end{proof}

For $k$ markings write $\mathcal M_{F_3}=\overline{\mathcal M}_{0,k}(F_3,\beta)$,
let $\varrho$ and $f$ be its universal curve and map, and set
$\mathcal V=R^1\varrho_*f^*K_{F_3}$. This is a vector bundle of rank
$m-1$, where $m=H\cdot\beta>0$: on every geometric fiber $f^*K_{F_3}$ has
degree $-m<0$ and nonpositive degree on every component, so
$H^0(C,f^*K_{F_3})=0$ and $h^1(C,f^*K_{F_3})=m-1$ by~\Ilemref{lem:tree},
and cohomology and base change, in the form of~\Ilemref{lem:basechange},
makes $R^1\varrho_*f^*K_{F_3}$ locally free of rank $m-1$, compatibly with
arbitrary base change.

\begin{proposition}[Virtual reduction]\label{l3:prop:virtual-reduction}
Under the identification of Proposition~\ref{l3:prop:zero-section},
\begin{align}
 [\overline{\mathcal M}_{0,k}(X_3,\beta)]^{\rm vir}
 &=c_{\rm top}(\mathcal V)\cap[\mathcal M_{F_3}]^{\rm vir},\label{l3:eq:vfc}\\
 [\overline{\mathcal M}_{0,k}(X_3,\beta)]^{\rm vir}_{\lambda}
 &=e_\lambda(\mathcal V)\cap[\mathcal M_{F_3}]^{\rm vir},\label{l3:eq:vfc-eq}
\end{align}
where in~\eqref{l3:eq:vfc-eq} the fiberwise $\C^*$-action on $X_3$ gives
$\mathcal V$ weight one.
\end{proposition}
\begin{proof}
This is~\Ithmref{thm:virtual}, whose proof identifies the obstruction-theory morphisms themselves and includes the equivariant formula.
\end{proof}

For $k=1$ define the \emph{proper evaluation cycle} on $F_3$,
\begin{equation}\label{l3:eq:Cbeta}
 C_\beta^\lambda=(\ev_{F_3})_*\bigl(
 e_\lambda(\mathcal V)\cap[\overline{\mathcal M}_{0,1}(F_3,\beta)]^{\rm vir}\bigr)
 \ \in\ H^*(F_3)[\lambda].
\end{equation}
Its coefficients are cohomology classes on the projective variety $F_3$;
no localization of the equivariant parameter is involved.  Expanding
$e_\lambda(\mathcal V)=\sum_{k=0}^{m-1}\lambda^{m-1-k}c_k(\mathcal V)$, the cycle
$c_k(\mathcal V)\cap[\mathcal M_{F_3}]^{\rm vir}$ has dimension $(m+1)-k$, so its
pushforward to the threefold $F_3$ vanishes for $k\le m-3$, is a multiple of the
fundamental class for $k=m-2$, and is a $2$-dimensional cycle class for
$k=m-1$.  Hence
\begin{equation}\label{l3:eq:Cbeta-structure}
 C^\lambda_\beta=C^0_\beta+\lambda\,c_\beta\,\mathbf 1,\qquad
 C^0_\beta=(\ev_{F_3})_*\bigl(c_{\rm top}(\mathcal V)\cap[\mathcal M_{F_3}]^{\rm vir}\bigr)\in H^2(F_3,\Q),
 \quad c_\beta\in\Q,
\end{equation}
homogeneous of total equivariant degree one; its polynomial dependence on
$\lambda$ is what makes the non-equivariant limit $C^0_\beta$ meaningful.
Since $\ev_{X_3}\circ\bar\iota=\iota\circ\ev_{F_3}$ for the induced isomorphism
$\bar\iota$ of stable-map stacks, and $\ev_{F_3}$ is proper, the equivariant
self-intersection formula $\iota^*\iota_*=e_\lambda(K_{F_3})\cup(-)$ gives
\begin{equation}\label{l3:eq:self-intersection}
 \iota^*(\ev_{X_3})_*[\overline{\mathcal M}_{0,1}(X_3,\beta)]^{\rm vir}_\lambda
 =e_\lambda(K_{F_3})\,C_\beta^\lambda=(\lambda-H)\,C_\beta^\lambda .
\end{equation}

Let $j:F_3\hookrightarrow Y_3$ be the inclusion.

\begin{lemma}[State space]\label{l3:lem:state-space}
$j^*:H^{2k}(Y_3,\Q)\to H^{2k}(F_3,\Q)$ is surjective for every $k$, and
$\ker j^*=\Ann(H)\cap H^{2k}(Y_3,\Q)$, where
$\Ann(H)=\{a\in H^*(Y_3,\Q):aH=0\}$.  Hence
\[
 H^{\rm even}(F_3,\Q)\cong H^*(Y_3,\Q)/\Ann(H),
\]
and this presentation is stable under the permutation action of $S_4$ on the
factors of $Y_3$.  The bilinear form
\[
 (a,b)_{X_3,\lambda}=\int_{F_3}\frac{ab}{\lambda-H},\qquad
 \frac1{\lambda-H}=\frac1\lambda+\frac{H}{\lambda^2}+\frac{H^2}{\lambda^3}+\frac{H^3}{\lambda^4},
\]
is nondegenerate on $H^{\rm even}(F_3,\C(\lambda))$.
\end{lemma}
\begin{proof}
This is~\Ipropref{prop:ambient} with $n=4$; since $\dim F_3=3$ is odd, all even cohomology is ambient, and the pairing and kernel statements specialize as displayed.
\end{proof}

\begin{remark}\label{l3:rem:odd-cohomology}
$H^{\rm odd}(F_3,\Q)\neq0$: from $c(T_{F_3})=1+H+H^2+2e_3$ one gets
$\int_{F_3}c_3(T_{F_3})=8$, so $b_3(F_3)=2+2b_2(F_3)-8=2$.  Thus $H^*(F_3)$ is \emph{not}
entirely ambient, and Lemma~\ref{l3:lem:state-space} concerns the even part only.
All classes occurring below --- the evaluation cycles~\eqref{l3:eq:Cbeta}, the
insertion $\gamma$, and every coefficient of the $I$-function --- are even.
Odd insertions contribute nothing to one-point genus-zero invariants for
degree reasons (the virtual class and $\psi$ are even), so nothing is lost by
working on the even part.  Also note that a permutation of the factors of $Y_3$
need not preserve the particular hypersurface $F_3$; the $S_4$-symmetry used
below is the ring-level symmetry of the presentation
$H^*(Y_3)/\Ann(H)$, which is all that the arguments require.
\end{remark}

Let $\C^*$ act on $X_3=\Tot(K_{F_3})$ by scaling the fibers, with equivariant
parameter $\lambda$.  Since $X_3$ is not proper, its equivariant Gromov--Witten
theory is defined through the following two conventions, both standard.

\begin{definition}[Local correlators and the local $J$-function]\label{l3:def:local}
\begin{enumerate}[label=(\roman*)]
\item The state space is $\mathcal H_{X_3}:=H^{\rm even}(F_3,\C(\lambda))$, identified
with the even part of $H^*_{\C^*}(X_3)\otimes\C(\lambda)$ via
$\varpi^*=(\iota^*)^{-1}$, with the pairing $(\,,)_{X_3,\lambda}$ of
Lemma~\ref{l3:lem:state-space}.  This is equivariant integration over $X_3$ by
localization to the fixed locus $F_3$, whose normal bundle $K_{F_3}$ has equivariant
Euler class $\lambda-H$.
\item For $\beta\ne0$, $k\ge1$, $\phi_i\in\mathcal H_{X_3}$ and $a_i\ge0$,
\[
 \bigl\langle\phi_1\psi^{a_1},\dots,\phi_k\psi^{a_k}\bigr\rangle^{X_3,\lambda}_{0,k,\beta}
 :=\int_{[\overline{\mathcal M}_{0,k}(X_3,\beta)]^{\rm vir}_\lambda}
 \prod_{i=1}^k\ev_i^*(\varpi^*\phi_i)\,\psi_i^{a_i}\ \in\ \C(\lambda),
\]
an integral over the proper stack of Proposition~\ref{l3:prop:zero-section};
by~\eqref{l3:eq:vfc-eq} it equals
\[
 \int_{[\mathcal M_{F_3}]^{\rm vir}}e_\lambda(\mathcal V)\prod_i\ev_i^*\phi_i\psi_i^{a_i}.
\]
It lies in $\Q[\lambda]$ for $\phi_i\in H^{\rm even}(F_3,\Q)$.
\item The small local $J$-function is
\begin{equation}\label{l3:eq:geometric-J}
 J_{X_3}^\lambda(\mathbf Q,z)=e^{\tau/z}\left[z\mathbf1+
 \sum_{\beta\ne0}\mathbf Q^\beta\sum_\alpha\phi^\alpha
 \Bigl\langle\frac{\phi_\alpha}{z-\psi}\Bigr\rangle^{X_3,\lambda}_{0,1,\beta}
 \right],
 \qquad \tau=\sum_{i=1}^4p_i\log Q_i,
\end{equation}
where $\{\phi_\alpha\}$ is a basis of $\mathcal H_{X_3}$ and $\{\phi^\alpha\}$ the
dual basis for $(\,,)_{X_3,\lambda}$; the expression is independent of the basis.
We write $\mathcal J^\lambda_{X_3}=J^\lambda_{X_3}/z$, which has leading term
$\mathbf 1$.
\end{enumerate}
\end{definition}

The right-hand side of~\eqref{l3:eq:geometric-J} is the usual small $J$-function
in the divisor coordinate $\tau\in H^2$, written with the Novikov variable and
$e^{\tau\cdot\beta}$ identified, as the string and divisor equations permit.
We use throughout the convention in which $J(\tau,z)=z\mathbf1+\tau+O(z^{-1})$
is the point of Givental's cone with that leading behaviour; the convention of
\cite{CoatesGivental,CCITtwisted} places $J(\tau,-z)$ on the cone instead, and the
two are related by the substitution $z\mapsto-z$ throughout, which we do not
write out.

Let $$G(\mathbf y)=\sum_{\mathbf d\ne0}\mathbf y^{\mathbf d}
(-1)^{D-1}\frac{D!\,(D-1)!}{\prod_{i=1}^4(d_i!)^2}$$ and define the regularized
$I$-function
\begin{align}\label{l3:eq:I-mu}
 I^{\mu,\lambda}(\mathbf y,z)&=
 e^{\sum_{i=1}^4 p_i\log y_i/z}
 \sum_{\mathbf d\in\mathbb N^4}\mathbf y^{\mathbf d}
 \frac{\prod_{r=1}^{D}(\mu+H+rz)\prod_{r=0}^{D-1}(\lambda-H-rz)}
 {\prod_{i=1}^4\prod_{r=1}^{d_i}(p_i+rz)^2},\\
 \nonumber D&=\sum_{i=1}^4 d_i,
\end{align}
with values in $H^*(Y_3,\C)[\mu,\lambda]$; each coefficient is a polynomial in
$\mu$ and $\lambda$.  Its specialization
\begin{equation}\label{l3:eq:I-function}
 \mathcal I_{X_3}^\lambda(\mathbf y,z):=j^*I^{0,\lambda}(\mathbf y,z)=
 e^{\sum_{i=1}^4 p_i\log y_i/z}
 \sum_{\mathbf d\in\mathbb N^4}\mathbf y^{\mathbf d}
 \frac{\prod_{r=1}^{D}(H+rz)\prod_{r=0}^{D-1}(\lambda-H-rz)}
 {\prod_{i=1}^4\prod_{r=1}^{d_i}(p_i+rz)^2}
\end{equation}
takes values in $\mathcal H_{X_3}$; the further specialization $\lambda=0$ is
defined, even though the pairing $(\,,)_{X_3,\lambda}$ has $(\lambda-H)^{-1}$
in it.
\begin{proposition}[Mirror map]\label{l3:prop:mirror-map}
In $\mathcal H_{X_3}((z^{-1}))[[\mathbf Q]]$, with $\mathbf Q$ and $\mathbf y$ related by
$Q_i=y_ie^{-G(\mathbf y)}$,
\begin{equation}\label{l3:eq:ItoJ}
 J_{X_3}^\lambda(\mathbf Q,z)=z\,e^{-\lambda G(\mathbf y)/z}\,
 \mathcal I_{X_3}^\lambda(\mathbf y,z).
\end{equation}
\end{proposition}
\begin{proof}
This is~\Ithmref{thm:mirror} together with~\Ipropref{prop:virtualcompare} and~\Ipropref{prop:ambient}. The
normalized local comparison holds against all ambient insertions; since $L=3$ is
odd these span $H^{\mathrm{even}}(F_3)$, so nondegeneracy identifies the whole
even one-point series. Here $g=-G$ and $Q_i=y_ie^{-G}$, and removing the unit
input $\lambda G$ gives the factor $e^{-\lambda G/z}$ in~\eqref{l3:eq:ItoJ}.
\end{proof}

\subsection{Convergence, the differential equation, and Frobenius at infinity}\label{l3:sec:analytic}

Put
\[
 V(x)=\sum_{i=1}^4x_i,\qquad
 W(x)=\Big(\sum_{i=1}^4x_i^{-1}\Big)^{-1},
\]
so that~\eqref{II:eq:general-schwinger} at $L=3$ becomes
\begin{equation}\label{l3:eq:schwinger3}
 I^{(3)}_\circleddash(s)=\int_{\mathbb R_+^4}e^{-V(x)+sW(x)}W(x)
 \prod_{i=1}^4\frac{\dd x_i}{x_i}.
\end{equation}

\begin{proposition}[Convergence]\label{l3:prop:convergence}
The integral~\eqref{l3:eq:schwinger3} is absolutely convergent for
$\operatorname{Re}s\le16$ and holomorphic for $\operatorname{Re}s<16$.
Absolute convergence fails for $\operatorname{Re}s>16$.
\end{proposition}
\begin{proof}
Since $|e^{sW}|=e^{aW}$ with $a=\operatorname{Re}s$, absolute convergence
depends only on $a$.  Write $x_i=ru_i$, where $r>0$, $u_i>0$, and $\sum u_i=1$;
the Jacobian of $(r,u_1,u_2,u_3)\mapsto x$ is $r^3$, so
$\prod_i\dd x_i/x_i=(\dd r/r)\,\dd^3u/(u_1u_2u_3u_4)$ on the open simplex
$\Delta_3$.  Both $V$ and $W$ are homogeneous of degree one:
$V(x)=r$, $W(x)=rw(u)$ with $w(u)=(\sum u_i^{-1})^{-1}$.  The
arithmetic--harmonic mean inequality gives $w\le1/16$ on $\Delta_3$, with
equality only at $u_*=(1/4,\ldots,1/4)$.  Hence for $a\le16$ one has
$1-aw(u)>0$ for every $u\in\Delta_3\setminus\{u_*\}$, and on that set the radial
integral is $\int_0^\infty e^{-r(1-aw)}rw\,\dd r/r=w/(1-aw)$.  By Tonelli
(the integrand is nonnegative, and $\{u_*\}$ is a null set), absolute
convergence is equivalent to finiteness of
\[
 \int_{\Delta_3}\frac{w(u)}{u_1u_2u_3u_4[1-aw(u)]}\,\dd^3u .
\]
For $a<16$ the integrand is bounded by the one for $a=16$, so it suffices to
treat $a=16$.  Put $e_3(u)=\sum_j\prod_{k\ne j}u_k$ and $e_4(u)=u_1u_2u_3u_4$,
so that $w=e_4/e_3$ and the integrand equals $1/P(u)$ with
$P=e_3(1-16w)=e_3-16e_4$.  On the simplex,
\begin{equation}\label{l3:eq:P-identity}
 P(u)=\sum_{i<j}u_ku_l(u_i-u_j)^2,\qquad \{k,l\}=\{1,2,3,4\}\setminus\{i,j\};
\end{equation}
indeed, expanding the right side gives $\sum_iu_i^2e_2(u_{\hat\imath})-12e_4$
(where $e_2(u_{\hat\imath})$ is the second elementary symmetric function of the
three variables other than $u_i$), and $\sum_iu_i=1$ gives
$e_3=\sum_iu_ie_3(u)=\sum_iu_i^2e_2(u_{\hat\imath})+4e_4$.
Also $\sum_{i<j}(u_i-u_j)^2=4\sum_iu_i^2-1=4\lVert u-u_*\rVert^2$ on the simplex.
In the neighborhood $U=\{\max_iu_i-\min_iu_i<1/8\}$ of $u_*$ all $u_i$ are
bounded below by $1/8$ (since $\max_iu_i\ge1/4$); hence the weights $u_ku_l$ in
\eqref{l3:eq:P-identity} lie in $[1/64,1]$ and $P\asymp\lVert u-u_*\rVert^2$ on
$U$.  Thus $1/P$ is locally integrable near $u_*$ in dimension three.  On the
ordered chamber $u_1\ge u_2\ge u_3\ge u_4$, outside $U$ one has
$u_1-u_4\ge1/8$ and, keeping only the term $(i,j)=(1,4)$ of
\eqref{l3:eq:P-identity}, $P\ge u_2u_3(u_1-u_4)^2\ge u_2u_3/64$; enlarging the
domain of integration to $0\le u_4\le u_3\le u_2\le1$ gives
\[
 \int\frac{\dd u_2\,\dd u_3\,\dd u_4}{P}
 \le64\int_0^1\dd u_2\int_0^{u_2}\dd u_3\int_0^{u_3}
 \frac{\dd u_4}{u_2u_3}=64.
\]
The $S_4$-symmetry of the integrand completes the estimate.  Holomorphy on
$\operatorname{Re}s<16$ follows by dominated convergence on compact subsets.
For $a>16$, the factor $1-aw(u)$ is negative on an open neighborhood of
$u_*$, and the radial integral diverges there, so absolute convergence fails.
\end{proof}

The period of the three-loop sunset is the $n=4$ case of the generalized Apéry numbers~\cite{Verrill,BBBG,Vanhove:2014wqa,Vanhove:2018mto}
\begin{equation}\label{l3:eq:domb}
 A_n =\sum_{d_1+d_2+d_3+d_4=n}\binom{n}{d_1,d_2,d_3,d_4}^{\!2}\,.
\end{equation}
They satisfy the recursion relation~\cite{Verrill,BBBG} 
\begin{equation}\label{l3:eq:domb-rec}
 n^3A_n=(20n^3-30n^2+18n-4)A_{n-1}-64(n-1)^3A_{n-2}\qquad(n\ge1,\ A_{-1}=0),
\end{equation}
and $A_n\le16^n$.

\begin{convention}[Branches]\label{l3:conv:branches}
On the Euclidean base ray $s=-t$, $t>16$, choose
\[
 y=-s^{-1}=t^{-1}>0,\qquad \Log s=\log t+i\pi,\qquad \ell:=\Log s-i\pi=\log t,
\]
and set
\begin{align}
 \pi(s)&=\frac1s\sum_{n\ge0}\frac{A_n}{s^n},\label{l3:eq:pi}\\
 R(s)&=i\pi-\Log s+\sigma(s),\qquad \sigma(s)=\sum_{n\ge1}\frac{A_n}{ns^n},
 \qquad Q=e^{R(s)}.\label{l3:eq:R}
\end{align}
Both series converge for $|s|>16$.  Then $R=-\ell+\sigma$,
$R(-t)=\log y+\sum_{n\ge1}(-1)^nA_ny^n/n$, and
$\dd R/\dd s=-\pi(s)$.  The function $Q(s)=-s^{-1}e^{\sigma(s)}$ is
single-valued and analytic on $|s|>16$, with $Q(s)=-s^{-1}+O(s^{-2})$; on the
Euclidean ray $Q=t^{-1}e^{\sigma(-t)}>0$.  Since $s^{-1}\mapsto Q$ has
derivative $-1$ at $s^{-1}=0$, it is a local biholomorphism there; hence every
function meromorphic at $s=\infty$ is a convergent Laurent series in $Q$, and
conversely.
\end{convention}

The Picard--Fuchs operator is~\cite{BKV3,Vanhove:2014wqa}
\begin{equation}\label{l3:eq:L3}
 \mathcal L_3=s^2(s-4)(s-16)\partial_s^3
 +6s(s^2-15s+32)\partial_s^2
 +(7s^2-68s+64)\partial_s+(s-4).
\end{equation}
In terms of $\theta_y=y\partial_y$, where $y=-1/s$, and of the operator
\[
 P_3=\theta_y^3+y\,p(\theta_y)+64y^2(\theta_y+1)^3,\qquad
 p(x)=20x^3+30x^2+18x+4,
\]
of the introduction, a direct substitution gives
\begin{equation}\label{l3:eq:L3-theta}
 \mathcal L_3u=P_3\Bigl(\frac uy\Bigr).
\end{equation}
Thus the equation $\mathcal L_3I^{(3)}_\circleddash=-24$ of Lemma~\ref{l3:lem:ode}
below is the three-loop case of~\eqref{II:eq:PF}. Since $\theta_y=-s\partial_s$,
one has $\theta_ys^k=-k\,s^k$ for every exponent $k$, and~\eqref{l3:eq:L3-theta}
gives
\begin{equation}\label{l3:eq:L3-monomial}
 \mathcal L_3\,s^k=(k+1)^3s^{k+1}-p(k)\,s^k+64k^3s^{k-1}.
\end{equation}
Its singular points are $s=0,4,16,\infty$, all regular.  By
\eqref{l3:eq:L3-monomial}, the indicial polynomial at infinity is $(\rho+1)^3$, and
$\mathcal L_3\pi=0$, therefore the coefficient recurrence of $\mathcal L_3$ at infinity for
$s^{-1}\sum c_ns^{-n}$ is exactly~\eqref{l3:eq:domb-rec}.  Thus $\pi$ is the unique
solution of $\mathcal L_3u=0$ of the form $s^{-1}+O(s^{-2})$ without logarithms:
at exponent $-1$ the coefficient of the $n$-th unknown in that recurrence is
$(-n)^3\ne0$ for $n\ge1$, so a log-free series solution is determined by its
leading coefficient.

\begin{lemma}[The inhomogeneous Picard--Fuchs equation]\label{l3:lem:ode}
On $\operatorname{Re}s<16$ the integral~\eqref{l3:eq:schwinger3} satisfies
$\mathcal L_3I^{(3)}_{\circleddash}(s)=-24$.
\end{lemma}

That the equal-mass sunset integral satisfies an inhomogeneous Picard--Fuchs
equation whose homogeneous part is the operator governing the associated
Calabi--Yau periods is standard in the Feynman-integral literature~\cite{BKV3,KlemmEtAl,Vanhove2026,Vanhove:2018mto},
 where the statement is
proved at three loops and, in the last two, for the sunset family at every
loop order.
We therefore quote the lemma and record only what fixes the two conventions on
which the rest of this section depends.

\begin{remark}[Normalization of the source term]\label{l3:rem:source}
The right-hand side is a constant, and its value $-24=-\int_{F_3}H^3$ is the one
compatible with the normalization of the parametric Schwinger representation in~\eqref{II:eq:general-schwinger} and
the $\theta_y$-form~\eqref{l3:eq:L3-theta}; it is the three-loop case of the
source $H^Ly^H$ certified in~\Ithmref{thm:PF}. Two checks confirm it independently of that certificate. Expanding $e^{sW}$ under the
integral sign for $|s|<16$, which Proposition~\ref{l3:prop:convergence}
justifies, writes the Taylor coefficients of $I^{(3)}_\circleddash$ as Bessel moments,
\[
 I^{(3)}_\circleddash(s)=\sum_{k\ge0}a_ks^k,\qquad
 a_k=\frac{8\,c_{4,2k+1}}{4^k\,k!^2},\qquad
 c_{4,w}=\int_0^\infty x^{w}K_0(x)^4\,\dd x ,
\]
and $\mathcal L_3I^{(3)}_\circleddash=-24$ becomes the moment recurrence of
\cite{BBBG,BorweinSalvy}, explicitly
\[
 64(k+1)^3a_{k+1}-(20k^3+30k^2+18k+4)a_k+k^3a_{k-1}=0
 \quad(k\ge1),
\]
whereas its constant coefficient is $64a_1-4a_0=-24$.
Equivalently, with $w=2k+1$, the moment recurrence is
\[
 64(w+1)c_{4,w+2}-4w(5w^2+3)c_{4,w}
 +(w-1)^5c_{4,w-2}=0\quad(w=3,5,\ldots).
\]
The endpoint condition is the initial relation
$16c_{4,3}-4c_{4,1}=-3$, which follows from the pole of $c_{4,2a-1}$ at $a=0$.
With $c_{4,1}=\tfrac78\zeta(3)$ and $c_{4,3}=\tfrac7{32}\zeta(3)-\tfrac3{16}$
from~\cite[\S4.3]{BBBG}, that relation holds; the individual values are not
needed anywhere below. Since $\mathcal L_3I^{(3)}_\circleddash+24$ is holomorphic on the connected
domain $\operatorname{Re}s<16$ and vanishes on $|s|<16$, the equation holds on
the whole domain.

The expansion also shows that $I^{(3)}_\circleddash$ is holomorphic on all of $|s|<16$. In
particular $s=4$ is a singular point of~\eqref{l3:eq:L3} at which the
particular solution $I^{(3)}_\circleddash$ is holomorphic; $4$ is not an apparent singularity of the homogeneous equation, whose exponents there are
$0,\tfrac12,1$.
\end{remark}

\subsection*{Frobenius structure at infinity}
The construction below is the standard Frobenius method at a point of maximal
unipotent monodromy, applied to $\mathcal L_3$ at $s=\infty$, where the
indicial polynomial is $(\rho+1)^3$. For the sunset family it is carried out in
\cite{KlemmEtAl} and~\cite{Vanhove2026}; we record only the normalizations used
below, since the later arguments refer to them by name.

Define the Frobenius ansatz
\[
 u(s,\rho)=s^\rho\sum_{n\ge0}c_n(\rho)s^{-n},\qquad c_0(\rho)=1\,.
\]
Substituting into  $\mathcal L_3$ and equating coefficients yields the recurrence
for $n\ge1$
\begin{equation}\label{l3:eq:frob-rec}
 (\rho-n+1)^3c_n(\rho)=p(\rho-n+1)\,c_{n-1}(\rho)-64(\rho-n+2)^3c_{n-2}(\rho)
 \qquad(n\ge1,\ c_{-1}=0).
\end{equation}
For $n=0$ this simplifies to
\begin{equation}\label{l3:eq:frob-defect}
 \mathcal L_3\,u(s,\rho)=(\rho+1)^3s^{\rho+1}.
\end{equation}
The denominators in~\eqref{l3:eq:frob-rec} do not vanish for $\rho$ near $-1$,
and the nearest finite singularity in the coordinate $s^{-1}$ is $1/16$, so
$\sum_nc_n(\rho)s^{-n}$ converges for $|s|>16$, locally uniformly in $\rho$
near $-1$; it is therefore holomorphic in $\rho$ there and may be
differentiated termwise. Comparing~\eqref{l3:eq:frob-rec} at $\rho=-1$
with~\eqref{l3:eq:domb-rec} gives $c_n(-1)=A_n$, and all derivatives
$c_n^{(j)}(-1)$ are rational numbers.

On the Euclidean branch put $\tilde u(s,\rho)=e^{-i\pi(\rho+1)}u(s,\rho)
=-e^{\rho\ell}\sum_nc_n(\rho)s^{-n}$, so that $\tilde u(s,-1)=\pi(s)$.  Define
\begin{equation}\label{l3:eq:frob-basis}
 \Phi_j=\frac1{j!}\,\partial_\rho^j\tilde u\Big|_{\rho=-1}\ (j=0,1,2),
 \qquad
 I_{\rm part}=-4\,\partial_\rho^3\tilde u\Big|_{\rho=-1}.
\end{equation}
By~\eqref{l3:eq:frob-defect}, $\mathcal L_3\Phi_j=0$ for $j\le2$ and
$\mathcal L_3I_{\rm part}=-24$.  Write
\[
 r_j(s)=\frac{\sum_nc_n^{(j)}(-1)s^{-n}}{\sum_nA_ns^{-n}}.
\]
Then $r_0=1$ and
$r_j(\infty)=0$ for $j\ge1$.  Restricting to a sufficiently small germ at
infinity so that $s\pi(s)$ is nonzero, one has
\begin{equation}\label{l3:eq:frob-explicit}
 \Phi_j=\pi\sum_{i=0}^{j}\binom ji\frac{\ell^{\,j-i}}{j!}\,r_i,\qquad
 I_{\rm part}=-4\pi\bigl(\ell^3+3r_1\ell^2+3r_2\ell+r_3\bigr).
\end{equation}

\begin{lemma}[Frobenius form of the integral]\label{l3:lem:frob-form}
There are unique convergent power series $\hat f_0,\hat f_1,\hat f_2\in\C\{Q\}$
such that, on the Euclidean germ $s=-t$ with $t$ sufficiently large,
\begin{equation}\label{l3:eq:frob-form-general}
 I^{(3)}_\circleddash(s)=\pi(s)\bigl(4R^3+\hat f_2(Q)R^2+\hat f_1(Q)R+\hat f_0(Q)\bigr).
\end{equation}
Moreover $I^{(3)}_\circleddash=I_{\rm part}+\alpha_2\Phi_2+\alpha_1\Phi_1+\alpha_0\Phi_0$ for unique
constants $\alpha_j\in\C$, with
\[
 \hat f_2(0)=\alpha_2/2,\qquad \hat f_1(0)=-\alpha_1,\qquad \hat f_0(0)=\alpha_0,
\]
and the coefficients of $\hat f_r-\hat f_r(0)$ lie in $\Q+\Q\alpha_1+\Q\alpha_2$.
In particular, if $\hat f_2(0)=\hat f_1(0)=0$ then
$\hat f_r-\hat f_r(0)\in Q\,\Q\{Q\}$ for $r=0,1,2$.  (The hats are removed in
Theorem~\ref{l3:thm:vanhove}, where the constant term is displayed separately.)
\end{lemma}
\begin{proof}
By Lemma~\ref{l3:lem:ode}, $I^{(3)}_\circleddash-I_{\rm part}$ solves $\mathcal L_3u=0$ on the simply
connected region $\{\operatorname{Re}s<16,\ |s|>16\}$, where $\Phi_0,\Phi_1,\Phi_2$
form a basis of solutions (Frobenius theory at the regular singular point
$\infty$ with indicial polynomial $(\rho+1)^3$); this gives the $\alpha_j$.
Substituting $\ell=\sigma-R$ into~\eqref{l3:eq:frob-explicit} expresses each of
$\Phi_j/\pi$ and $I_{\rm part}/\pi$ as a polynomial in $R$ whose coefficients
are convergent power series in $s^{-1}$ with rational coefficients and, for
the non-leading coefficients, zero constant term:
\begin{align*}
 I_{\rm part}/\pi&=4R^3-12(\sigma+r_1)R^2+12(\sigma^2+2r_1\sigma+r_2)R
  -4(\sigma^3+3r_1\sigma^2+3r_2\sigma+r_3),\\
 \Phi_2/\pi&=\tfrac12R^2-(\sigma+r_1)R+\tfrac12(\sigma^2+2r_1\sigma+r_2),\cr
 \Phi_1/\pi&=-R+(\sigma+r_1),\qquad \Phi_0/\pi=1 .
\end{align*}
By Convention~\ref{l3:conv:branches}, $s^{-1}=-Q+O(Q^2)$ is a convergent power
series in $Q$ with rational coefficients and zero constant term, so every
coefficient above is a convergent power series in $Q$; collecting powers of $R$
gives~\eqref{l3:eq:frob-form-general} and the formulas for $\hat f_r(0)$ and for the
field of coefficients.  For proving the uniqueness of the $\hat f_r$, we use that
 if $\sum_{j\le3}g_j(Q)R^j=0$ on the germ with $g_j$ convergent Laurent series, then,
since the local monodromy around $s=\infty$ fixes $Q$ and sends
$R\mapsto R-2\pi i$ for the loop traversed counterclockwise in $s$, hence
$R\mapsto R-2\pi ik$ for its iterates $k\in\Z$ (both sides being analytic
continuations along the same loop), one gets
$\sum_jg_j(Q)(R-2\pi ik)^j=0$ for $k=0,\ldots,3$, and a Vandermonde argument
gives $g_j=0$.  Equivalently, $R$ is transcendental over the field of
convergent Laurent series in $Q$.
\end{proof}

The three constants $\alpha_j$ are connection coefficients between $s=0$ and
$s=\infty$, and cannot be obtained from the differential equation.  They are
determined by the leading asymptotics of $I^{(3)}_\circleddash(-t)$ as $t\to\infty$, which we
compute from a Mellin transform.

\begin{proposition}[Leading asymptotics]\label{l3:prop:leading}
As $t\to+\infty$,
\begin{equation}\label{l3:eq:leading}
 I^{(3)}_\circleddash(-t)=\frac{4\log^3t+16\zeta(3)}{t}+O\!\left(\frac{\log^3t}{t^2}\right).
\end{equation}
Consequently, in Lemma~\ref{l3:lem:frob-form}, $\hat f_2(0)=\hat f_1(0)=0$ and
$\hat f_0(0)=-16\zeta(3)$; that is, $\alpha_2=\alpha_1=0$ and
$\alpha_0=-16\zeta(3)$.
\end{proposition}
\begin{proof}
\emph{Step 1: a priori form.}  By Lemma~\ref{l3:lem:frob-form}, since
$\pi(-t)=-t^{-1}(1+O(t^{-1}))$, $R(-t)=-\log t+O(t^{-1})$ and
$Q=O(t^{-1})$, there are constants $\kappa_0,\ldots,\kappa_3$ with
\begin{equation}\label{l3:eq:apriori}
 I^{(3)}_\circleddash(-t)=\frac1t\sum_{j=0}^3\kappa_j\log^jt+O\!\left(\frac{\log^3t}{t^2}\right),
 \qquad
 \kappa_3=4,\ \kappa_2=-\hat f_2(0),\ \kappa_1=\hat f_1(0),\ \kappa_0=-\hat f_0(0).
\end{equation}

\emph{Step 2: Mellin transform.}  For $\operatorname{Re}\nu>0$ put
$M(\nu)=\int_{\mathbb R_+^4}e^{-V}W^\nu\prod_i\dd x_i/x_i$.  As in the proof of
Lemma~\ref{l3:lem:ode},
\begin{equation}\label{l3:eq:M-bessel}
 M(\nu)=\frac{8\cdot4^{1-\nu}}{\Gamma(\nu)}\,\mathcal K(2\nu),\qquad
 \mathcal K(w)=\int_0^\infty x^{w-1}K_0(x)^4\,\dd x\quad(\operatorname{Re}w>0).
\end{equation}
For $0<\operatorname{Re}z<1$, Fubini and $\int_0^\infty t^{z-1}e^{-tW}\dd t=
\Gamma(z)W^{-z}$ give
\begin{equation}\label{l3:eq:mellin-I}
 \int_0^\infty t^{z-1}I^{(3)}_\circleddash(-t)\,\dd t=\Gamma(z)M(1-z).
\end{equation}

\emph{Step 3: two continuations.}  Fix once and for all a $t_0>16$ large enough
that~\eqref{l3:eq:apriori} is valid for $t\ge t_0$; such a $t_0$ exists because
\eqref{l3:eq:apriori} was derived from the convergent expansion of
Lemma~\ref{l3:lem:frob-form}, which holds on a germ at infinity.  Split the left
side of~\eqref{l3:eq:mellin-I} as $\int_0^{t_0}+\int_{t_0}^\infty$.  The first
integral is holomorphic for $\operatorname{Re}z>0$, since $I^{(3)}_\circleddash$ is continuous on
the compact interval $[-t_0,0]$ by Proposition~\ref{l3:prop:convergence} and
$\int_0^{t_0}t^{\operatorname{Re}z-1}\dd t<\infty$ there.  By
\eqref{l3:eq:apriori}, the second equals
\begin{multline*}
 \sum_{j=0}^3\kappa_j\int_{t_0}^\infty t^{z-2}\log^jt\,\dd t
 +\int_{t_0}^\infty t^{z-1}\Bigl(I^{(3)}_\circleddash(-t)-\frac1t\sum_j\kappa_j\log^jt\Bigr)\dd t
 \cr=\sum_{j=0}^3\frac{\kappa_j\,j!}{(1-z)^{j+1}}+h(z),
\end{multline*}
with $h$ holomorphic on $\operatorname{Re}z<2$.  Hence the left side of
\eqref{l3:eq:mellin-I} continues meromorphically to $0<\operatorname{Re}z<2$ with a
single pole at $z=1$, and the polar part there is $\sum_j\kappa_jj!\,\nu^{-j-1}$
in the variable $\nu=1-z$.  On the other hand, the right side continues via
\eqref{l3:eq:M-bessel}. We have near $x=0$,
$K_0(x)=-(\log(x/2)+\gamma_E)+O(x^2\log x)$, so
$K_0(x)^4=(\log(x/2)+\gamma_E)^4+O(x^2\log^4x)$ and therefore
\[
\begin{aligned}
 \mathcal K(w)&=\int_0^1x^{w-1}\bigl(\log(x/2)+\gamma_E\bigr)^4\dd x+H_0(w)\\
 &=\int_0^\infty e^{-wu}(u+c)^4\,\dd u+H_0(w),
 \qquad \text{with}\qquad c=\log2-\gamma_E,
\end{aligned}
\]
where $H_0$ is holomorphic on $\operatorname{Re}w>-2$, using $x=e^{-u}$ in the second line.  Thus $\mathcal K$ is meromorphic on $\operatorname{Re}w>-2$
with a single pole at $w=0$ and polar part
$\mathcal K_{\rm pol}(w)=\sum_{j=0}^4\binom4j\,j!\,c^{4-j}w^{-j-1}$.  By the
identity theorem, the two continuations of~\eqref{l3:eq:mellin-I} agree, so
\begin{equation}\label{l3:eq:polar-match}
 \sum_{j=0}^3\frac{\kappa_j\,j!}{\nu^{j+1}}
 =\text{polar part at }\nu=0\text{ of }\;
 \Gamma(1-\nu)\,\frac{8\cdot4^{1-\nu}}{\Gamma(\nu)}\,\mathcal K(2\nu).
\end{equation}

\emph{Step 4: the polar part.}  Write $1/\Gamma(\nu)=\nu/\Gamma(1+\nu)$ and use
$\log\Gamma(1+x)=-\gamma_Ex+\sum_{k\ge2}(-1)^k\zeta(k)x^k/k$ to get
\[
 \frac{\Gamma(1-\nu)}{\Gamma(1+\nu)}
 =\exp\Bigl(2\gamma_E\nu+\tfrac23\zeta(3)\nu^3+O(\nu^5)\Bigr),
 \qquad\text{and}\qquad
 4^{1-\nu}=4\,e^{-2\nu\log2}.
\]
Only the polar part of $\mathcal K(2\nu)$ contributes, and
$\mathcal K_{\rm pol}(2\nu)=\int_0^\infty e^{-2\nu u}(u+c)^4\dd u$.  Hence the
right side of~\eqref{l3:eq:polar-match} is the polar part of
\[
 32\,\nu\,\exp\bigl(b\nu+\tfrac23\zeta(3)\nu^3\bigr)
 \Bigl(\frac{24}{(2\nu)^5}+\frac{24c}{(2\nu)^4}+\frac{12c^2}{(2\nu)^3}
 +\frac{4c^3}{(2\nu)^2}+\frac{c^4}{2\nu}\Bigr),\] 
with
\[b=2\gamma_E-2\log2=-2c .
\]
Expanding $$e^{b\nu+\frac23\zeta(3)\nu^3}=1+b\nu+\frac{b^2}2\nu^2
+\left(\frac{b^3}6+\frac23\zeta(3)\right)\nu^3+O(\nu^4)$$ and collecting, the polar part
equals
\[
 32\Bigl[\frac{3}{4\nu^4}+\frac{3(b+2c)}{4\nu^3}+\frac{3(b+2c)^2}{8\nu^2}
 +\frac{(b+2c)^3+4\zeta(3)}{8\nu}\Bigr]
 =\frac{24}{\nu^4}+\frac{16\zeta(3)}{\nu},
\]
because $b+2c=0$, then every term involving $\gamma_E$ or $\log2$ cancels.
Comparing with the left side of
\eqref{l3:eq:polar-match}, $\kappa_3\cdot3!=24$, $\kappa_2=\kappa_1=0$ and
$\kappa_0=16\zeta(3)$.  This is~\eqref{l3:eq:leading}, and the values of the
$\hat f_r(0)$ follow from~\eqref{l3:eq:apriori}.
\end{proof}

\begin{theorem}[Normalized three-loop expansion ($n=4$ case)]\label{l3:thm:vanhove}
On the Euclidean germ $s=-t$, $t$ sufficiently large,
\begin{equation}\label{l3:eq:frob-form}
 I_{\circleddash}^{(3)}(s)=\pi(s)
 \bigl(4R^3+f_2(Q)R^2+f_1(Q)R+f_0(Q)\bigr)
 -16\zeta(3)\pi(s),
\end{equation}
where $f_r(Q)=\sum_{m\ge1}d_r(m)Q^m\in Q\,\Q\{Q\}$, and the $f_r$ are uniquely
determined by~\eqref{l3:eq:frob-form}.  In terms of Lemma~\ref{l3:lem:frob-form},
$f_2=\hat f_2$, $f_1=\hat f_1$ and $f_0=\hat f_0+16\zeta(3)$. The boundary
constant has been displayed separately, so that from here on every $f_r$ has
zero constant term.  This is the normalization used in the rest of the paper.  Explicitly, $I^{(3)}_\circleddash=I_{\rm part}-16\zeta(3)\pi$,
so the $d_r(m)$ are the rational numbers obtained by expanding
\eqref{l3:eq:frob-explicit}.  The identity persists under simultaneous analytic
continuation of both sides along any path in $|s|>M$, $M$ large, provided
$\Log s$ and $R$ are continued along the same path.
\end{theorem}
\begin{proof}
Combine Lemma~\ref{l3:lem:frob-form} and Proposition~\ref{l3:prop:leading}.  The last
assertion holds because both sides are analytic continuations of the same germ.
\end{proof}

\begin{remark}\label{l3:rem:vanhove-comparison} (a) Theorem~\ref{l3:thm:vanhove}
is the $L=3$ case of~\cite[Theorem~4.1]{Vanhove2026}, whose proof proceeds
through the exact convergent residue expansion of
\cite[Theorem~2.1]{Vanhove2026}.  The proof given here uses only the
differential equation and the leading Mellin pole; in particular the rationality
of all $d_r(m)$ is obtained from the rationality of the Frobenius recurrence,
not from an expansion of Gamma quotients.  The magnitude $16\zeta(3)$ of the
boundary constant agrees with the generating series in~\cite{Vanhove2026}; its
sign depends on the normalization of the period and of the logarithm, which
differ between the two treatments (see the normalization dictionary in
Section~\Iref{sec:dictionary}, in particular Section~\Iref{sec:zetaambiguity}).
With the period and logarithm fixed here, the coefficient is exactly
$-16\zeta(3)$.

(b) The first coefficients obtained from~\eqref{l3:eq:frob-explicit} are
\[
\begin{aligned}
 f_2&=-24Q+132Q^2-800Q^3+6018Q^4-\tfrac{242424}{5}Q^5+\cdots,\\
 f_1&=48Q^2-528Q^3+4652Q^4+\cdots,\\
 f_0&=48Q-138Q^2+\tfrac{7144}{9}Q^3+\cdots .
\end{aligned}
\]
and $d_2(1),\dots,d_2(4)$ agree with the values tabulated in~\cite{Vanhove2026}.

(c) The particular solution $I_{\rm part}$ is characterized among all solutions
of $\mathcal L_3u=-24$ by $\hat f_2(0)=\hat f_1(0)=\hat f_0(0)=0$; the Feynman
integral differs from it by $-16\zeta(3)\pi$ only.  The coefficient identity of
Theorem~\ref{l3:thm:coefficient-bridge} below concerns $\hat f_2-\hat f_2(0)$ and is therefore
independent of the values of the boundary constants; only the statement that
the boundary term is exactly $-16\zeta(3)\pi(s)$ uses Proposition~\ref{l3:prop:leading}.
\end{remark}

\subsection{Diagonal restriction and the primary potential}
From now on set $y_i=y$ and $Q_i=Q$ for all $i$, and write $|\beta|=H\cdot\beta$
so that $\mathbf Q^\beta=Q^{|\beta|}$ and $\tau=H\log Q$.  The diagonally
restricted $I$-function is a symmetric expression in $p_1,\dots,p_4$ read in the
$S_4$-stable presentation of Lemma~\ref{l3:lem:state-space}; we use this
repeatedly.

\begin{lemma}[The coordinate $T$]\label{l3:lem:T}
The coefficient of $z^{-1}$ in $\mathcal I_{X_3}^0(y,y,y,y,z)$ equals $H\,T(y)$ with
\begin{equation}\label{l3:eq:T}
 T(y)=\log y+S(y),\qquad S(y)=\sum_{n\ge1}\frac{(-1)^nA_n}{n}y^n,\qquad
 Q=e^{T(y)}=y\,e^{S(y)} .
\end{equation}
The holomorphic part $S$ has radius of convergence exactly $1/16$; the
logarithm is read on the branch of Convention~\ref{l3:conv:branches}, where
$T(-s^{-1})=R(s)$.  Moreover $Q=y-4y^2+22y^3-152y^4+O(y^5)$ is a holomorphic
germ with $Q(0)=0$, $Q'(0)=1$, and its inverse germ is
$y(Q)=Q+4Q^2+10Q^3+32Q^4+O(Q^5)$.
\end{lemma}
\begin{proof}
Assign degree one to $z$ and to each of $H,p_i$; the summand of
\eqref{l3:eq:I-function} is then homogeneous of degree $0$, so its $z^{-1}$
coefficient is its part of cohomological degree one.  Every term of a summand
with $D\ge1$ contains the factor $-H$ coming from $r=0$; hence the degree-one
part is $-H$ times the degree-zero part of the rest, i.e.
\[
 \left[\tfrac Hz\right]
 \frac{\prod_{r=1}^{D}(H+rz)\prod_{r=0}^{D-1}(-H-rz)}
      {\prod_{i=1}^4\prod_{r=1}^{d_i}(p_i+rz)^2}
 =\frac{(-1)^D}{D}\binom{D}{d_1,d_2,d_3,d_4}^{\!2},
\]
using $D!(D-1)!=(D!)^2/D$.  For the same reason a summand with $D\ge1$ has no
degree-zero part, so it cannot combine with the degree-one part
$H\log y$ of the exponential prefactor; the degree-one part of
$\mathcal I^0_{X_3}$ is therefore $H\log y+H\sum_{n\ge1}(-1)^nA_ny^n/n$ by
\eqref{l3:eq:domb}.  Convergence of $S$ for $|y|<1/16$ follows from $A_n\le16^n$;
the radius is exactly $1/16$ because, taking $n=4k$ and retaining only the
term $d_1=d_2=d_3=d_4=k$ in~\eqref{l3:eq:domb}, one gets
$A_{4k}\ge\binom{4k}{k,k,k,k}^{2}$, whence $\limsup A_n^{1/n}\ge16$ by Stirling's formula.  The
identification $T(-s^{-1})=R(s)$ is Convention~\ref{l3:conv:branches}.  Finally
$Q=ye^{S(y)}$ is holomorphic at $0$ with $Q(0)=0$, $Q'(0)=1$, so the inverse
germ exists and is unique.
\end{proof}

\begin{lemma}[Well-definedness of $K_2$]\label{l3:lem:K2-welldefined}
The coefficient of $z^{-2}$ in $e^{-HT(y)/z}\mathcal I_{X_3}^0(y,y,y,y,z)$ is a
multiple of $H^2$ by a holomorphic germ in $y$ of radius $\ge1/16$.  Writing it
as $H^2K_2(Q)$ with $y=y(Q)$ defines a unique $K_2\in Q\,\Q\{Q\}$, and
\begin{equation}\label{l3:eq:K2def}
 e^{-HT(y)/z}\mathcal I_{X_3}^0(y,y,y,y,z)=1+\frac{H^2}{z^2}K_2(Q)+O(z^{-3}).
\end{equation}
\end{lemma}
\begin{proof}
The coefficient of $z^{-1}$ vanishes by Lemma~\ref{l3:lem:T}.  The coefficient of
$z^{-2}$ is a class of cohomological degree two, i.e.\ an element of $H^4(Y_3,\Q)$
read in $H^4(F_3,\Q)$; being obtained from a diagonally restricted expression it
is invariant under the permutation action of $S_4$ on $p_1,\dots,p_4$.  The
invariant subspace of $H^4(Y_3,\Q)=\bigoplus_{i<j}\Q\,p_ip_j$ is spanned by
$e_2(p)=\sum_{i<j}p_ip_j=H^2/2$, so the coefficient is a multiple of $H^2$, and
the multiple is unique because $H^2\ne0$ in $H^4(F_3,\Q)$ (indeed
$\int_{F_3}H^2\cdot H=24$).  Here the scalar is produced by the symmetry, so no division by the nilpotent class $H$ is needed.

The prefactor is $e^{-HT/z}\cdot e^{H\log y/z}=e^{-HS(y)/z}$, so no logarithm
survives, and the coefficient is a power series in $y$.  For its radius, expand
the finite products in the $\mathbf d$-summand of~\eqref{l3:eq:I-function} to
second order in the classes, then the coefficients are $\binom D{\mathbf d}^2$ times
sums of at most two harmonic numbers $h_r\le1+\log r$, $r\le D$, so the
$y^n$-coefficient is bounded by a constant times $A_n(1+\log n)^2$, and
$A_n\le16^n$ gives radius $\ge1/16$.  Composing with the holomorphic inverse
germ $y(Q)$ of Lemma~\ref{l3:lem:T} gives $K_2\in\C\{Q\}$; its coefficients are
rational, and $K_2(0)=0$ because $\mathcal I^0_{X_3}=1+O(y)$ at $z=\infty$.
\end{proof}

\begin{proposition}\label{l3:prop:potential}
For $\beta\ne0$ let $C^\lambda_\beta$ be the proper evaluation cycle
\eqref{l3:eq:Cbeta}, and put $C^\lambda_+(Q)=\sum_{\beta\ne0}Q^{|\beta|}C^\lambda_\beta$
and $C^0_+(Q)=\sum_{\beta\ne0}Q^{|\beta|}C^0_\beta$.  Then
\begin{equation}\label{l3:eq:primary-extraction}
 [z^{-2}]\bigl(e^{-\tau/z}\mathcal J_{X_3}^\lambda-\mathbf1\bigr)
 =(\lambda-H)\,C^\lambda_+(Q),
\end{equation}
$C^0_+(Q)=-H\,K_2(Q)$, and the one-point potential~\eqref{l3:eq:potential-intro}
satisfies
\begin{equation}\label{l3:eq:F-K}
 \mathcal F_\gamma(Q)=-12\,K_2(Q).
\end{equation}
In particular $\mathcal F_\gamma$ is a convergent power series near $Q=0$, with
the same radius as $K_2$.
\end{proposition}
\begin{proof}
Expanding $1/(z(z-\psi))=z^{-2}+\psi z^{-3}+\cdots$ in~\eqref{l3:eq:geometric-J},
\[
 [z^{-2}]\bigl(e^{-\tau/z}\mathcal J_{X_3}^\lambda-\mathbf1\bigr)
 =\sum_{\beta\ne0}Q^{|\beta|}\sum_\alpha\phi^\alpha
 \langle\phi_\alpha\rangle^{X_3,\lambda}_{0,1,\beta}
 =\sum_{\beta\ne0}Q^{|\beta|}\sum_\alpha\phi^\alpha\int_{F_3}C^\lambda_\beta\,\phi_\alpha ,
\]
by Definition~\ref{l3:def:local}(ii),~\eqref{l3:eq:vfc-eq}, and the projection
formula for $\ev_{F_3}$.  The element $u\in\mathcal H_{X_3}$ with
$(u,\phi_\alpha)_{X_3,\lambda}=\int_{F_3}C\phi_\alpha$ for all $\alpha$ is
$u=(\lambda-H)C$, because $\int_{F_3}(\lambda-H)C\,\phi_\alpha/(\lambda-H)=\int_{F_3}C\phi_\alpha$
and the Poincar\'e pairing is nondegenerate on the even part.  This proves
\eqref{l3:eq:primary-extraction}; it is the algebraic form of the equivariant
self-intersection formula $\iota^*\iota_*=e_\lambda(K_{F_3})\cup(-)$, as recorded in
\eqref{l3:eq:self-intersection}.

By~\eqref{l3:eq:Cbeta-structure}, $C^\lambda_\beta=C^0_\beta+\lambda c_\beta\mathbf1$
with $C^0_\beta\in H^2(F_3,\Q)$ and $c_\beta\in\Q$, so both sides of
\eqref{l3:eq:primary-extraction} are polynomial in $\lambda$ and may be
specialized to $\lambda=0$.  Using $\mathcal J^0_{X_3}=\mathcal I^0_{X_3}$
(Proposition~\ref{l3:prop:mirror-map} at $\lambda=0$), $\tau=HT$, and
\eqref{l3:eq:K2def}, the specialization reads
\[
 -H\,C^0_+(Q)=H^2K_2(Q),\qquad\text{i.e.}\qquad
 H\cup\bigl(C^0_+(Q)+HK_2(Q)\bigr)=0 .
\]
Both $C^0_+$ and $HK_2$ lie in $H^2(F_3)$, and $H\cup-:H^2(F_3)\to H^4(F_3)$ is an
isomorphism by Hard Lefschetz; hence
\begin{equation}\label{l3:eq:C-K}
 C^0_+(Q)=-H\,K_2(Q).
\end{equation}
This is a cancellation by an isomorphism, not a division by the nilpotent
class $H$.  Pairing with $\gamma=H^2/2$ and using the projection formula,
$\ev_{X_3}^*\varpi^*\gamma=\ev_{F_3}^*\gamma$, and $\int_{F_3}H^3=24$:
\[
 \mathcal F_\gamma(Q)=\int_{F_3}\gamma\cup C^0_+(Q)=-K_2(Q)\int_{F_3}\tfrac{H^2}2\cdot H
 =-12K_2(Q).\qedhere
\]
\end{proof}

\subsection{The logarithmic B-model recurrence}

\label{l3:app:logarithmic-recurrence}

This subsection supplies the calculation behind $\mathscr P_yB=0$ in
Lemma~\ref{l3:lem:Bop}.  The differential equation of the diagonally restricted logarithmic series does not follow automatically from that of $A$, so we derive it.  All computations take place over the dual numbers
$\mathbb D=\C[\varepsilon]/(\varepsilon^2)$; every equality below is an equality
in $\mathbb D$, i.e.\ modulo $\varepsilon^2$.  (We write $\varepsilon$, not
$\rho$, to keep this deformation parameter distinct from the Frobenius exponent
of Section~\ref{l3:sec:analytic}, although the two play analogous roles.)

Throughout, $y^\varepsilon$ means $1+\varepsilon\log y$, and
$S_{-1}=S_{-2}=0$.  For $\mathbf d\in\mathbb N^4$ put
\[
 C_{\mathbf d}(\varepsilon)=
 \frac{\Gamma(|\mathbf d|+1+\varepsilon)^2}
      {\prod_{i=1}^4\Gamma(d_i+1+\varepsilon/4)^2},
 \qquad
 S_n(\varepsilon)=\sum_{|\mathbf d|=n}C_{\mathbf d}(\varepsilon).
\]
From $\tfrac{d}{d\varepsilon}\log C_{\mathbf d}
=2\psi(|\mathbf d|+1+\varepsilon)-\tfrac12\sum_i\psi(d_i+1+\varepsilon/4)$ and
$\psi(k+1)=h_k-\gamma_E$,
\begin{equation}\label{l3:eq:S-derivative}
 S_n(0)=A_n,\qquad
 S_n'(0)=\sum_{|\mathbf d|=n}\frac{(n!)^2}{\prod_{i=1}^4(d_i!)^2}
 \left(2h_n-\frac12\sum_{i=1}^4 h_{d_i}\right),
\end{equation}
the Euler constants cancelling because $2-4\cdot\tfrac12=0$, in the same way as
in the proof of Lemma~\ref{l3:lem:B-expansion}.  Hence, by
Lemma~\ref{l3:lem:B-expansion},
\begin{equation}\label{l3:eq:rho-family}
 y^\varepsilon\sum_{n\ge0}(-1)^nS_n(\varepsilon)y^n=A(y)+\varepsilon B(y).
\end{equation}

Let $\theta_y=y\partial_y$ and put
\[
 F(y,\varepsilon)=\sum_{k\ge0}\frac{y^k}{\Gamma(k+1+\varepsilon/4)^2},
 \qquad h(y,\varepsilon)=F(y,\varepsilon)^4 .
\]
Since $\Gamma(k+1+\tfrac\varepsilon4)=(k+\tfrac\varepsilon4)\Gamma(k+\tfrac\varepsilon4)$,
\[
 (\theta_y+\tfrac\varepsilon4)^2F=\sum_{k\ge0}\frac{y^k}{\Gamma(k+\varepsilon/4)^2}
 =yF+\frac{1}{\Gamma(\varepsilon/4)^2}=yF,
\]
the correction term vanishing because
$\Gamma(\varepsilon/4)^{-2}=\varepsilon^2/16+O(\varepsilon^3)=0$ in $\mathbb D$.
(This is the one step where working modulo $\varepsilon^2$ is essential.)
Hence $\widetilde F=y^{\varepsilon/4}F$ satisfies
$\theta_y^2\widetilde F=y\widetilde F$.  With $v=\widetilde F^4$, repeated
differentiation gives
\begin{align*}
 \theta_y v&=4\widetilde F^3\theta_y\widetilde F,\\
 \theta_y^2v&=4yv+12\widetilde F^2(\theta_y\widetilde F)^2,\\
 \theta_y^3v&=4yv+10y\theta_y v+24\widetilde F(\theta_y\widetilde F)^3,\\
 \theta_y^4v&=4yv+14y\theta_y v+16y\theta_y^2v-24y^2v+24(\theta_y\widetilde F)^4,\\
 \theta_y^5v&=4yv+18y\theta_y v+30y\theta_y^2v+16y\theta_y^3v-48y^2v-24y^2\theta_y v
             +96y\widetilde F(\theta_y\widetilde F)^3 .
\end{align*}
Eliminating the last cubic term with the third line yields
\begin{equation}\label{l3:eq:symmetric-fourth}
 \left[\theta_y^5-y\,p(\theta_y)+64y^2(\theta_y+1)\right]v=0,
 \qquad p(x)=20x^3+30x^2+18x+4 .
\end{equation}
Since $v=y^{\varepsilon}h$ and $\theta_y(y^\varepsilon h)=y^\varepsilon(\theta_y+\varepsilon)h$,
\begin{equation}\label{l3:eq:h-operator}
 \left[(\theta_y+\varepsilon)^5-y\,p(\theta_y+\varepsilon)
       +64y^2(\theta_y+\varepsilon+1)\right]h=0 .
\end{equation}

The diagonal Gamma series is obtained from $h$ by the coefficientwise multiplier
\[
 C^+(y,\varepsilon)=\Gamma(\theta_y+1+\varepsilon)^2h(y,\varepsilon)
 =\sum_{n\ge0}S_n(\varepsilon)y^n ;
\]
here $\Gamma(\theta_y+1+\varepsilon)^2$ is the formal operator multiplying the
coefficient of $y^n$ by the invertible element
$\Gamma(n+1+\varepsilon)^2\in\mathbb D$, not an analytic function of $\theta_y$.
Writing $\Phi=\Gamma(\theta_y+1+\varepsilon)^2$, one has $\Phi\circ y\circ\Phi^{-1}
=y(\theta_y+1+\varepsilon)^2$ and $\Phi\circ y^2\circ\Phi^{-1}
=y^2(\theta_y+2+\varepsilon)^2(\theta_y+1+\varepsilon)^2$, while $\Phi$ commutes with
$f(\theta_y)$.  Conjugating~\eqref{l3:eq:h-operator} therefore gives
\begin{multline}\label{l3:eq:fifth-order-deformed}
 \Big[(\theta_y+\varepsilon)^5-y(\theta_y+\varepsilon+1)^2p(\theta_y+\varepsilon)\\
 +64y^2(\theta_y+\varepsilon+1)^3(\theta_y+\varepsilon+2)^2\Big]C^+(y,\varepsilon)=0 .
\end{multline}
For $n\ge1$ the coefficient of $y^n$ reads
\[
 (n+\varepsilon)^5S_n(\varepsilon)
 -(n+\varepsilon)^2p(n-1+\varepsilon)S_{n-1}(\varepsilon)
 +64(n-1+\varepsilon)^3(n+\varepsilon)^2S_{n-2}(\varepsilon)=0 ,
\]
while for $n=0$ it reads $\varepsilon^5S_0=0$, which holds in $\mathbb D$.  The
factor $(n+\varepsilon)^2$ is invertible in $\mathbb D$ for $n\ge1$, so
\begin{equation}\label{l3:eq:S-recurrence}
 (n+\varepsilon)^3S_n(\varepsilon)-p(n-1+\varepsilon)S_{n-1}(\varepsilon)
 +64(n-1+\varepsilon)^3S_{n-2}(\varepsilon)=0 .
\end{equation}
Since $p(-x-1)=-p(x)$ one has
$p(n-1+\varepsilon)=20(n+\varepsilon)^3-30(n+\varepsilon)^2+18(n+\varepsilon)-4$,
so~\eqref{l3:eq:S-recurrence} is the deformation of~\eqref{l3:eq:domb-rec} used in
Lemma~\ref{l3:lem:Bop}; at $n=0$ the third-order relation reads
$\varepsilon^3S_0=0$, again valid in $\mathbb D$.
 Substituting $y\mapsto-y$ turns~\eqref{l3:eq:S-recurrence} into the statement
\[
 \left[(\theta_y+\varepsilon)^3+y\,p(\theta_y+\varepsilon)
 +64y^2(\theta_y+\varepsilon+1)^3\right]
 \sum_{n\ge0}(-1)^nS_n(\varepsilon)y^n=0,
\]
and multiplication by $y^\varepsilon$ conjugates every $\theta_y+\varepsilon$ to
$\theta_y$.  By~\eqref{l3:eq:rho-family} the constant term in $\varepsilon$ is
$\mathscr P_yA=0$ and the coefficient of $\varepsilon$ is $\mathscr P_yB=0$.
This proves the logarithmic equation coefficientwise, hence in the ring of
convergent germs by the bounds $A_n\le16^n$ and $h_n=O(\log n)$ used in the
proof of Theorem~\ref{l3:thm:coefficient-bridge}.

\subsection{Quantum-Serre coefficients and the coefficient bridge}

Introduce the sign-twisted series
\begin{equation}\label{l3:eq:IW}
 I^\sharp(y,z)=e^{H\log y/z}\sum_{D\ge0}(-1)^Dy^D
 \sum_{|\mathbf d|=D}
 \frac{\prod_{m=1}^{D}(H+mz)^2}
 {\prod_{i=1}^4\prod_{m=1}^{d_i}(p_i+mz)^2},
\end{equation}
and write $I^\sharp=A(y)\mathbf1+B(y)H/z+O(z^{-2})$, $\tau^\sharp=B/A$.

\begin{remark}\label{l3:rem:IW-geometry}
Up to the sign $(-1)^D$,~\eqref{l3:eq:IW} is the hypergeometric $I$-function of a
smooth member $W\in|-K_{F_3}|$, that is, of a $(1,1,1,1)$-complete intersection of
two divisors in $Y_3$, which is a $K3$ surface~\cite{BKV3,BFKNS}.  
Lemma~\ref{l3:lem:Bop} is a statement about the formal series~\eqref{l3:eq:IW} alone,
and no mirror theorem for $W$ is invoked.
\end{remark}

\begin{lemma}[A coefficient identity]\label{l3:lem:qserre}
Coefficientwise, with the canonical ambient lift of~\eqref{l3:eq:I-mu} used before restriction to $F_3$,
\[
 z\thetay\,\mathcal I_{X_3}^0=H\,I^\sharp,\qquad
 A=\thetay T,\qquad
 \tau^\sharp-T=\thetaQ K_2 .
\]
\end{lemma}
\begin{proof}
On the degree-$D$ summand, $z\thetay$ multiplies by $H+Dz$.  For $D\ge1$,
\[
 \prod_{m=0}^{D-1}(-H-mz)=(-1)^DH\prod_{m=1}^{D-1}(H+mz),
\]
and $(H+Dz)\prod_{m=1}^{D-1}(H+mz)=\prod_{m=1}^{D}(H+mz)$, so the numerator of
the $\mathbf d$-summand of $z\thetay\mathcal I^0_{X_3}$ becomes
$(-1)^DH\prod_{m=1}^D(H+mz)^2$, which is $H$ times the corresponding numerator
in~\eqref{l3:eq:IW}; the $D=0$ terms agree as well, both being $He^{H\log y/z}$.
This proves the first identity.  Applying $z\thetay$ to~\eqref{l3:eq:K2def} in the
form
$\mathcal I^0_{X_3}=1+\frac{HT}{z}+\frac{H^2}{z^2}\bigl(\tfrac{T^2}2+K_2\bigr)+O(z^{-3})$
gives
$z\thetay\mathcal I^0_{X_3}=H\thetay T+\frac{H^2}{z}\thetay(\tfrac{T^2}2+K_2)+O(z^{-2})$.
Comparing coefficients of $H$ and of $H^2/z$ with
$HI^\sharp=HA+\frac{H^2}{z}B+O(z^{-2})$ yields $A=\thetay T$ and
$B=\thetay(\tfrac{T^2}{2}+K_2)=T\thetay T+\thetay K_2$.  Since $T=\log Q$ we
have $\thetay=(\thetay T)\thetaQ=A\thetaQ$, so $\thetay K_2=A\thetaQ K_2$ and
$B=A(T+\thetaQ K_2)$, i.e.\ $\tau^\sharp=B/A=T+\thetaQ K_2$.  (The comparison of
coefficients of $H$ and $H^2$ is legitimate for the same $S_4$-symmetry reason as
in Lemma~\ref{l3:lem:K2-welldefined}.)
\end{proof}

Set $a_n=(-1)^nA_n$, so that $A(y)=\sum_{n\ge0}a_ny^n$ by~\eqref{l3:eq:domb}.

\begin{lemma}[Expansion of $B$]\label{l3:lem:B-expansion}
$B(y)=A(y)\log y+\widetilde B(y)$ with $\widetilde B$ a holomorphic germ of
radius $\ge1/16$:
\begin{equation}\label{l3:eq:B-expansion}
 \widetilde B(y)=\sum_{n\ge1}b_ny^n,\quad
 b_n=(-1)^n\!\sum_{|\mathbf d|=n}\frac{(n!)^2}{\prod_{i=1}^4(d_i!)^2}
 \left(2h_n-\frac12\sum_{i=1}^4 h_{d_i}\right),\quad h_n=\sum_{j=1}^n\frac1j .
\end{equation}
\end{lemma}
\begin{proof}
Expanding the $\mathbf d$-summand of~\eqref{l3:eq:IW} to first order in the
classes,
\[
 \frac{\prod_{m=1}^{D}(H+mz)^2}{\prod_i\prod_{m=1}^{d_i}(p_i+mz)^2}
 =\binom{D}{\mathbf d}^{2}\left[1+\frac1z\Bigl(2h_DH-2\sum_ih_{d_i}p_i\Bigr)
 +O(z^{-2})\right],
\]
which gives $A$ at order $z^0$.  For the $H/z$ coefficient, the sum over
$\{\mathbf d:|\mathbf d|=D\}$ is $S_4$-symmetric, so
$\sum_{\mathbf d}\binom D{\mathbf d}^2\sum_ih_{d_i}p_i
=\frac H4\sum_{\mathbf d}\binom D{\mathbf d}^2\sum_ih_{d_i}$; this is the origin
of the factor $\tfrac12=\tfrac24$ in~\eqref{l3:eq:B-expansion}.  Adding the
contribution $A\log y$ of the exponential prefactor gives the expansion; the
radius follows from $h_n\le1+\log n$ and $A_n\le16^n$.
\end{proof}

\begin{lemma}[The scalar B-model operator]\label{l3:lem:Bop}
Both $A(y)$ and $B(y)$ are annihilated by
\begin{equation}\label{l3:eq:Py}
 \mathscr P_y=\thetay^3+y\,p(\thetay)+64y^2(\thetay+1)^3,
 \qquad p(x)=20x^3+30x^2+18x+4 .
\end{equation}
Under $y=-s^{-1}$ and the gauge $u(s)=s^{-1}v(-s^{-1})$ one has
$(\mathscr P_yv)(-s^{-1})=-\mathcal L_3u(s)$.  Consequently
\[
 \pi(s)=\frac{A(-s^{-1})}{s},\qquad
 \pi(s)\,\tau^\sharp(-s^{-1})=\frac{B(-s^{-1})}{s}
\]
solve $\mathcal L_3u=0$.
\end{lemma}
\begin{proof}
Reading off the coefficient of $y^n$ in $\mathscr P_yA=0$ gives
$n^3a_n+p(n-1)a_{n-1}+64(n-1)^3a_{n-2}=0$, which for $a_n=(-1)^nA_n$ is
\eqref{l3:eq:domb-rec}, since $p(n-1)=20n^3-30n^2+18n-4$.  For $B$,
Subsection~\ref{l3:app:logarithmic-recurrence} derives the deformed recurrence. Over the dual numbers
$\mathbb D=\C[\varepsilon]/(\varepsilon^2)$ the sums
\[
 S_n(\varepsilon)=\sum_{|\mathbf d|=n}
 \frac{\Gamma(n+1+\varepsilon)^2}{\prod_{i=1}^4\Gamma(d_i+1+\varepsilon/4)^2}
\]
satisfy $S_n(0)=A_n$, $S_n'(0)=$ the unsigned harmonic sum of
\eqref{l3:eq:B-expansion} (Euler's constant cancels because $2-4\cdot\tfrac12=0$; see~\eqref{l3:eq:S-derivative}), and
\[
 \mathscr P_y\Bigl(y^\varepsilon\sum_{n\ge0}(-1)^nS_n(\varepsilon)y^n\Bigr)=0
 \quad\text{in }\mathbb D .
\]
The constant term in $\varepsilon$ is $\mathscr P_yA=0$ and the coefficient of
$\varepsilon$ is $\mathscr P_yB=0$, by~\eqref{l3:eq:B-expansion}.

For the conjugation, $\mathscr P_y$ is the operator $P_3$
of~\eqref{l3:eq:L3-theta}, and $v(-s^{-1})=s\,u(s)=-u(s)/y$ under $y=-s^{-1}$;
hence $(\mathscr P_yv)(-s^{-1})=-P_3(u/y)=-\mathcal L_3u$
by~\eqref{l3:eq:L3-theta}.  Finally
$A(-s^{-1})=\sum_nA_ns^{-n}$, so $A(-s^{-1})/s=\pi(s)$ by~\eqref{l3:eq:pi}, and
$B(-s^{-1})/s=\pi\,\tau^\sharp(-s^{-1})$ since $\tau^\sharp=B/A$.
\end{proof}

\begin{theorem}[Coefficient bridge]\label{l3:thm:coefficient-bridge}
As convergent analytic germs at $Q=0$,
\begin{equation}\label{l3:eq:bridge}
 f_2(Q)=12\bigl(\tau^\sharp(y(Q))-T(y(Q))\bigr)
 =12\,\thetaQ K_2(Q)=-\thetaQ\mathcal F_\gamma(Q).
\end{equation}
\end{theorem}
\begin{proof}
\emph{Convergence.}  Neither $\tau^\sharp=B/A$ nor $T$ is a holomorphic
germ at $Q=0$, since both contain $\log y$; their difference is.  Indeed
$\tau^\sharp-T=\widetilde B/A-S$ by Lemmas~\ref{l3:lem:T} and~\ref{l3:lem:B-expansion},
and $A$, $\widetilde B$, $S$ are holomorphic germs in $y$ of radius $\ge1/16$
with $A(0)=1$; composing with the holomorphic inverse germ $y(Q)$ shows that
every term of~\eqref{l3:eq:bridge} is a convergent germ at $Q=0$.

\emph{The scalar equation in the coordinate $R$.}  On the germ, $Q$ and $s$
determine each other (Convention~\ref{l3:conv:branches}), and $\partial_R$ acts on
functions of $Q$ as $\thetaQ$, since $Q=e^R$.  Since $R'(s)=-\pi(s)$, the chain rule reads
$\partial_s=-\pi\,\partial_R$ on functions of $R$.  Writing
$\mathcal L_3=\mathsf A_3\partial_s^3+\mathsf A_2\partial_s^2+\mathsf A_1\partial_s
+\mathsf A_0$ (the coefficients read off~\eqref{l3:eq:L3}; the letter $\mathsf A$
avoids the $a_n$ of Lemma~\ref{l3:lem:Bop}) and $\pi',\pi''$
for $s$-derivatives, one finds for any $\mathcal H$
\[
\begin{aligned}
 \mathcal L_3(\pi\mathcal H)&=\pi\bigl(L_3\partial_R^3+L_2\partial_R^2
 +L_1\partial_R\bigr)\mathcal H,\\
 L_3&=-\mathsf A_3\pi^3,\qquad L_2=6\mathsf A_3\pi\pi'+\mathsf A_2\pi^2,\\
 L_1&=-\mathsf A_3\Bigl(\frac{3\pi'^2}{\pi}+4\pi''\Bigr)-3\mathsf A_2\pi'-\mathsf A_1\pi.
\end{aligned}
\]
with no zeroth-order term because $\mathcal L_3\pi=0$.  The $L_i$ are analytic
on a punctured disc at $Q=0$ and meromorphic at $Q=0$; explicitly
$L_3=-s^2(s-4)(s-16)\pi^3\sim -s\sim Q^{-1}$, so they are \emph{not} power
series in $Q$.  All the statements below are equalities in the
field of meromorphic germs, and the operator
$L(Q,\thetaQ)=L_3\thetaQ^3+L_2\thetaQ^2+L_1\thetaQ$ is used only through its
kernel.

\emph{The $R^2$-component.}  Substitute~\eqref{l3:eq:frob-form} into
$\mathcal L_3I^{(3)}_\circleddash=-24$ (Lemma~\ref{l3:lem:ode}) and divide by $\pi$; the source
$-24/\pi$ is a function of $Q$ alone, with no powers of $R$.  Since
$\partial_R(gR^k)=(\thetaQ g)R^k+kgR^{k-1}$, the coefficient of $R^2$ is
\begin{equation}\label{l3:eq:f2rec}
 L(Q,\thetaQ)f_2+12L_1=0,
\end{equation}
the term $12L_1$ coming from $L_1\partial_R(4R^3)=12L_1R^2$ and the constancy of
the leading coefficient $4$ (Lemma~\ref{l3:lem:frob-form}).  Comparison of
coefficients of powers of $R$ is legitimate by the transcendence of $R$ over the
field of convergent Laurent series in $Q$, proved in Lemma~\ref{l3:lem:frob-form}.

\emph{The A-model solution of the same equation.}  By Lemmas~\ref{l3:lem:qserre}
and~\ref{l3:lem:Bop}, $\pi\tau^\sharp=\pi(R+\thetaQ K_2)$ solves
$\mathcal L_3u=0$.  Dividing by $\pi$ and collecting the coefficient of $R^0$
gives $L_1+L(Q,\thetaQ)\thetaQ K_2=0$; multiplying by $12$ shows that
$12\thetaQ K_2$ also satisfies~\eqref{l3:eq:f2rec}.

\emph{Uniqueness.}  Put $\Delta=f_2-12\thetaQ K_2$.  Then
$L(Q,\thetaQ)\Delta=0$, i.e.\ $\pi\Delta$ solves $\mathcal L_3u=0$; and
$\pi\Delta$ is log-free, being $\pi$ times a convergent power series in $Q$.
The coefficient recurrence~\eqref{l3:eq:frob-rec} at $\rho=-1$ has nonzero
coefficient $(-n)^3$ for every subsequent unknown, $n\ge1$; therefore the space
of log-free solutions of this form is one-dimensional and spanned by $\pi$; hence
$\Delta$ is constant.  Both $f_2$ and $\thetaQ K_2$ vanish at $Q=0$ --- for
$f_2$ this is Proposition~\ref{l3:prop:leading}, for $K_2$ it is
Lemma~\ref{l3:lem:K2-welldefined} --- so $\Delta=0$.  The last equality in
\eqref{l3:eq:bridge} is Proposition~\ref{l3:prop:potential}.
\end{proof}

\begin{remark}
The identity also explains the absence of a quadratic term after passing to the
logarithmic period-ratio coordinate $\tau^\sharp$. Since $\tau^\sharp=R+\thetaQ K_2$ and
$f_2=12\thetaQ K_2$, the $R^2$ term of $4R^3+f_2R^2+\cdots$ cancels when the
expansion is rewritten in $\tau^\sharp$.  This matches the distinction between
$R^{(L)}$ and the K\"ahler parameter noted in~\cite{Vanhove2026}.
\end{remark}

\subsection{The three-loop theorem and all three numerical descendant potentials}
For $a=0,1,2$ define
\[
 \mathcal D_a(Q)=\sum_{\beta>0}Q^{|\beta|}
 \int_{[\overline{\mathcal M}_{0,1}(X_3,\beta)]^{\rm vir}}
 \psi^a\ev_{X_3}^*\varpi^*\left(\frac{H^{2-a}}{(2-a)!}\right).
\]
Thus $\mathcal D_0=\mathcal F_\gamma$. These insertions have total complex
codimension two, equal to the local one-marked virtual dimension.
To define the two further coefficients without losing information, retain
$I^{0,0}(y,y,y,y,z)$ from~\eqref{l3:eq:I-mu} in
$\mathcal R=\Q[p_1,\ldots,p_4]/(p_1^2,\ldots,p_4^2)$, before applying $j^*$,
and set
\begin{equation}\label{l3:eq:full-ambient-K}
 e^{-HT(y)/z}I^{0,0}(y,y,y,y,z)
 =1+\frac{H^2}{z^2}K_2(Q)+\frac{H^3}{z^3}K_3(Q)
       +\frac{H^4}{z^4}K_4(Q),\qquad Q=e^{T(y)}.
\end{equation}
Its $K_2$ agrees with the definition above. On $F_3$ one has $H^4=0$; therefore
$K_4$ must be read from the ambient series, divided as in the companion
paper~\cite{PaperI} (see~\Ithmref{thm:extraction}), and only then used for
the numerical descendant.

\begin{theorem}[Complete descendant specialization]\label{l3:thm:full-descendant-assembly}
As convergent germs at $Q=0$,
\begin{equation}\label{l3:eq:all-descendants}
 \mathcal D_0=-12K_2,\qquad \mathcal D_1=-24K_3,\qquad
 \mathcal D_2=-24K_4,
\end{equation}
and the coefficients of the three-loop expansion satisfy
\begin{equation}\label{l3:eq:all-coefficients}
 f_2=-\thetaQ\mathcal D_0,\qquad
 f_1=-2\mathcal D_0-\thetaQ\mathcal D_1,\qquad
 f_0=-\mathcal D_1-\thetaQ\mathcal D_2.
\end{equation}
Thus if $N_{a,m}=[Q^m]\mathcal D_a$, then
\[
 d_2(m)=-mN_{0,m},\quad
 d_1(m)=-2N_{0,m}-mN_{1,m},\quad
 d_0(m)=-N_{1,m}-mN_{2,m}.
\]
The full physical expression is
\[
 I^{(3)}_\circleddash(s)=\pi(s)\bigl[4R^3-(\thetaQ\mathcal D_0)R^2
 -(2\mathcal D_0+\thetaQ\mathcal D_1)R
 -\mathcal D_1-\thetaQ\mathcal D_2-16\zeta(3)\bigr].
\]
\end{theorem}
\begin{proof}
The multipliers in~\eqref{l3:eq:all-descendants} are $-4!/(2-a)!$ from
\Ithmref{thm:extraction}, and the assembly with its coefficients is
\Icorref{cor:nfour}. Since $\Pi_3=-yA=\pi$ no further sign conversion is
required. This proof uses the general analytic assembly; the alternative in
Section~\ref{ext:reconstruction} avoids it. In particular, the identity for
$f_2$ alone, which is Theorem~\ref{l3:thm:main}, does not use it.
\end{proof}

The theorem also identifies $d_1$ and $d_0$ with descendants. No assumption on middle cohomology is needed, since in this dimension all even one-point evaluation classes are ambient (Lemma~\ref{l3:lem:state-space}).

\begin{theorem}\label{l3:thm:main}
On the Euclidean germ $s=-t$ for $t$ sufficiently large, and under simultaneous
analytic continuation of both sides along any path in $|s|>M$ with $M$ large,
\begin{multline}\label{l3:eq:main}
 I^{(3)}_\circleddash(s)=\pi(s)\left[
 4R(s)^3+\sum_{m\ge1}Q^m
 \left(d_0(m)+d_1(m)R(s)-mN_m(\gamma)R(s)^2\right)\right]
 \cr-16\zeta(3)\pi(s),
\end{multline}
where $Q=e^{R(s)}$, $d_0(m),d_1(m)\in\Q$, and $N_m(\gamma)$ is the genus-zero
invariant~\eqref{l3:eq:potential-intro}.  In particular
\begin{equation}\label{l3:eq:coeff-main}
 d_2(m)=-m\,N_m(\gamma)\qquad(m\ge1).
\end{equation}
\end{theorem}
\begin{proof}
Combine Theorem~\ref{l3:thm:vanhove} with Theorem~\ref{l3:thm:coefficient-bridge} and
compare coefficients of $Q^m$ in $f_2=-\thetaQ\mathcal F_\gamma$.  Properness of
the moduli stacks and the virtual definition of the $N_m(\gamma)$ are
Proposition~\ref{l3:prop:zero-section} and Proposition~\ref{l3:prop:virtual-reduction}.
\end{proof}

\begin{remark}[Independence of the boundary constants]\label{l3:rem:independence}
Only the explicit term $-16\zeta(3)\pi(s)$ in~\eqref{l3:eq:main} uses
Proposition~\ref{l3:prop:leading} in full.  Identity~\eqref{l3:eq:coeff-main} uses it
only through $\hat f_2(0)=0$; and even that could be replaced by any
determination of the single constant $\alpha_2$, since
Theorem~\ref{l3:thm:coefficient-bridge} shows unconditionally that
$\hat f_2-\hat f_2(0)=12\thetaQ K_2$.
\end{remark}

Expanding~\eqref{l3:eq:I-function} through cohomological degree two and converting
from $y$ to $Q$ gives
\[
 K_2(Q)=-2Q+\frac{11}{2}Q^2-\frac{200}{9}Q^3+\frac{1003}{8}Q^4
 -\frac{20202}{25}Q^5+O(Q^6),
\]
whence $\mathcal F_\gamma=-12K_2$ and $f_2=-\thetaQ\mathcal F_\gamma$ give, in particular,
$d_2(1),\ldots,d_2(4)=-24,132,-800,6018$, in agreement with the independent
Feynman-side computation of Remark~\ref{l3:rem:vanhove-comparison}(b).

The following elementary check confirms the normalization of the whole chain
$K_2\rightsquigarrow\mathcal F_\gamma\rightsquigarrow f_2$ at $m=1$ without
using mirror symmetry at all.

\begin{proposition}[Classical verification at $m=1$]\label{l3:prop:m1-check}
$N_1(\gamma)=24$.
\end{proposition}

\begin{proof}
A class $\beta$ with $|\beta|=1$ satisfies $p_i\cdot\beta=\delta_{ij}$ for
exactly one $j$; call it $\beta_j$.  A stable map of class $\beta_j$ has
irreducible image, so its image is a curve meeting the $j$-th ruling once and
contracted by the other three projections, that is
$L_b=\mathbb P^1\times\{b\}$ with $b\in(\mathbb P^1)^3$ (coordinates in the
factors other than $j$).  Writing the defining multilinear equation of $F_3$ as
$A(b)x_j^{(0)}+B(b)x_j^{(1)}$, one has $L_b\subset F_3$ if and only if
$A(b)=B(b)=0$.  Thus the lines of class $\beta_j$ are parametrized by the
complete intersection $C_j=\{A=B=0\}$, of class $h^2$ in $(\mathbb P^1)^3$,
where $h$ is the sum of the three remaining point classes.

At a point of $C_j$, dependence of $dA$ and $dB$ would give
$[x_j^{(0)}:x_j^{(1)}]$ for which the base differential of
$Ax_j^{(0)}+Bx_j^{(1)}$ vanishes. Its fiber differential already vanishes
there; this would be a singular point of $F_3$. Hence $dA,dB$ are independent,
so $C_j$ is a smooth regular codimension-two zero locus.
The ambient one-marked degree-$\beta_j$ moduli is
$\mathbb P^1\times(\mathbb P^1)^3$; the section cutting out the
hypersurface moduli cuts out $\mathbb P^1\times C_j$ over arbitrary bases.

Since $L=L_b\subset F_3\subset Y_3$ and $N_{L/Y_3}=\mathcal O_L^{\oplus3}$, the normal
sequence of $L\subset F_3\subset Y_3$ reads
$0\to N_{L/F_3}\to\mathcal O_L^{\oplus3}\to\mathcal O_L(H\cdot\beta_j)=
\mathcal O_L(1)\to0$.  A surjection $\mathcal O_L^{\oplus3}\twoheadrightarrow
\mathcal O_L(1)$ is given by three sections of $\mathcal O_L(1)$ without common
zero, hence spanning $H^0(\mathcal O_L(1))$; after a change of basis the map is
$(f_1,f_2,f_3)\mapsto f_1s_0+f_2s_1$ for a basis $s_0,s_1$, so
\[
 N_{L/F_3}\cong\mathcal O_L\oplus\mathcal O_L(-1).
\]
Therefore $H^1(L,N_{L/F_3})=0$, and from $0\to T_L\to T_{F_3}|_L\to N_{L/F_3}\to0$ with
$H^1(\mathcal O_L(2))=0$ we get $H^1(L,T_{F_3}|_L)=0$.  Consequently
$\overline{\mathcal M}_{0,1}(F_3,\beta_j)$ is unobstructed of dimension
$1+1=2=\vdim$, its virtual class is its fundamental class, and it has no boundary, because a contracted subtree with $v\geq1$ components,
one attachment, and at most one marking has at most $2v$ special points
in total, whereas stability requires at least $3v$. Thus it is the universal line
$\mathbb P^1\times C_j$, with $\ev(t,b)=(t,b)$.  Since $\operatorname{rk}
\mathcal V=|\beta_j|-1=0$, Proposition~\ref{l3:prop:virtual-reduction} identifies
$N_{\beta_j}(\gamma)$ with $\int_{\mathbb P^1\times C_j}\ev^*\gamma$.

Split $\gamma=\sum_{i<k}p_ip_k$ into the part containing $p_j$ and the rest.  The
second part is pulled back from the one-dimensional $C_j$ in degree $4$, hence
vanishes.  The first part is $p_j\cdot h$, and $\ev^*p_j$ is the point class of
the $\mathbb P^1$ factor, so
\[
 \int_{\mathbb P^1\times C_j}\ev^*\gamma=\deg\bigl(h|_{C_j}\bigr)
 =\int_{(\mathbb P^1)^3}h\cdot h^2=\int_{(\mathbb P^1)^3}h^3=3!=6 .
\]
Summing over the four classes gives $N_1(\gamma)=4\cdot6=24$.
\end{proof}

This matches the coefficient of $Q$ in $\mathcal F_\gamma=-12K_2$ read off
above, and hence $d_2(1)=-1\cdot24=-24$, which is the coefficient produced by
the Frobenius recurrence in Remark~\ref{l3:rem:vanhove-comparison}(b).  In
particular the sign conventions of Convention~\ref{l3:conv:branches}, the
normalization $\gamma=H^2/2$, and the factor $-12$ relating $\mathcal F_\gamma$
to $K_2$ are all consistent.
\subsection{The Gamma-class boundary number}

We use the convention $\widehat\Gamma(E)=\prod_j\Gamma(1+\delta_j)$ for the
Chern roots $\delta_j$ of $E$, with no factors of $2\pi i$.  In $K^0(F_3)$,
\[
 [T_{X_3}|_{F_3}]=[T_{F_3}]+[K_{F_3}]=[T_{Y_3}|_{F_3}]-[\mathcal O_{F_3}(H)]+[\mathcal O_{F_3}(-H)],
\]
using $T_{Y_3}|_{F_3}=T_{F_3}+N_{F_3/Y_3}$ and $N_{F_3/Y_3}=\mathcal O_{F_3}(H)$, $K_{F_3}=\mathcal O_{F_3}(-H)$.
Since the Chern roots of $T_{Y_3}$ are $2p_1,\ldots,2p_4$,
\begin{equation}\label{l3:eq:gamma-class}
 \widehat\Gamma_{X_3}|_{F_3}=
 \frac{\Gamma(1-H)}{\Gamma(1+H)}\prod_{i=1}^4\Gamma(1+2p_i).
\end{equation}

\begin{proposition}[Numerical Gamma identity]\label{l3:prop:gamma}
$\log(\widehat\Gamma_{X_3}|_{F_3})=\tfrac23\zeta(3)H^3$ exactly, and
\[
 \int_{F_3}\bigl[\widehat\Gamma_{X_3}|_{F_3}\bigr]_6=\frac{2\zeta(3)}3\int_{F_3}H^3=16\zeta(3).
\]
\end{proposition}
\begin{proof}
From $\log\Gamma(1+x)=-\gamma_Ex+\sum_{k\ge2}(-1)^k\zeta(k)x^k/k$ and
$p_i^2=0$ we get $\sum_i\log\Gamma(1+2p_i)=-2\gamma_EH$, while
\[
 \log\Gamma(1-H)-\log\Gamma(1+H)
 =2\gamma_EH+\sum_{k\ge2}\frac{\zeta(k)}{k}\bigl(1-(-1)^k\bigr)H^k
 =2\gamma_EH+\frac{2\zeta(3)}{3}H^3,
\]
because the even-$k$ terms cancel and $H^k=0$ on $F_3$ for $k\ge4$.  The
$\gamma_E$-terms cancel between the two contributions, so $\log\widehat\Gamma_{X_3}|_{F_3}$
is purely of degree six; exponentiating, its degree-six part is itself.  Now
$\int_{F_3}H^3=24$.
\end{proof}

This is the absolute value of the boundary coefficient in~\eqref{l3:eq:frob-form}.
With the normalization $R'=-\pi$ used in this work, it is the negative of the coefficient
of  the Feynman integral.
The Mellin computation and the direct Gamma calculation are independent, and
Proposition~\ref{l3:prop:gamma} does not identify the integral lattice
underlying the Feynman Frobenius basis with the Gamma-class integral structure.
We do not need such a statement.

\subsection{Analytic domain of the three-loop identity}

Theorem~\ref{l3:thm:main} is an identity of germs at $s=\infty$, together with
its analytic continuation along paths; the two sides of~\eqref{l3:eq:main} are
solutions of $\mathcal L_3u=-24$ and continue along any path in $|s|>16$, where
$\mathcal L_3$ has no singular point.  The following explicit statements concern
the \emph{series} on the right of~\eqref{l3:eq:main}, which need not converge on
all of $|s|>16$.

By Convention~\ref{l3:conv:branches}, $Q(s)=-s^{-1}e^{\sigma(s)}$ is single-valued
and analytic on $|s|>16$, whereas $R=\log Q$ is determined only up to
$2\pi i\Z$ and is single-valued on the universal cover of $\{|s|>16\}$.  From
$A_n\le16^n$,
\[
 |Q(s)|\le\frac1{|s|}\exp\!\left(\sum_{n\ge1}\frac{16^n}{n|s|^n}\right)
 =\frac1{|s|-16}\qquad(|s|>16).
\]
Let $\rho>0$ be the common radius of convergence of $f_2$ and $\mathcal F_\gamma$
(Proposition~\ref{l3:prop:potential} and Theorem~\ref{l3:thm:coefficient-bridge}; the
two series differ by the operation $\thetaQ$ and a sign, which do not change
the radius).  For $M>16+\rho^{-1}$ and $|s|>M$ one has $|Q(s)|<\rho$, and the
identity
\begin{equation}\label{l3:eq:exterior-identity}
 f_2(Q(s))=-(\thetaQ\mathcal F_\gamma)(Q(s))
\end{equation}
holds as an equality of \emph{single-valued} analytic functions on
$\{|s|>M\}$, since only $Q$ and not $R$ enters.  Likewise, if $\rho'>0$ is a
common radius of convergence of $f_0,f_1,f_2$ (which exists by
Lemma~\ref{l3:lem:frob-form}) and $M'>16+\rho'^{-1}$, then the full formula
\eqref{l3:eq:main} holds on the universal cover of $\{|s|>M'\}$ with all series
convergent.  These estimates say nothing about continuation through the singular points $s=4,16$ of $\mathcal L_3$.

\Icorref{cor:Qradius} of the companion paper makes these bounds explicit: all the three-loop descendant series are holomorphic for
$|Q|<1/12$. Thus $|s|>28$ is a sufficient exterior domain for the convergent
series identities above, using $|Q(s)|\leq1/(|s|-16)$. The analogous four-loop
estimate is $|s|>45$, using $1/20$ and $25$. These domains are sufficient, not optimal, and concern the Taylor series in $Q$. Separately,~\Ipropref{prop:sTaylor} of the
companion paper proves that the Taylor series of the physical Feynman integral
in $s$ has exact radius $16$ at three loops and $25$ at four loops; in both
cases the Feynman integral is continuous but not $C^1$ at the physical
threshold.

\medskip\noindent\textbf{Additional real-axis continuation from the companion
paper~\cite{PaperI}.} By~\Ithmref{thm:torus} and
\Icorref{cor:realcontinuation}, $A(y)>0$ and the real mirror coordinate
$Q(y)=ye^{S(y)}$ is strictly increasing for $y>0$. Its range is $(0,Q_*)$,
where
\[
Q_*=\exp\left(-\int_{(S^1)^4}\log|w_1+w_2+w_3+w_4|^2\,d\nu\right)
 =e^{-7\zeta(3)/\pi^2}=0.4263\ldots
\]
for normalized Haar measure $\nu$. The closed form is the Mahler-measure
evaluation of Section~\ref{ext:lowend}.
The three descendant germs have the
specified real analytic continuation along this interval, and the complete
formula of Theorem~\ref{l3:thm:full-descendant-assembly} holds for all real
$s<0$ with these continuations. This extends the exterior-domain statement; Taylor convergence on the whole interval is not claimed.

\section{Three-loop integrality for arbitrary integral insertions, and an
explicit quantum-\texorpdfstring{$K$}{K} formula}\label{app:arith3}\label{l3:sec:bps}

\Icorref{cor:primaryintegral} states primary integrality for
$\gamma_L=\delta_0=H^{L-1}/(L-1)!$ in every $L\geq3$. Here we give the simpler
fourfold projective-bundle argument for every integral primary class $\eta\in
H^4(F_3,\Z)$ and specialize the quantum-$K$ formula. The dimension restriction
concerns this compactification's semipositivity argument, not the scope of
primary integrality on the compact Calabi--Yau cover of the companion
paper~\cite{PaperI}. The notation is that of Section~\ref{app:l3notation}.

For $\eta\in H^4(F_3,\Z)$ define classwise numbers $b_\beta(\eta)$ by
\begin{equation}\label{l3:eq:classwise-bps}
 N_\beta(\eta)=\sum_{\substack{k\ge1,\ \alpha\\ \beta=k\alpha}}
 \frac{b_\alpha(\eta)}{k^2},
\end{equation}
and aggregate numbers $n_m(\gamma)$ by
\begin{equation}\label{l3:eq:dilog}
 \mathcal F_\gamma(Q)=\sum_{m\ge1}n_m(\gamma)\Li_2(Q^m),
 \qquad \Li_2(q)=\sum_{k\ge1}\frac{q^k}{k^2}.
\end{equation}
Here $N_\beta(\eta)$ denotes~\eqref{l3:eq:potential-intro} with $\gamma$ replaced
by $\eta$.  The classwise sums range over integral roots in $H_2(F_3,\Z)$;
we assign zero to $N_\alpha$ and $b_\alpha$ on noneffective classes of positive
$H$-degree.  This convention makes all root sums unambiguous.
Both transforms are triangular with unit diagonal, hence uniquely solvable over
$\Q$; no convergence is needed for this formal step.  Comparing coefficients,
\begin{align}
 N_m(\gamma)&=\sum_{k\mid m}\frac{n_{m/k}(\gamma)}{k^2},\label{l3:eq:cover}\\
 n_m(\gamma)&=\sum_{k\mid m}\frac{\mu(k)}{k^2}N_{m/k}(\gamma),\label{l3:eq:mobius}
\end{align}
the second by multiplying~\eqref{l3:eq:cover} by $m^2$, putting $d=m/k$ to get
$m^2N_m=\sum_{d\mid m}d^2n_d$, and applying ordinary M\"obius inversion to
$G(m)=m^2N_m$, $F(d)=d^2n_d$.

\begin{theorem}[Genus-zero Gopakumar--Vafa integrality, one insertion; quoted]\label{l3:thm:chou-input}
Let $Z'$ be a smooth projective complex fourfold which is semi-positive, and
let $\Gamma\in H^4(Z',\Z)$.  For an effective class $A$ with
$K_{Z'}\cdot A=0$ define $B_A(\Gamma)$ by
\[
 \operatorname{GW}^{Z'}_{0,A}(\Gamma)
 =\sum_{\substack{k\ge1,\ A'\\ A=kA'}}\frac{B_{A'}(\Gamma)}{k^{2}},
 \qquad
 \operatorname{GW}^{Z'}_{0,A}(\Gamma)
 =\int_{[\overline{\mathcal M}_{0,1}(Z',A)]^{\rm vir}}\ev^*\Gamma .
\]
Then $B_A(\Gamma)\in\Z$.
\end{theorem}

\begin{remark}\label{l3:rem:chou}
This is the one-insertion, fourfold, $K$-degree-zero case of
\cite[Thm.~1.1]{Chou}, which restates~\cite[Thm.~9.2]{IP}. For a
semi-positive variety of dimension $\ge3$ and integral even insertions
$\Gamma_1,\dots,\Gamma_n$, the genus-zero Gopakumar--Vafa type invariants
defined through the Klemm--Pandharipande multiple-cover weight
$k^{-(3-n)}$ in the $K$-degree-zero sector~\cite{KP} are
integers; for $n=1$ the weight is $k^{-2}$.  Chou's proof is algebraic: it
compares the invariants with integral quantum $K$-theoretic Euler
characteristics through virtual Kawasaki--Riemann--Roch, semi-positivity being
what bounds the virtual dimensions of the nontrivial Kawasaki strata; the
one-insertion computation there~\cite[Prop.~3.4]{Chou} concerns insertions of
degree $2(\dim Z'-K_{Z'}\cdot A-2)$, which for $\dim Z'=4$ and $K_{Z'}\cdot A=0$
is exactly $H^4(Z')$, the degree used below.  The semi-positivity hypothesis
is satisfied by $Z$ in either of the usual senses, namely the algebraic one, that
$c_1(TZ)\cdot A\ge0$ for every effective curve class, holds because $-K_Z$ is
nef (Lemma~\ref{l3:lem:compactification}); and the symplectic one of
\cite{IP}, which in real dimension eight forbids classes $A$ with
$\omega(A)>0$ and $c_1(A)=-1$, holds because $c_1(TZ)=2\xi$ takes only even
values on integral classes.
\end{remark}

Put
\[
 Z=\mathbb P_{\rm lines}(\mathcal O_{F_3}\oplus K_{F_3}),\qquad
 \varpi_Z:Z\to F_3,\qquad \xi=c_1(\mathcal O_Z(1)),
\]
let $D_0,D_\infty$ be the sections corresponding to the summands $\mathcal O_{F_3}$
and $K_{F_3}$, and let $i:F_3=D_0\hookrightarrow Z$.

\begin{lemma}[Compactification geometry]\label{l3:lem:compactification}
The variety $Z$ is a smooth projective fourfold, being a $\mathbb P^1$-bundle
over the smooth projective threefold $F_3$. Moreover
\[
 Z\setminus D_\infty\cong X_3,\qquad [D_\infty]=\xi,\qquad i^*\xi=0,\qquad
 \mathcal O_Z(D_\infty)|_{D_\infty}\cong K_{F_3}^{-1},\qquad c_1(TZ)=2\xi,
\]
and $\xi$, hence $-K_Z$, is nef.
\end{lemma}

\begin{lemma}[No maps at infinity]\label{l3:lem:no-infinity}
Every connected genus-zero stable map to $Z$ of nonzero class $i_*\beta$ has
image in $Z\setminus D_\infty=X_3$, and
\begin{equation}\label{l3:eq:compact-local-moduli}
 \overline{\mathcal M}_{0,1}(X_3,\beta)\cong\overline{\mathcal M}_{0,1}(Z,i_*\beta)
\end{equation}
as stacks, with identical obstruction theories and virtual classes.
\end{lemma}

\begin{lemma}[Curve classes]\label{l3:lem:lattice}
$H_2(Z,\Z)=i_*H_2(F_3,\Z)\oplus\Z[\ell]$ with $\xi\cdot\ell=1$.  If $kA=i_*\beta$
in $H_2(Z,\Z)$ with $k\ge1$, then $A=i_*\alpha$ with $k\alpha=\beta$; and $i_*\alpha$ is
effective if and only if $\alpha$ is.
\end{lemma}

\begin{theorem}[BPS integrality]\label{l3:thm:bps-integrality}
For every $\eta\in H^4(F_3,\Z)$ and every nonzero effective $\beta$ one has
$b_\beta(\eta)\in\Z$.  Consequently $n_m(\gamma)\in\Z$ for every $m\ge1$.
\end{theorem}

Combining~\eqref{l3:eq:coeff-main},~\eqref{l3:eq:cover} and
Theorem~\ref{l3:thm:bps-integrality},
\begin{equation}\label{l3:eq:d2-bps}
 d_2(m)=-m\,N_m(\gamma)=-m\sum_{k\mid m}\frac{n_{m/k}(\gamma)}{k^2},
 \qquad n_m(\gamma)\in\Z,
\end{equation}
and, since $\thetaQ\Li_2(Q^m)=-m\log(1-Q^m)$, the formal Euler product
\begin{equation}\label{l3:eq:euler-product}
 \exp(f_2(Q))=\prod_{m\ge1}(1-Q^m)^{m\,n_m(\gamma)}\in1+Q\,\Z[[Q]].
\end{equation}
Putting $A_m=m^2N_m(\gamma)=-m\,d_2(m)$ gives $A_m=\sum_{d\mid m}d^2n_d(\gamma)
\in\Z$ and, for $p$ prime, $a\ge1$ and $p\nmid e$,
\begin{equation}\label{l3:eq:supercongruence}
 A_{p^ae}-A_{p^{a-1}e}\equiv0\pmod{p^{2a}},
\end{equation}
since subtracting the two divisor sums leaves exactly the divisors of $p^ae$ of
$p$-adic valuation $a$, each contributing $p^{2a}d'^2n_{p^ad'}$.

Table~\ref{l3:tab:coefficients} records the exact values through $m=19$.
The Gromov--Witten column is obtained from~\eqref{l3:eq:F-K}, the sunset column
can also be obtained independently from~\eqref{l3:eq:frob-rec}, and the BPS
column is the M\"obius transform~\eqref{l3:eq:mobius} of the Gromov--Witten column.
All nineteen rows agree by exact rational arithmetic. The table is a consistency check and is not used in the proofs.

\begingroup
\scriptsize
\setlength{\tabcolsep}{4pt}
\begin{longtable}{@{}rlll@{}}
\caption{Sunset coefficients, Gromov--Witten aggregates, and BPS transforms.}
\label{l3:tab:coefficients}\\
\toprule
$m$&$d_2(m)$&$N_m(\gamma)=-d_2(m)/m$&$n_m(\gamma)$\\
\midrule
\endfirsthead
\toprule
$m$&$d_2(m)$&$N_m(\gamma)=-d_2(m)/m$&$n_m(\gamma)$\\
\midrule
\endhead
1&$-24$&$24$&$24$\\
2&$132$&$-66$&$-72$\\
3&$-800$&$800/3$&$264$\\
4&$6018$&$-3009/2$&$-1488$\\
5&$-242424/5$&$242424/25$&$9696$\\
6&$415088$&$-207544/3$&$-69240$\\
7&$-25875552/7$&$25875552/49$&$528072$\\
8&$33914049$&$-33914049/8$&$-4238880$\\
9&$-954894248/3$&$954894248/27$&$35366424$\\
10&$15208750932/5$&$-7604375466/25$&$-304177440$\\
11&$-324462499968/11$&$324462499968/121$&$2681508264$\\
12&$289525458104$&$-72381364526/3$&$-24127104048$\\
13&$-37323117924480/13$&$37323117924480/169$&$220846851624$\\
14&$201031033602384/7$&$-100515516801192/49$&$-2051337209592$\\
15&$-289441832285280$&$19296122152352$&$19296122151264$\\
16&$5872668577670337/2$&$-5872668577670337/32$&$-183520891992384$\\
17&$-509335737033576096/17$&$509335737033576096/289$&$1762407394579848$\\
18&$921823061292056348/3$&$-460911530646028174/27$&$-17070797440168296$\\
19&$-60149764030617173760/19$&$60149764030617173760/361$&$166619844960158376$\\
\bottomrule
\end{longtable}
\endgroup

\subsection{Proofs of the supporting statements}

\subsubsection*{Compactification and curve-by-curve nefness}

\begin{proof}
A line in $\mathcal O_x\oplus K_{F_3,x}$ with nonzero first coordinate is spanned
by a unique $(1,v)$, which gives $Z\setminus D_\infty\cong X_3$; the projection of
the tautological subbundle to $\mathcal O_{F_3}$ is a section of $\mathcal O_Z(1)$
vanishing simply along $D_\infty$, which gives $[D_\infty]=\xi$, $i^*\xi=0$ and
the normal bundle $N_{D_\infty/Z}\cong K_{F_3}^{-1}$.  The relative Euler sequence
gives $c_1(T_{Z/F_3})=2\xi+\varpi_Z^*c_1(K_{F_3})$; adding
$\varpi_Z^*c_1(TF_3)=-\varpi_Z^*c_1(K_{F_3})$ gives $c_1(TZ)=2\xi$.  For nefness let
$C\subset Z$ be an integral curve.  If $C\not\subset D_\infty$ then
$\xi\cdot C=D_\infty\cdot C\ge0$.  If $C\subset D_\infty$ then
$\xi\cdot C=\deg(N_{D_\infty/Z}|_C)=-K_{F_3}\cdot\varpi_{Z*}C=H\cdot\varpi_{Z*}C>0$,
as $H$ is ample on $F_3$.
\end{proof}

\subsubsection*{Excluding the infinity section in families}

\begin{proof}
Since $i^*\xi=0$, the total intersection is $\xi\cdot i_*\beta=0$.  By the proof
of Lemma~\ref{l3:lem:compactification} each nonconstant component has nonnegative
$\xi$-degree, strictly positive if contained in $D_\infty$; as the total is
zero, no nonconstant component lies in $D_\infty$ and each has $\xi$-degree $0$,
hence (all local intersection multiplicities being nonnegative) meets $D_\infty$
nowhere.  A contracted component mapping into $D_\infty$ would be joined through
a chain of contracted components to a nonconstant one, which would then meet
$D_\infty$; contradiction.

For the stack statement, $X_3\hookrightarrow Z$ is an open immersion, so
$\overline{\mathcal M}_{0,1}(X_3,\beta)$ is the open substack of
$\overline{\mathcal M}_{0,1}(Z,i_*\beta)$ on which the universal image avoids
$D_\infty$.  In a family, the locus where the image meets $D_\infty$ is closed,
being the image under the proper map $\varrho$ of the closed substack
$\mu^{-1}(D_\infty)$ of the universal curve.  The preceding paragraph says this
closed substack has no geometric points, hence is empty.  This gives an isomorphism of stacks, not only a bijection on geometric points.  Finally $T_Z|_{X_3}=T_{X_3}$, so the relative tangent-obstruction
complexes and their maps to the cotangent complex agree, and so do the virtual
classes; properness is that of $\overline{\mathcal M}_{0,1}(Z,i_*\beta)$.
\end{proof}

\subsubsection*{Integral covering roots and effectivity}

\begin{proof}
The splitting is the projective-bundle formula together with the section $i$,
and $H_2(F_3,\Z)\cong\Z^4$ by~\Ipropref{prop:geometry}.  If $kA=i_*\beta$,
intersecting with $\xi$ gives $k\,\xi\cdot A=0$, so the fiber coefficient of $A$
vanishes and $A=i_*\alpha$; applying $\varpi_{Z*}$ gives $k\alpha=\beta$, and
$H_2(F_3,\Z)$ is torsion-free.  If $\alpha$ is effective so is $i_*\alpha$.
Conversely, an effective cycle representing $i_*\alpha$ has total
$D_\infty$-intersection $0$, so by the proof of Lemma~\ref{l3:lem:no-infinity} each
of its components avoids $D_\infty$, and its image under $\varpi_Z$ is an
effective cycle of class $\alpha$.
\end{proof}

\subsubsection*{Fourfold semipositivity and arbitrary integral insertions}

\begin{proof}
By Lemma~\ref{l3:lem:compactification} and Remark~\ref{l3:rem:chou}, $Z$ is a
smooth projective semi-positive fourfold, and
$c_1(TZ)\cdot i_*\beta=2\,\xi\cdot i_*\beta=0$, so $i_*\beta$ lies in the
$K_Z$-degree-zero sector.  The class $\varpi_Z^*\eta\in H^4(Z,\Z)$ is
integral, and
\[
 \vdim_\C\overline{\mathcal M}_{0,1}(Z,i_*\beta)=(4-3)+c_1(TZ)\cdot i_*\beta+1=2,
\]
so $\operatorname{GW}^Z_{0,i_*\beta}(\varpi_Z^*\eta)$ is a well-defined number
and Theorem~\ref{l3:thm:chou-input} applies.  By
Lemma~\ref{l3:lem:no-infinity} and $\varpi_Z\circ i=\mathrm{id}_{F_3}$,
$\operatorname{GW}^Z_{0,i_*\beta}(\varpi_Z^*\eta)=N_\beta(\eta)$; by
Lemma~\ref{l3:lem:lattice} the $k$-th roots of $i_*\beta$ in $H_2(Z,\Z)$ are
exactly the images of the $k$-th roots of $\beta$ in $H_2(F_3,\Z)$, with the same
effectivity.  Hence the two multiple-cover transforms~\eqref{l3:eq:classwise-bps}
and that of Theorem~\ref{l3:thm:chou-input} coincide, and by uniqueness of the
triangular transform $b_\beta(\eta)=B_{i_*\beta}(\varpi_Z^*\eta)\in\Z$.

For fixed $m$, the nef classes $p_i$ give nonnegative integer coordinates for
effective $\beta$, summing to $m$; the curve-lattice identification of~\Ipropref{prop:geometry} therefore gives only finitely many such
classes.  Summing~\eqref{l3:eq:classwise-bps} over them and using
$|k\alpha|=k|\alpha|$ gives $N_m(\gamma)=\sum_{k\mid
m}k^{-2}\sum_{|\alpha|=m/k}b_\alpha(\gamma)$; comparing with~\eqref{l3:eq:cover}
and using uniqueness again, $n_m(\gamma)=\sum_{|\beta|=m}b_\beta(\gamma)\in\Z$.
\end{proof}

We now give a quantum-$K$ refinement for the same
$F_3$, $X_3=\operatorname{Tot}(K_{F_3})$, and insertion $\gamma=H^2/2$.
It is independent of the analytic--enumerative bridge, but uses the same
compactification and the quantum-$K$ theorem underlying Chou's integrality
proof.  It is therefore a refinement of that argument rather than an independent proof.

\subsection{Compactification and identification of invariants}
Use the lines convention and notation of Section~\ref{l3:sec:bps}:
\[
 Z=\mathbb P_{\rm lines}(\mathcal O_{F_3}\oplus K_{F_3}),\qquad
 \varpi_Z:Z\longrightarrow F_3,\qquad i:F_3=D_0\hookrightarrow Z.
\]
Lemmas~\ref{l3:lem:compactification}--\ref{l3:lem:lattice} establish, respectively,
$-K_Z=2\xi$ nef, the isomorphism of stable-map stacks
\[
 \overline{\mathcal M}_{0,1}(X_3,\beta)
 \simeq \overline{\mathcal M}_{0,1}(Z,i_*\beta),
\]
with their standard obstruction theories, and the equality of effective
integral root sets.  The stack isomorphism is induced by an open immersion
and holds over arbitrary bases. In particular
\begin{equation}\label{l3:eq:qk-comparison}
 N_\beta(\gamma)=\operatorname{GW}^{Z}_{0,i_*\beta}(\varpi_Z^*\gamma),
 \qquad b_\beta(\gamma)=B_{i_*\beta}(\varpi_Z^*\gamma).
\end{equation}
The second equality follows from the first by uniqueness of the classwise
multiple-cover transform, with the zero convention for noneffective roots.
None of these statements uses a mirror theorem.

\subsection{An integral lift with exact Chern character}
Set $L_i=\varpi_Z^*\mathcal O_{F_3}(p_i)$ and define
\begin{equation}\label{l3:eq:qk-lift}
 \mathcal E_\gamma=\sum_{1\le i<j\le4}
       (1-[L_i^{-1}])(1-[L_j^{-1}])\in K^0_{\rm alg}(Z).
\end{equation}
Here $K^0_{\rm alg}$ is the Grothendieck ring of algebraic vector bundles;
its unit is the trivial line bundle.  The lift is an integral virtual class,
not a vector bundle of positive rank.

\begin{lemma}\label{l3:lem:qk-ch}
In rational even cohomology,
$\operatorname{ch}(\mathcal E_\gamma)=\varpi_Z^*\gamma$ exactly.
\end{lemma}

For $\alpha\ne0$ effective define the true quantum-$K$ invariant
\begin{equation}\label{l3:eq:qk-kappa}
 \kappa_\alpha=
 \chi\left(\overline{\mathcal M}_{0,1}(Z,i_*\alpha),
       \ev_1^*\mathcal E_\gamma\otimes\mathcal O^{\rm vir}\right)\in\Z,
\end{equation}
where $\mathcal O^{\rm vir}$ is the virtual structure sheaf for the standard
stable-map obstruction theory~\cite{Lee}.  Put $\kappa_0=0$ and put
$\kappa_\alpha=0$ for noneffective classes of positive $H$-degree.
The Euler characteristic is taken on the proper Deligne--Mumford stack.
In characteristic zero its coherent cohomology groups are finite-dimensional
and vanish in sufficiently high degree. The tame coarse-space morphism has
exact pushforward, and the coarse space is proper.  Taking their alternating
sum, extended additively to integral virtual coherent classes, gives an integer.

\begin{remark}[True versus fake quantum $K$-theory]
Equation~\eqref{l3:eq:qk-kappa} is not obtained by integrating just
$\operatorname{ch}(\ev_1^*\mathcal E_\gamma)
\operatorname{Td}(T^{\rm vir})$ against the virtual class of the stack.
That integral is the identity-sector, or fake, invariant.  The true Euler
characteristic includes the other inertia-sector contributions in virtual
Kawasaki--Riemann--Roch; these account for the multiple-cover correction in
Chou's result.
\end{remark}

\subsection{The explicit formula}
\begin{theorem}[Quantum-$K$ formula]\label{l3:thm:qk-strengthening}
For every nonzero effective $\beta\in H_2(F_3,\Z)$,
\begin{equation}\label{l3:eq:qk-formula}
 {
 b_\beta(\gamma)=\sum_{\substack{k\ge1,\ \alpha\in H_2(F_3,\Z)\\k\alpha=\beta}}
                         \mu(k)k\,\kappa_\alpha.}
\end{equation}
In particular $b_\beta(\gamma)$ is integral.  Equivalently,
\begin{equation}\label{l3:eq:qk-inverse}
 \kappa_\beta=\sum_{k\alpha=\beta}k\,b_\alpha(\gamma).
\end{equation}
\end{theorem}

Put $\kappa_m=\sum_{|\alpha|=m}\kappa_\alpha$.  Finiteness of each sum and
Theorem~\ref{l3:thm:qk-strengthening} give
\[
 n_m(\gamma)=\sum_{k\mid m}\mu(k)k\,\kappa_{m/k},\qquad
 \kappa_m=\sum_{k\mid m}k\,n_{m/k}(\gamma).
\]
For example the BPS values of Table~\ref{l3:tab:coefficients} give
$\kappa_1=24$, $\kappa_2=-24$, $\kappa_3=336$, and $\kappa_4=-1536$.
These values follow from the formula; they were not computed independently by localization. The identification of the
quantum-$K$ BPS transforms with the BPS coefficients extracted from the
sunset expansion is provided by the coefficient bridge
Theorem~\ref{l3:thm:coefficient-bridge}, in particular
$$
f_2(Q)=-\theta_Q\mathcal F_\gamma(Q).
$$
Thus Theorem~\ref{l3:thm:qk-strengthening} supplies a quantum-$K$ realization
of the BPS transforms already identified with the sunset coefficients by the
coefficient bridge.
\subsection{Proofs for the quantum-\texorpdfstring{$K$}{K} specialization}

\subsubsection*{The exact integral Chern-character lift}

\begin{proof}
The relation $p_i^2=0$ on $(\mathbb P^1)^4$ remains true after restriction
and pullback. 
Hence
$$ \operatorname{ch}(\mathcal E_\gamma) =\sum_{i<j}\varpi_Z^*(p_i p_j) =\varpi_Z^*\gamma. $$
Thus \(\mathcal E_\gamma\) is an integral \(K\)-theory lift
of the degree-four cohomology class $\varpi_Z^*\gamma$ appearing in the
one-insertion specialization of Chou's formula~\cite[Prop.~3.4]{Chou}.

\end{proof}

\subsubsection*{The explicit fourfold specialization and inversion}

\begin{proof}
We apply~\cite[Prop.~3.4]{Chou},  with target $Z$,
class $A=i_*\beta$, cohomology insertion $\varpi_Z^*\gamma$, and its exact
$K$-lift~\eqref{l3:eq:qk-lift}.  The hypotheses are checked as follows.
The target is a smooth projective fourfold, with \(c_1(TZ)=2\xi\), which is nef, and  \(c_1(TZ)\cdot A=0\).
 (It also satisfies symplectic semipositivity since its
$c_1$ takes even values on integral classes, as explained in
Remark~\ref{l3:rem:chou}.)  The insertion
has real cohomological degree $ \deg_{\mathbb{R}}(\varpi_Z^*\gamma)=4$.
For the fourfold $Z$ and $c_1(TZ)\cdot A=0$, one has $\vdim_\C\overline{\mathcal M}_{0,1}(Z,A)=(\dim_\C Z-3)+c_1(TZ)\cdot A+1=2$. Since $\deg_{\mathbb R}(\varpi_Z^*\gamma)=4$, the insertion has complex codimension~$2$, so the resulting invariant is numerical.
Lemma~\ref{l3:lem:qk-ch} supplies an integral \(K\)-theory lift of Chou's cohomological insertion, with no rational denominators or unwanted higher-degree components.
His $K$-degree-zero formula is precisely
\[
 B_A(\varpi_Z^*\gamma)
   =\sum_{kA'=A}\mu(k)k\,
     \chi\left(\overline{\mathcal M}_{0,1}(Z,A'),
       \ev_1^*\mathcal E_\gamma\otimes\mathcal O^{\rm vir}\right).
\]
The lattice and effectivity statement of Lemma~\ref{l3:lem:lattice} identifies
$A'=i_*\alpha$ in every term;~\eqref{l3:eq:qk-comparison} then proves
\eqref{l3:eq:qk-formula}.  The sum is finite because $k$ divides
$H\cdot\beta>0$.  Its terms and coefficients are integral.
Finally, the arithmetic functions $k\mapsto k$ and $k\mapsto\mu(k)k$ are
Dirichlet-convolution inverses, which gives~\eqref{l3:eq:qk-inverse}.
\end{proof}

\section{Two, three and four loops, and the local--relative comparison}
\label{app:loops}
This section collects the comparisons with the adjacent loop orders. The
first part is the two-loop $L=2$ comparison which is not covered by the hypotheses
used in this paper, and the point of the comparison is to separate the cover
exponent from the lower-loop restrictions on curve lattices and on
threshold convergence. The second part records the $L=3$ and $L=4$
normalizations against the published coefficient tables. The third records
the Euler factors relating the local invariants to relative ones.

\subsection{The two-loop threefold with zero or one primary marking}
\label{app:L2comparison}

In this section $r$ denotes the degree of a multiple cover, not the loop
number or the number of insertions. At two loops the sunset integral is
related to the prepotential of the local Calabi--Yau threefold~\cite{BKV},
whose multiple-cover formula has weight $r^{-3}$, whereas here we use the
weight $r^{-2}$. The two weights belong to different observables rather than
to different genus-zero cover laws: for primary invariants with $\ell$ marked
insertions the weight is $r^{\ell-3}$ (see~\cite[equations (2) and (5)]{KP}
and~\Ithmref{thm:IP}), which gives $r^{-3}$ for the unmarked prepotential
($\ell=0$) and $r^{-2}$ for the one-point invariants ($\ell=1$).

Let $X_2$ be a smooth Calabi--Yau threefold with proper positive-degree
stable-map moduli, and let $H$ be an integral divisor. This includes the
local two-loop setting when its positive-degree local invariants are
formed using the actual curve lattice of the base surface. Only the
positive-degree part of the prepotential is considered. Define
\[
 U_\beta=\int_{[\overline{\mathcal M}_{0,0}(X_2,\beta)]^{\vir}}1,
 \qquad
 U_\beta=\sum_{r\alpha=\beta}\frac{\nu_\alpha}{r^3}.
\]
The unmarked virtual dimension is zero. The one-divisor invariant,
in contrast, is defined on the one-marked stack, of virtual dimension
one. The primary divisor equation gives the exact comparison
\begin{equation}\label{eq:L2markedcover}
\begin{split}
 P_{H,\beta}&:=\langle H\rangle_{0,1,\beta}
   =(H\cdot\beta)U_\beta\\
   &=\sum_{r\alpha=\beta}
               \frac{(H\cdot\alpha)\nu_\alpha}{r^2}.
\end{split}
\end{equation}
In the last equality $H\cdot\beta=r(H\cdot\alpha)$ for the indicated
covering root. Thus the \emph{same threefold} already has the one-point
$r^{-2}$ transform, with
\[
 b_\alpha(H)=(H\cdot\alpha)\nu_\alpha.
\]
This is not an identification of the inserted integers with the raw
unmarked numbers $\nu_\alpha$. The divisor identity itself does not
prove integrality of the unmarked transform; it compares the two
normalizations whenever those quantities are defined.

Equivalently, define the logarithmic Novikov derivative by
$\theta_H\mathbf Q^\beta=(H\cdot\beta)\mathbf Q^\beta$.
As positive-degree formal potentials one has
\begin{align}
 \Phi_0^+(\mathbf Q)
    &=\sum_{\alpha>0}\nu_\alpha\Li_3(\mathbf Q^\alpha),
       \label{eq:L2Li3}\\
 \theta_H\Phi_0^+(\mathbf Q)
    &=\sum_{\alpha>0}(H\cdot\alpha)\nu_\alpha
                             \Li_2(\mathbf Q^\alpha).
       \label{eq:L2Li2}
\end{align}
The class variable is the actual integral curve class. A later
one-variable specialization may aggregate classes of fixed $H$-degree,
but only after their lattice and the finiteness of that aggregation
have been specified. Under the one-variable specialization
$\mathbf Q^\beta\mapsto Q^{H\cdot\beta}$ the operator $\theta_H$ becomes
$\theta_Q=Q\,d/dQ$, which is the form used in~\eqref{l3:eq:two-loop-bridge}.

\subsection{Why the higher-loop insertion gives the same power}

For the local $(L+1)$-fold of this paper, genus-zero virtual dimension is
\[
 \operatorname{vdim}_{\C}\overline{\mathcal M}_{0,\ell}(X_L,\beta)
      =L-2+\ell.
\]
At $L=2$, the unmarked space has virtual dimension zero and supplies
an uninserted numerical invariant. At $L>2$, its virtual dimension is
positive; integrating the unit alone does not supply the numerical
curve count in question. We use one marking and
\[
 \delta_0=\frac{H^{L-1}}{(L-1)!},\qquad
 \operatorname{codim}_{\C}(\delta_0)=L-1
 =\operatorname{vdim}_{\C}\overline{\mathcal M}_{0,1}(X_L,\beta).
\]
Although $\delta_0$ has degree $L-1$, it is one insertion at one marking.
For instance, $H^2$ at a single marking and two divisor insertions at two
markings are different observables. Hence $\ell=1$, and the primary weight
is $r^{-2}$ for every $L\geq3$.

The extra factor of $r$ is visible geometrically for an isolated rational
curve in a threefold. A degree-$r$ cover contributes $r^{-3}$ through the
Aspinwall--Morrison obstruction integral, and a marked incidence condition
has $r$ preimages under the cover~\cite[\S1.1, equation (5)]{KP}. Equation~\eqref{eq:L2markedcover} makes
this comparison exact, without counting automorphisms. For moving curves in
higher dimensions, incidence conditions on the image family account for the
codimension of the insertion, not for additional markings. The integrality
proof does not rely on any such local model.

\subsection{What the comparison does not identify}

The two-loop sunset Feynman integral is related to the local prepotential by
the mirror-symmetry and Legendre-transform construction of
Bloch--Kerr--Vanhove~\cite[introduction, equations (1.5)--(1.6)]{BKV}.
The Feynman integral itself is not being identified with $\Phi_0^+$, nor does
eq.~\eqref{eq:L2Li2} say that higher-loop Feynman integrals are obtained by simply
differentiating the two-loop Feynman integral. The comparison concerns primary
potentials. It supplies no multiple-cover law for the descendants
$\psi^a\delta_a$, $a>0$, appearing in the full assembly
\Ieqref{eq:allloop}.

There are also reasons why $L=2$ cannot simply be substituted into every theorem of this paper. The weak-Lefschetz proof of
$H_2(F_L,\Z)\cong H_2((\PP^1)^{L+1},\Z)$ requires $2<L$.
For a surface at $L=2$ it does not give that isomorphism; its actual
curve lattice must be used instead. Independently, at the physical
threshold $s=(L+1)^2$ the local simplex density is proportional to
$\|u-u_*\|^{-2}$ in real dimension $L$. At $L=2$ its radial integral
is proportional to $\int_0^\varepsilon dr/r$ and diverges logarithmically.
For $L\ge3$ it is integrable, as used in~\Ipropref{prop:analyticdomain}. These lattice and endpoint
restrictions are not the origin of the primary exponent change.

\begingroup
\small
\subsection{Curve classes and normalizations at two, three and four loops}

For the purposes of this section set
\[
 Y_L=(\mathbb P^1)^{L+1},\quad
 F_L\in|\mathcal O_{Y_L}(1,\ldots,1)|\text{ smooth},\quad
 X_L=\Tot(K_{F_L}),\quad \gamma_L=\frac{H^{L-1}}{(L-1)!}.
\]
Adjunction gives $K_{F_L}=-H$ and $\int_{F_L}H^L=(L+1)!$.
Moreover $H^k/k!=e_k(p_1,\ldots,p_{L+1})$, so $\gamma_L$ is integral.
Its codimension is $L-1=\vdim\overline{\mathcal M}_{0,1}(X_L,\beta)$.
Here and below all curve classes have positive $H$-degree.

\subsection{The surface case and its curve classes}

The surface $F_2$ is a del Pezzo surface of degree six, hence isomorphic to
$\mathrm{Bl}_3\mathbb P^2$ with the three points noncollinear.
The restricted ambient divisor classes span a rank-three sublattice,
whereas $H_2(F_2,\Z)$ has rank four. Consequently an ambient degree
$\mathbf d=(p_1\cdot\beta,p_2\cdot\beta,p_3\cdot\beta)$ need not specify
$\beta$. We use the coarsened invariants
\[
 N_{\mathbf d}(H)=\sum_{j_*\beta=\mathbf d}N_\beta(H),\qquad
 N_m(H)=\sum_{H\cdot\beta=m}N_\beta(H).
\]
The sums are finite because the effective cone of this del Pezzo surface is
rational polyhedral, generated by its six $(-1)$-curves, and $H$ is ample.
The zero locus on the ambient stable-map stack is the disjoint union
of the stacks indexed by all these lifts. The virtual comparison must be
applied to this union. Convexity of $\mathcal O(H)$, concavity of
$\mathcal O(-H)$ in positive degree, the arbitrary-base factorization,
and the compatibility of obstruction theories used in
Sections~\Iref{sec:local} and~\Iref{sec:mirror}
apply to each summand. The two-class modification and extraction then
compute the ambient numerical pairings of this coarsened theory; no
identification of the full curve lattice with the ambient lattice is used.

This is the sextic del Pezzo geometry of~\cite[Theorem 6.1 and \S6.5]{BKV}.
It is toric, whereas $F_3$ has $b_3=2$ by Remark~\ref{l3:rem:odd-cohomology}, and hence
is not a smooth projective toric variety. Thus the toric-target mirror
theorem applies directly to the surface, while at three loops we apply
it to $Y_3$ and use the two twists to obtain the local hypersurface theory.

\subsection{A coefficient calculation with fixed normalizations}

In this section use the same period sign at each loop order:
\begin{equation}\label{l3:eq:adjacent-conventions}
 A_L(y)=\sum_{m\geq0}(-1)^m c_m^{(L)}y^m,\quad
 c_m^{(L)}=\sum_{|\mathbf d|=m}\binom m{\mathbf d}^{\!2},\quad
 \pi_L(s)=-yA_L(y),\quad y=-s^{-1}.
\end{equation}
Set $g_L=\sum_{m\geq1}(-1)^mc_m^{(L)}y^m/m$,
$T_L=\log y+g_L$, $Q=ye^{g_L}$, and $\theta_y= y\,d/dy$.
Thus $\theta_y T_L=A_L$ and $dT_L/ds=-\pi_L$.
The coefficient denoted $K_2^{(L)}$ is the scalar multiplying
$H^2z^{-2}$ after multiplication of the local $I$-function by
$e^{-HT_L/z}$. The coefficient of $z^{-1}$ after this normalization is zero.
For $L=2,3,4$, numerical primary extraction gives
\begin{equation}\label{l3:eq:adjacent-primary}
 \mathcal F_{\gamma_L}=-L(L+1)K_2^{(L)}.
\end{equation}
For $L\geq3$ this is~\Ithmref{thm:extraction}; at $L=2$ use the coarsened
virtual comparison just described. Equivalently, pair the proper primary
evaluation cycle with $H$, before removing the equivariant parameter.
The factor is $-\int_{F_2}H^2=-6$, and all insertions in this pairing
are ambient. This does not require ambient surjectivity in degree two.

The displayed coefficients are computed by finite sums as follows.
Put $h_j=\sum_{r=1}^j1/r$, $h_0=0$, and
\[
 b_L(y)=\sum_{m\geq1}(-1)^my^m
 \sum_{|\mathbf d|=m}\binom m{\mathbf d}^{\!2}
 \left(2h_m-\frac{2}{L+1}\sum_{i=1}^n h_{d_i}\right).
\]
The coefficient identity of Lemma~\ref{l3:lem:qserre}, with $L+1$ factors,
gives
\begin{equation}\label{l3:eq:adjacent-algorithm}
 \theta_QK_2^{(L)}=\frac{b_L}{A_L}-g_L,\qquad
 \theta_y K_2^{(L)}=b_L-A_Lg_L,\qquad K_2^{(L)}(0)=0.
\end{equation}
These formulas follow by expanding the finite Euler products to first
order in the divisor classes and using permutation symmetry; in the
average each $p_i$ is replaced by $H/(L+1)$.
Integrating the second formula coefficientwise and substituting the
unique inverse of $Q=ye^{g_L}$ determines $K_2^{(L)}$ over $\Q$.
The multinomial bound $c_m^{(L)}\leq(L+1)^{2m}$ and the harmonic-number
bound give convergent germs. For $L=2$ this yields
\[
 K_2^{(2)}=-Q+\frac74Q^2-\frac{28}{9}Q^3+\frac{135}{16}Q^4
 -\frac{626}{25}Q^5+\frac{751}{9}Q^6+O(Q^7).
\]

The two-loop operator, acting on $G(y)=I^{(2)}_\circleddash(-1/y)/y$, is
\begin{equation}\label{l3:eq:two-loop-operator}
 P_2=\theta_y^2+y(10\theta_y^2+10\theta_y+3)+9y^2(\theta_y+1)^2,
 \qquad P_2G=6.
\end{equation}
If instead $I^{(2)}_\circleddash(s)=s^{-1}v(-s^{-1})$, then $v=-G$ and $P_2v=-6$.
For a direct normalization of its Frobenius solutions, set $a_0(\rho)=1$,
$a_{-1}(\rho)=0$ and, for $m\geq1$, impose
\[
 (m+\rho)^2a_m(\rho)
 +[10(m-1+\rho)^2+10(m-1+\rho)+3]a_{m-1}(\rho)
 +9(m-1+\rho)^2a_{m-2}(\rho)=0.
\]
Then $U(\rho,y)=y^\rho\sum_{m\geq0}a_m(\rho)y^m$ satisfies
$P_2U=\rho^2y^\rho$. Its first two jets are
$U(0,y)=A_2$ and $\partial_\rho U(0,y)=A_2\log y+b_2$;
these coefficient identities follow by substitution in the displayed
recurrence, or the Bessel symmetric-power calculation with three factors.
The boundary calculation below selects $G=3\partial_\rho^2U(0,y)$.
Collecting powers of $R=T_2$ therefore gives
\begin{equation}\label{l3:eq:two-loop-bridge}
 I^{(2)}_\circleddash(s)=\pi_2(s)(-3R^2+f_1^{(2)}(Q)R+f_0^{(2)}(Q)),\qquad
 f_1^{(2)}=-6\left(\frac{b_2}{A_2}-g_2\right)
          =\theta_Q\mathcal F_H.
\end{equation}
In particular this is an identity of germs, independently of any finite
coefficient check.

The divisor equation gives $N_m(H)=mN_m^{\rm loc}$. Define the ordinary
threefold multiple-cover coefficients by
$N_m^{\rm loc}=\sum_{k\mid m}n_{m/k}/k^3$. Exact rational expansion gives
\begin{center}
\begin{tabular}{@{}rrrr@{}}\toprule
$m$&$[Q^m]f_1^{(2)}$&$N_m(H)$&$n_m$\\\midrule
1&6&6&6\\
2&$-21$&$-21/2$&$-6$\\
3&56&$56/3$&6\\
4&$-405/2$&$-405/8$&$-12$\\
5&$3756/5$&$3756/25$&30\\
6&$-3004$&$-1502/3$&$-84$\\
7&$86442/7$&$86442/49$&252\\
8&$-209301/4$&$-209301/32$&$-816$\\\bottomrule
\end{tabular}
\end{center}
The next three values are $2790,-9930,36498$.
In degree one the six $(-1)$-curves each have local normal bundle
$\mathcal O(-1)\oplus\mathcal O(-1)$ and contribute one.
No all-degree integrality statement at two loops is proved here.

\subsection{Two-loop boundary normalization}\label{app:L2boundary}

The Schwinger integral converges for $\operatorname{Re}s<9$, but diverges
at $s=9$. Near the symmetric point of the two-dimensional simplex its
positive integrand is comparable to $|u-u_*|^{-2}$.
The asymptotic normalization at negative infinity is
\begin{equation}\label{l3:eq:two-loop-boundary}
 I^{(2)}_\circleddash(-t)=\frac{3\log^2t}{t}+O(t^{-2}(1+\log t)^2).
\end{equation}
For completeness this follows from the Mellin method of
Section~\ref{app:analytic3} without numerical boundary fitting.
With $M_2(\nu)=\int e^{-V}W^\nu\prod_{i=1}^3dx_i/x_i$ and
$\mathcal K_3(w)=\int_0^\infty x^{w-1}K_0(x)^3dx$,
\[
 M_2(\nu)=\frac{16\,4^{-\nu}}{\Gamma(\nu)}\mathcal K_3(2\nu),\qquad
 \int_0^\infty t^{z-1}I^{(2)}_\circleddash(-t)dt=\Gamma(z)M_2(1-z).
\]
For $c=\log2-\gamma_E$, the polar part of $\mathcal K_3(w)$ at zero is
$\sum_{j=0}^3\binom3j j!c^{3-j}w^{-j-1}$.
Consequently the polar part of $\Gamma(1-\nu)M_2(\nu)$ is $6\nu^{-3}$:
the linear exponential cancels, and
\[
 \frac{\Gamma(1-\nu)}{\Gamma(1+\nu)}e^{-2\gamma_E\nu}
 =1+\frac23\zeta(3)\nu^3+O(\nu^5).
\]
There are no $\nu^{-2}$ or $\nu^{-1}$ terms. The convergent regular-singular
expansion and the cutoff argument of Section~\ref{app:analytic3}
then give~\eqref{l3:eq:two-loop-boundary}. Thus both homogeneous connection
constants relative to $3\partial_\rho^2U(0,y)$ vanish.

The Gamma-class computation is consistent with this analytic result:
$c(TF_2)=1+H+H^2$ and $c(TX_2)|_{F_2}=(1+H+H^2)(1-H)=1$.
In general for a del Pezzo surface $S$,
\[
 \int_S c_2(T_{\Tot(K_S)}|_S)=\chi_{\rm top}(S)-K_S^2=12-2K_S^2.
\]
This is an equality of characteristic numbers. Its vanishing at degree
six does not assert the same boundary normalization for other local surfaces.

\subsection{Four loops}\label{l3:app:four-loop}

For $L=4$, $\int_{F_4}H^4=120$ and $\gamma_4=H^3/6$. The general results in
\Ithmref{thm:extraction}, \Ithmref{thm:PF} and~\Ithmref{thm:laplace} and
\Icorref{cor:nfive} apply since $F_4$ is smooth and ample in $Y_4$,
$L\geq3$, and the insertion is ambient. Full even-cohomology surjectivity is not needed; indeed $b_4(F_4)=28$, whereas the ambient image in
degree four has dimension 10, by~\Ipropref{prop:ambient}. The operator
acting on $I^{(4)}_\circleddash(-1/y)/y$ is the operator
$P_4=\theta_y^4+\sum_{j=1}^3y^jQ_j(\theta_y)$ of the introduction, in which
$Q_3(\theta_y)=225(\theta_y+1)^2(\theta_y+2)^2$, and
\[
 P_4\left(\frac{I^{(4)}_\circleddash(-1/y)}y\right)=120.
\]
Its finite nonzero singularities in the $s$ coordinate are $1,9,25$.
Formula~\eqref{l3:eq:adjacent-algorithm} gives
\[
 K_2^{(4)}=-3Q+\frac{45}{4}Q^2-\frac{211}{3}Q^3
 +\frac{10029}{16}Q^4-\frac{164028}{25}Q^5+O(Q^6),
\]
\[
 \mathcal F_{\gamma_4}=-20K_2^{(4)}
 =60Q-225Q^2+\frac{4220}{3}Q^3-\frac{50145}{4}Q^4
 +\frac{656112}{5}Q^5-1529195Q^6+O(Q^7).
\]
With the negative period $\pi_4$ of~\eqref{l3:eq:adjacent-conventions},
\begin{equation}\label{l3:eq:four-loop-full}
 I^{(4)}_\circleddash(s)=\pi_4(s)\left[-5R^4+\sum_{j=0}^3f_j^{(4)}(Q)R^j
 +80\zeta(3)\left(R-\frac1{20}\theta_Q\mathcal F_{\gamma_4}\right)\right],
\end{equation}
where the $f_j^{(4)}$ have rational coefficients and zero constant terms,
and
\[
 f_3^{(4)}=\theta_Q\mathcal F_{\gamma_4}
 =60Q-450Q^2+4220Q^3-50145Q^4+656112Q^5+O(Q^6).
\]
This is~\Icorref{cor:nfive} with $\pi_4=-\Pi_4$; the rational
coefficients here are the negatives of its $h_j$.
Thus with a period asymptotic to $1/s$, the leading rational coefficient
is $(-1)^{L+1}(L+1)$ and the next one satisfies
\begin{equation}\label{l3:eq:uniform-sign}
 f_{L-1}^{(L)}=(-1)^L\theta_Q\mathcal F_{\gamma_L}.
\end{equation}
For $L\geq3$ this is~\Ithmref{thm:assembly}; the two-loop case is
\eqref{l3:eq:two-loop-bridge}.

In particular the four-loop boundary contribution is nonzero:
\[
 I^{(4)}_\circleddash(-1/y)=y\bigl(5\log^4y-80\zeta(3)\log y\bigr)
       +O(y^2(1+|\log y|)^4).
\]
Although the degree-four part of $\widehat\Gamma_{X_4}|_{F_4}$ vanishes,
its degree-three part is $\frac23\zeta(3)H^3$. Multiplication by the
degree-one part $-H\log y$ of $y^{-H}$ contributes
$-\frac23\zeta(3)\int_{F_4}H^4\log y=-80\zeta(3)\log y$.
This also explains why the boundary term in~\eqref{l3:eq:four-loop-full}
contains a logarithmic homogeneous solution, rather than just a constant
multiple of the holomorphic period.

The dilogarithmic transform, defined by
$N_m(\gamma_4)=\sum_{k\mid m}b_{m/k}/k^2$, starts
\[
 b_m=60,-240,1400,-12480,131220,-1529520,19192260,-254513280,\ldots.
\]
These displayed values follow by rational M\"obius inversion of
\eqref{l3:eq:adjacent-algorithm}. The all-degree arithmetic statement for the distinguished insertion
is~\Icorref{cor:primaryintegral}, with its explicitly stated algebraic
quantum-$K$ input (see~\Ithmref{thm:chou}).

\subsection{Euler factors and the local--relative comparison}\label{l3:app:rel-loc}

For an anticanonical divisor class $D$ and positive integer
$k=D\cdot\beta$, the finite-product identity is
\[
 \prod_{a=0}^{k-1}(-D-az)=(-1)^kD\prod_{a=1}^{k-1}(D+az).
\]
It is an equality in the cohomology-valued coefficient ring.
Multiplication by the nilpotent class $D$ cannot be inverted there.
Moreover replacing $(k-1)!$ by $k!$ multiplies a scalar coefficient
by the integer $k$; it does not insert the cohomology class $D$.
The degree-zero term is specified separately by the normalization of
the generating function.

There is also a precise numerical correspondence. Choose a smooth $D\in|H|$ on $F_L$; it exists by Bertini's theorem, since
$|H|$ is the restriction of the very ample system $|\mathcal O(1,\ldots,1)|$
and is therefore basepoint free. The target $F_L$ is smooth projective, $D$ is
smooth and ample, $\mathcal O_{F_L}(-D)=K_{F_L}$, and $k=D\cdot\beta>0$. Therefore
\cite[Theorem 1.1 and Corollary 1.2]{vGGR} apply to $\Tot(\mathcal O_{F_L}(-D))=X_L$.
With one ordinary marking carrying the primary insertion $\gamma_L$
and one marking of maximal contact $k$, they give
\[
 R_\beta(\gamma_L)=(-1)^{k+1}kN_\beta(\gamma_L).
\]
Here $R_\beta(\gamma_L)$ is the genus-zero logarithmic invariant of the pair
$(F_L,D)$ in class $\beta$, with one interior marking carrying $\gamma_L$ and one
marking of maximal contact order $k=D\cdot\beta$ with $D$, carrying no
insertion;
the ordinary primary insertion pulls back under the forgetful morphism,
as required by the cited corollary. In particular, summing at degree $m$
in the three-loop case gives the numerical consequence
\[
 d_2(m)=(-1)^m\sum_{H\cdot\beta=m}R_\beta(\gamma).
\]
This uses the virtual-cycle correspondence and the established coefficient
bridge.

\endgroup

\section{Reconstruction from primary evaluation
classes}\label{ext:reconstruction} Throughout this section we write $\Dcal_a$
for the diagonal series $\Dcal_a^\Delta$, as in Section~\ref{app:analytic3}, so
that $\Dcal_0=\mathcal F_{\gamma_L}$, and
$\theta_Q^{-1}\sum_{m>0}c_mQ^m=\sum_{m>0}(c_m/m)Q^m$. Here $L=3$ or $L=4$,
$n=L+1$, $F_L\subset(\mathbb P^1)^{L+1}$ is the smooth $(1,\ldots,1)$
hypersurface and $X_L=\Tot(K_{F_L})$. For a positive curve class $\beta$ write
$|\beta|=m_\beta=H\cdot\beta$ (both notations are used) and
$\beta_i=p_i\cdot\beta$, the integral curve-lattice coordinates
of~\Ipropref{prop:geometry}. The \emph{primary evaluation divisor} is the proper
one-point evaluation class
\[
 D_\beta=C^0_{0,\beta}=(\ev_{F_L})_*\bigl(c_{m_\beta-1}(\mathcal V)\cap
 [\overline{\mathcal M}_{0,1}(F_L,\beta)]^\vir\bigr)\in H^2(F_L,\Q)
\]
of~\Ieqref{eq:Cdef}, written $D_\beta=\sum_id_{\beta,i}p_i$; thus
$\int_{F_L}D_\beta\,b$ is the one-point primary invariant of $X_L$ in class
$\beta$ with insertion $\varpi^*b$, for $b\in H^{2L-2}(F_L,\Q)$.
Section~\Iref{ext:sec:primary} of the companion paper gives its full
multidegree extraction, in particular~\Ithmref{ext:thm:primary}, and its integral cover
transform
\[
 B_\beta=\sum_{r\mid\beta}\frac{\mu(r)}{r^2}D_{\beta/r},\qquad
 D_\beta=\sum_{r\mid\beta}\frac1{r^2}B_{\beta/r},
\]
where $r\mid\beta$ means that $\beta/r$ is an integral class. The classwise
one-point descendants are, as in~\Ieqref{eq:descdef},
\[
 N_{a,\beta}=\int_{[\overline{\mathcal M}_{0,1}(X_L,\beta)]^\vir}
 \psi^a\ev^*\varpi^*\frac{H^{L-1-a}}{(L-1-a)!},\qquad 0\le a\le L-1,
\]
so that $\Dcal_a=\sum_{\beta>0}N_{a,\beta}Q^{|\beta|}$.
All sums $\alpha+\gamma=\beta$ below are ordered sums of positive effective
classes.

By~\Ithmref{thm:CYcover} of the companion paper there are a smooth
connected projective $(L+1)$-fold $\widetilde Z$ with an actual isomorphism
$K_{\widetilde Z}\cong\mathcal O_{\widetilde Z}$, a retraction
$q\colon\widetilde Z\to F_L$, and an embedding $j_+\colon F_L\hookrightarrow\widetilde Z$
onto a divisor $F_{L,+}$ with normal bundle $K_{F_L}=\mathcal O_{F_L}(-H)$, such
that for every positive $\beta$ and every number of markings the stable-map
stack of $\widetilde Z$ in class $j_{+*}\beta$ is identified, together with its
standard obstruction theory and virtual class, with that of $X_L$ in class
$\beta$. We fix the embedding $j_+$ throughout this section. A bracket
\[
 \langle\tau_{a_1}(\mu_1),\ldots,\tau_{a_k}(\mu_k)\rangle_\beta
 =\langle\tau_{a_1}(q^*\mu_1),\ldots,\tau_{a_k}(q^*\mu_k)\rangle^{\widetilde Z}_{0,k,j_{+*}\beta},
\]
with classes $\mu_l$ on $F_L$, denotes a genus-zero descendant invariant of
$\widetilde Z$; brackets with insertions given directly by classes on
$\widetilde Z$ are understood in the same way. By the marked virtual
comparison such a bracket equals the local invariant of $X_L$ in class $\beta$
with insertions $\varpi^*\mu_l$. A divisor $E$ on $F_L$ enters through $q^*E$,
and $q^*E\cdot j_{+*}\beta=E\cdot\beta$. The divisor, string, dilaton, and
genus-zero topological recursion relations, and the Lee--Pandharipande
relation recalled below, hold on the smooth projective target $\widetilde Z$;
we apply them there and transport the results.

Let $\{\mathsf b_k\}$ be a homogeneous basis of $H^*(\widetilde Z,\Q)$ and
$\{\mathsf b^k\}$ the Poincar\'e-dual basis, so that $\sum_k\mathsf b_k\otimes\mathsf b^k$
is the class of the diagonal. Since the evaluation maps of the sector factor
through $j_+$, $\ev_*[\overline{\mathcal M}_{0,1}(\widetilde Z,j_{+*}\alpha)]^\vir=j_{+*}D_\alpha$,
and therefore
\begin{equation}\label{ext:eq:contract}
 \sum_k\langle\mathsf b_k\rangle_\alpha\,\mathsf b^k=j_{+*}D_\alpha,
 \qquad
 j_+^*j_{+*}D_\alpha=c_1(K_{F_L})D_\alpha=-HD_\alpha .
\end{equation}
Only classes of even degree contribute, and on even cohomology the roles of
$\{\mathsf b_k\}$ and $\{\mathsf b^k\}$ may be exchanged. The second identity is the
self-intersection formula for the smooth divisor $F_{L,+}$ of the smooth
variety $\widetilde Z$, with $N_{F_{L,+}/\widetilde Z}\cong K_{F_L}$. Consequently an
insertion $j_{+*}D_\alpha$ in a bracket of class $\gamma$ contributes
$-HD_\alpha$, that is, $\langle\ldots,j_{+*}D_\alpha\rangle_\gamma=-\langle\ldots,HD_\alpha\rangle_\gamma$.
In particular we never invert an ambient intersection matrix, and we do not assume that ambient classes exhaust the middle cohomology of $F_4$.
The two-mark cotangent relation used below is the virtual relation of
Lee--Pandharipande~\cite[Corollary~1]{ext:LP}; their more restrictive
reconstruction theorem is not needed.
\subsection{All cohomological descendants from full primaries}
Here $n=4$, $N_{0,\beta}=3\sum_i d_{\beta,i}$, and the relevant
compact-target pairing is
\begin{equation}
 C(\alpha,\gamma)
 =\int_{\widetilde Z}j_{+*}D_\alpha\,j_{+*}D_\gamma
 =-\int_{F_3} HD_\alpha D_\gamma
 =-2\sum_{i\ne j}d_{\alpha,i}d_{\gamma,j}. \label{ext:pairfour}
\end{equation}
The minus sign is the normal bundle $K_{F_3}$.
\begin{theorem}[Fourfold recursion]\label{ext:thm:fourrec}
\begin{equation}
 N_{1,\beta}=-\frac2{m_\beta}N_{0,\beta}
   +\frac1{2m_\beta}\sum_{\alpha+\gamma=\beta}
                |\alpha||\gamma|C(\alpha,\gamma),\qquad
 N_{2,\beta}=-\frac12\sum_{\alpha+\gamma=\beta}C(\alpha,\gamma).
 \label{ext:recfour}
\end{equation}
\end{theorem}
\begin{proof}
All brackets are those of the compact target, with $L=3$, and $m=m_\beta$.
Applying the descendant divisor equation twice gives
\begin{align}
 \langle\tau_1(H),H,H\rangle_\beta&=m^2N_{1,\beta}+4mN_{0,\beta},
 \label{ext:eq:four-left1}\\
 \langle\tau_2(1),H,H\rangle_\beta&=m^2N_{2,\beta}+2mN_{1,\beta}+2N_{0,\beta};
 \label{ext:eq:four-left2}
\end{align}
for instance $\langle\tau_1(H),H\rangle_\beta=mN_{1,\beta}+\langle H^2\rangle_\beta
=mN_{1,\beta}+2N_{0,\beta}$. On the other hand, genus-zero topological recursion
expresses $\psi_1$ on $\overline{\mathcal M}_{0,3}(\widetilde Z,j_{+*}\beta)$ as the
boundary divisor separating the first marking from the other two. A degree-zero
component carrying the first marking and the node alone is unstable and does not
occur. The stable degree-zero component carrying the markings $2,3$ and the node
contributes, after contraction of the diagonal, $\langle H,H^2\rangle_\beta=2mN_{0,\beta}$
in the first case and $\langle\tau_1(1),H^2\rangle_\beta=-\langle H^2\rangle_\beta=-2N_{0,\beta}$
in the second, the latter by the dilaton equation with one other marking. For a
positive splitting $\alpha+\gamma=\beta$ with the first marking on $\alpha$, the
divisor equation gives $\langle\mathsf b^k,H,H\rangle_\gamma=m_\gamma^2\langle\mathsf b^k\rangle_\gamma$,
the divisor and dilaton equations give $\langle H,\mathsf b_k\rangle_\alpha=m_\alpha\langle\mathsf b_k\rangle_\alpha$
and $\langle\tau_1(1),\mathsf b_k\rangle_\alpha=-\langle\mathsf b_k\rangle_\alpha$, and
\eqref{ext:eq:contract} turns $\sum_k\langle\mathsf b_k\rangle_\alpha\langle\mathsf b^k\rangle_\gamma$
into $\int_{\widetilde Z}j_{+*}D_\alpha\,j_{+*}D_\gamma=C(\alpha,\gamma)$. Hence
\begin{align}
 \langle\tau_1(H),H,H\rangle_\beta&=2mN_{0,\beta}
   +\sum_{\alpha+\gamma=\beta}m_\alpha m_\gamma^2C(\alpha,\gamma),
 \label{ext:eq:four-right1}\\
 \langle\tau_2(1),H,H\rangle_\beta&=-2N_{0,\beta}
   -\sum_{\alpha+\gamma=\beta}m_\gamma^2C(\alpha,\gamma).
 \label{ext:eq:four-right2}
\end{align}
Since $C$ is symmetric, symmetrizing the ordered sums gives
$\sum m_\alpha m_\gamma^2C=\frac m2\sum m_\alpha m_\gamma C$ and
$\sum m_\gamma^2C+\sum m_\alpha m_\gamma C=\frac{m^2}2\sum C$. Comparing
\eqref{ext:eq:four-left1} with~\eqref{ext:eq:four-right1} gives the first identity
of~\eqref{ext:recfour}; substituting it into~\eqref{ext:eq:four-left2} and comparing
with~\eqref{ext:eq:four-right2} gives the second. The full compact-target diagonal is
contracted before restriction, so no primitive cohomology is discarded, and the
minus sign in~\eqref{ext:pairfour} is the normal bundle $K_{F_3}$.
\end{proof}
By Section~\Iref{ext:sec:primary} of the companion paper, on a coordinate ray
$D_{re_i}=(H-2p_i)/r^2$ has zero self-pairing, and one recovers
$(N_0,N_1,N_2)=(6/r^2,-12/r^3,0)$.

\begin{theorem}[One diagonal primary series suffices]
\label{ext:thm:Bdiag}
With zero constant terms,
\begin{equation}
 \Dcal_2=\frac{\Dcal_0^2}{12},\qquad
 \theta_Q \Dcal_1=-2\Dcal_0-\frac{(\theta_Q \Dcal_0)^2}{12}. \label{ext:diagfour}
\end{equation}
If $\Dcal_0=\sum_{m\geq1}n_m\Li_2(Q^m)$, then
\[
 \Dcal_1=-2\sum_{m\geq1}\frac{n_m}{m}\Li_3(Q^m)
 -\frac1{12}\theta_Q^{-1}
       \left(\sum_{m\geq1}mn_m\Li_1(Q^m)\right)^2,\]
       \[
 \Dcal_2=\frac1{12}\left(\sum_{m\geq1}n_m\Li_2(Q^m)\right)^2 .
\]
\end{theorem}
\begin{proof}
By the dimension-uniform primary extraction, given in~\Ithmref{ext:thm:primary},
$d_{\beta,i}$ is the coefficient of $\mathbf x^\beta$ in an explicit series in
which $g$ is symmetric in $x_1,\ldots,x_4$ and the factor $f_i$ depends on $i$
only through the $i$th exponent. Hence $d_{\sigma\beta,\sigma(i)}=d_{\beta,i}$
for every permutation $\sigma$ of the coordinates, where $\sigma\beta$ permutes
the multidegree. Summing over the permutation-invariant set of classes with
$|\beta|=m$ shows that the coefficient of $p_i$ in $\sum_{|\beta|=m}D_\beta$ is
independent of $i$, so $\sum_{|\beta|=m}D_\beta=c_mH$. Pairing with $H^2/2$ and
using $\int_{F_3}H^3=24$ gives $12c_m=[Q^m]\Dcal_0$, that is,
$\sum_{|\beta|=m}D_\beta=[Q^m]\Dcal_0\,H/12$. Only the symmetry of the extraction formula is used, not any automorphism of $F_3$. Using $\int_FH^3=24$, the
aggregate of~\eqref{ext:pairfour} in bidegree $(k,l)$ is
$-[Q^k]\Dcal_0[Q^l]\Dcal_0/6$. Summing~\eqref{ext:recfour}
proves~\eqref{ext:diagfour}. Formal integration with zero constants proves the
displayed integrated formulas. The $n_m$ are the primary integers of
\Icorref{cor:primaryintegral} of the companion paper.\end{proof}

\subsection{Simplified assembly and a differential proof}
The fourfold identities give
\[
 f_2=-\theta_Q \Dcal_0,\quad f_1=(\theta_Q \Dcal_0)^2/12,\]
 \[
 f_0=2\theta_Q^{-1}\Dcal_0+\tfrac1{12}\theta_Q^{-1}((\theta_Q \Dcal_0)^2)
                  -\tfrac16\Dcal_0\theta_Q \Dcal_0,
 \qquad K_4=-\tfrac12K_2^2.
\]
These supplement the assembly proof above with a derivation from the
primary series and the scalar differential equation. Set
$u=\theta_Q \Dcal_0/12$, $\tau=W_1/A=R-u$, and $B=I^{(3)}_\circleddash/(-yA)$.
The identity for $\tau$ follows from the primary mirror comparison.
Put $\Delta=(1+4y)(1+16y)$. The monic third-order Picard--Fuchs
operator is the symmetric square of
\[
 \partial_y^2+\left(\frac1y+\frac{\Delta'}{2\Delta}\right)\partial_y
                  +\frac{2(1+8y)}{y\Delta}.
\]
Writing $A=h_0^2$, $W_1=h_0h_1$ gives
$\tau'=1/(y\sqrt\Delta A)$ and $R'=A/y$ by the normalized Wronskian.
Conjugating the physical equation $P_3(AB)=24$ therefore gives
\[
 \frac{d^3B}{d\tau^3}=24\sqrt\Delta A^2=24\frac{dR}{d\tau},
 \qquad \frac{d^2B}{d\tau^2}=24R.
\]
The last integration constant is zero by the normalized logarithmic germ.
For a direct comparison with the enumerative formula set
\[
 B_{\rm can}=4\tau^3+4u^3-\Dcal_1-2\Dcal_0u-16\zeta(3).
\]
Since $d\Dcal_0/dR=12u$ and $d\Dcal_1/dR=-2\Dcal_0-12u^2$, differentiation gives
$dB_{\rm can}/d\tau=12R^2-24uR+2\Dcal_0$ and
$d^2B_{\rm can}/d\tau^2=24R$.
Both $B$ and $B_{\rm can}$ have germ
$4\log^3y-16\zeta(3)+O(y(1+|\log y|)^3)$, so their possible affine
difference in $\tau$ vanishes. Positivity of $A$ and $\Delta$ extends this
identity along every $y>0$. This proof avoids the signed analytic assembly
of the companion paper~\cite{PaperI}, while retaining its geometric primary mirror comparison.
Finally the primary Lambert expansion gives the formal product
\[
 e^{12(R-\tau)}=\prod_{m\geq1}(1-Q^m)^{-mn_m}.
\]
Thus $e^\tau$ and $Q=e^R$ are distinct coordinates.
\subsection{The full diagonal and a triangular primary algorithm}
For $n=5$ set
\[
 P_\beta(\eta)=\int_FD_\beta\eta,\qquad
 T_\beta(\mu,\nu)=\langle\mu,\nu\rangle_\beta,\qquad
 A_\beta(\nu)=\langle\tau_1(\nu)\rangle_\beta,
\]
where $\eta\in H^6(F_4,\Q)$ has complex codimension three and $\mu,\nu\in H^4(F_4,\Q)$ have complex codimension two; the brackets are those of the compact target, with $L=4$, and $\beta_i=p_i\cdot\beta$.
Contraction of the \emph{full} compact-target diagonal against
the primary evaluation class gives $j_{+*}D_\alpha$; restriction
then gives $-HD_\alpha$. This is why all subsequent expressions
close on the ambient tensor without inverting an ambient pairing.

\begin{theorem}[Triangular two-point reconstruction]
\label{ext:thm:fiveT}
\begin{align}
 T_\beta(p_ip_j,\nu)
 ={}&\beta_i\beta_j A_\beta(\nu)+
 \beta_iP_\beta(p_j\nu)+\beta_jP_\beta(p_i\nu)
 +\sum_{\alpha+\gamma=\beta}\alpha_i\alpha_j
                           T_\gamma(HD_\alpha,\nu).
 \label{ext:eq:fiveT}
\end{align}
For any $i$ with $\beta_i>0$,
\begin{equation}
 A_\beta(\nu)=-\frac2{\beta_i}P_\beta(p_i\nu)
   -\frac1{\beta_i^2}\sum_{\alpha+\gamma=\beta}
                         \alpha_i^2T_\gamma(HD_\alpha,\nu). \label{ext:pivotfive}
\end{equation}
Thus the classes $D_\alpha$ with $|\alpha|\le|\beta|$ determine $A_\beta(\nu)$ for every $\nu$, and $T_\beta(\mu,\nu)$ for every $\nu$ and every ambient $\mu$.
\end{theorem}
\begin{proof}
For divisor classes $E,E'$ on $F_4$ write $e_\beta=E\cdot\beta$ and
$e'_\beta=E'\cdot\beta$, and let $\nu\in H^4(F_4,\Q)$. Applying the descendant
divisor equation twice,
\[
 \langle\tau_1(\nu),E,E'\rangle_\beta
 =e_\beta e'_\beta A_\beta(\nu)+e'_\beta P_\beta(E\nu)+e_\beta P_\beta(E'\nu).
\]
Genus-zero topological recursion for $\psi_1$ on
$\overline{\mathcal M}_{0,3}(\widetilde Z,j_{+*}\beta)$ separates the first marking from
the two divisor markings. A degree-zero component carrying the first marking and
the node alone is unstable. The stable degree-zero component carrying the two
divisor markings and the node contributes $\langle\nu,EE'\rangle_\beta=T_\beta(EE',\nu)$.
A positive splitting with the first marking on $\gamma$ and the divisor markings
on $\alpha$ contributes, by the divisor equation and~\eqref{ext:eq:contract},
$e_\alpha e'_\alpha\langle\nu,j_{+*}D_\alpha\rangle_\gamma=-e_\alpha e'_\alpha T_\gamma(HD_\alpha,\nu)$.
Comparing the two evaluations gives, for arbitrary divisors $E,E'$,
\begin{equation}\label{ext:eq:fiveEE}
 T_\beta(EE',\nu)=e_\beta e'_\beta A_\beta(\nu)+e'_\beta P_\beta(E\nu)+e_\beta P_\beta(E'\nu)
 +\sum_{\alpha+\gamma=\beta}e_\alpha e'_\alpha\,T_\gamma(HD_\alpha,\nu).
\end{equation}
Taking $E=p_i$ and $E'=p_j$ gives~\eqref{ext:eq:fiveT}. Taking $E=E'=p_i$ and
using $p_i^2=0$ gives~\eqref{ext:pivotfive}. Every tensor on the right
of~\eqref{ext:pivotfive} has smaller positive degree, and the classes $p_ip_j$
span the ambient part of $H^4(F_4,\Q)$; induction on $|\beta|$ gives the final
assertion. Compatibility between pivot choices follows from the virtual
identities of the actual theory.
\end{proof}

\subsection{Explicit inverse for the numerical one-point descendants}
Fix a divisor class $E$ on $F_4$ with $e_\beta=E\cdot\beta\ne0$ for the classes
considered; one may always take $E=H$, since $H$ is ample on $F_4$ and hence
$H\cdot\beta>0$ for every nonzero effective class $\beta$. For another divisor
$\lambda$ define ordinary two-point descendants
\[
 J_\beta(\lambda;\nu)=\langle\tau_1(\lambda),\nu\rangle_\beta,
 \qquad K_\beta(\nu)=\langle\tau_2(1),\nu\rangle_\beta.
\]
These are ordinary genus-zero descendant Gromov--Witten invariants; 
they should not be confused with the Pandharipande--Zinger~\cite{PZ} multiple-cover
 transforms of primary genus-zero invariants on Calabi--Yau fivefolds.
What is reconstructed below are the dimensionally relevant numerical one-point
descendants $N_{a,\beta}$, $0\le a\le3$, of the Calabi--Yau fivefold sector,
together with the auxiliary two-point invariants $T_\beta$, $J_\beta$ and $K_\beta$;
no assertion is made about other descendant insertions.
\begin{theorem}[All-degree inverse]
\label{ext:thm:inverse}
\begin{align}
 e_\beta^2A_\beta(\nu)
 &=T_\beta(E^2,\nu)-2e_\beta P_\beta(E\nu)
       -\sum_{\alpha+\gamma=\beta}e_\alpha^2
                                  T_\gamma(HD_\alpha,\nu),\notag\\
 J_\beta(\lambda;\nu)
 &=-(\lambda\cdot\beta)A_\beta(\nu)-P_\beta(\lambda\nu)
       -\sum_{\alpha+\gamma=\beta}(\lambda\cdot\alpha)
                                  T_\gamma(HD_\alpha,\nu),\notag\\
 K_\beta(\nu)&=A_\beta(\nu)+
             \sum_{\alpha+\gamma=\beta}T_\gamma(HD_\alpha,\nu).
 \label{ext:eq:JK}
\end{align}
Writing $Q_{2,\beta}(\lambda)=\langle\tau_2(\lambda)\rangle_\beta$
and $Q_{3,\beta}=\langle\tau_3(1)\rangle_\beta$,
\begin{align}
 e_\beta^2Q_{2,\beta}(\lambda)
 ={}&J_\beta(\lambda;E^2)
   -\sum_{\alpha+\gamma=\beta}e_\gamma^2
                                      J_\alpha(\lambda;HD_\gamma)
   -2e_\beta A_\beta(E\lambda)-P_\beta(E^2\lambda),\notag\\
 e_\beta^2Q_{3,\beta}
 ={}&K_\beta(E^2)
   -\sum_{\alpha+\gamma=\beta}e_\gamma^2K_\alpha(HD_\gamma)
   -2e_\beta Q_{2,\beta}(E)-A_\beta(E^2).
 \label{ext:eq:Q23}
\end{align}
The answer is independent of $E$, and
\[
 (N_0,N_1,N_2,N_3)_\beta
  =\left(P_\beta(H^3/6),A_\beta(H^2/2),Q_{2,\beta}(H),Q_{3,\beta}\right).
\]
The full primary divisors $D_\beta$, or equivalently $B_\beta$,
alone determine every displayed invariant.
\end{theorem}
\begin{proof}
The first identity is~\eqref{ext:eq:fiveEE} with $E'=E$. On
$\overline{\mathcal M}_{0,2}(\widetilde Z,j_{+*}\beta)$, $\beta\ne0$, the
Lee--Pandharipande relation reads $\psi_1+\psi_2=[D_{1|2}]$, where $D_{1|2}$
parametrizes domains split by a node with the first marking on a component of
class $j_{+*}\alpha$ and the second on one of class $j_{+*}\gamma$,
$\alpha+\gamma=\beta$; both classes are positive, because a degree-zero component
with one marking and one node is unstable. Capping this relation with
$\ev_1^*\lambda\,\ev_2^*\nu$, respectively with $\psi_1\ev_2^*\nu$, and splitting the
boundary along the diagonal gives
\begin{align*}
 J_\beta(\lambda;\nu)+\langle\lambda,\tau_1(\nu)\rangle_\beta
 &=\sum_{\alpha+\gamma=\beta}\sum_k\langle\lambda,\mathsf b_k\rangle_\alpha
   \langle\mathsf b^k,\nu\rangle_\gamma,\\
 K_\beta(\nu)+\langle\tau_1(1),\tau_1(\nu)\rangle_\beta
 &=\sum_{\alpha+\gamma=\beta}\sum_k\langle\tau_1(1),\mathsf b_k\rangle_\alpha
   \langle\mathsf b^k,\nu\rangle_\gamma .
\end{align*}
By the divisor equation $\langle\lambda,\tau_1(\nu)\rangle_\beta=(\lambda\cdot\beta)A_\beta(\nu)
+P_\beta(\lambda\nu)$ and $\langle\lambda,\mathsf b_k\rangle_\alpha=(\lambda\cdot\alpha)\langle\mathsf b_k\rangle_\alpha$;
by the dilaton equation with one other marking
$\langle\tau_1(1),\tau_1(\nu)\rangle_\beta=-A_\beta(\nu)$ and
$\langle\tau_1(1),\mathsf b_k\rangle_\alpha=-\langle\mathsf b_k\rangle_\alpha$. Together with
$\sum_k\langle\mathsf b_k\rangle_\alpha\langle\mathsf b^k,\nu\rangle_\gamma
=\langle j_{+*}D_\alpha,\nu\rangle_\gamma=-T_\gamma(HD_\alpha,\nu)$
from~\eqref{ext:eq:contract}, this gives the formulas for $J_\beta$ and $K_\beta$
in~\eqref{ext:eq:JK}. For~\eqref{ext:eq:Q23}, the divisor equation applied twice gives
\begin{align*}
 \langle\tau_2(\lambda),E,E\rangle_\beta
 &=e_\beta^2Q_{2,\beta}(\lambda)+2e_\beta A_\beta(E\lambda)+P_\beta(E^2\lambda),\\
 \langle\tau_3(1),E,E\rangle_\beta
 &=e_\beta^2Q_{3,\beta}+2e_\beta Q_{2,\beta}(E)+A_\beta(E^2),
\end{align*}
while topological recursion for $\psi_1$ on
$\overline{\mathcal M}_{0,3}(\widetilde Z,j_{+*}\beta)$, separating the first marking
from the two $E$ markings, gives
\begin{align*}
 \langle\tau_2(\lambda),E,E\rangle_\beta
 &=J_\beta(\lambda;E^2)-\sum_{\alpha+\gamma=\beta}e_\gamma^2J_\alpha(\lambda;HD_\gamma),\\
 \langle\tau_3(1),E,E\rangle_\beta
 &=K_\beta(E^2)-\sum_{\alpha+\gamma=\beta}e_\gamma^2K_\alpha(HD_\gamma).
\end{align*}
Here the first terms come from the stable degree-zero component carrying the two
$E$ markings and the node; a degree-zero component carrying the first marking and
the node alone is unstable; and in a positive splitting the first marking lies on
$\alpha$, the divisor equation gives $\langle\mathsf b^k,E,E\rangle_\gamma=e_\gamma^2\langle\mathsf b^k\rangle_\gamma$,
and~\eqref{ext:eq:contract} turns the contraction into an insertion $j_{+*}D_\gamma$,
that is, into $-HD_\gamma$. Comparing the two evaluations proves~\eqref{ext:eq:Q23}. Theorem~\ref{ext:thm:fiveT} supplies
the tensor from $D$; the cover inversion supplies $D$ from $B$.
All equations hold for the actual invariants, hence are
independent of the auxiliary divisor.
\end{proof}
\subsection{Two-input diagonal closure}
\begin{theorem}[Two diagonal inputs suffice]
\label{ext:thm:fiveclosure}
Let $\mathcal T=\sum_\beta T_\beta(H^2,H^2)Q^{|\beta|}$.
Then
\[
 (1-\theta_Q^2\Dcal_0/20)\mathcal T=2\theta_Q^2\Dcal_1+12\theta_Q \Dcal_0,
\]
and
\begin{align}
 \theta_Q^2\Dcal_2={}&-6\theta_Q \Dcal_1-12\Dcal_0
 +\frac{\theta_Q \Dcal_1\theta_Q^2\Dcal_0-\theta_Q \Dcal_0\theta_Q^2\Dcal_1}{10}
 +\frac3{10}\Dcal_0\theta_Q^2\Dcal_0-\frac35(\theta_Q \Dcal_0)^2,\label{ext:eq:diagfive2}\\
 \theta_Q^2\Dcal_3={}&
 \frac{\Dcal_0\theta_Q^2\Dcal_1-\Dcal_1\theta_Q^2\Dcal_0}{10}
 +\frac35\Dcal_0\theta_Q \Dcal_0-2\theta_Q \Dcal_2. \label{ext:eq:diagfive3}
\end{align}
These determine $\Dcal_2,\Dcal_3$ uniquely from $\Dcal_0,\Dcal_1$ as zero-constant
formal power series.
\end{theorem}
\begin{proof}
As in the proof of Theorem~\ref{ext:thm:Bdiag}, the permutation symmetry of the
extraction formula of~\Ithmref{ext:thm:primary} gives
$\sum_{|\beta|=m}D_\beta=c_mH$, and pairing with $H^3/6$, using $\int_{F_4}H^4=120$,
gives $c_m=[Q^m]\Dcal_0/20$. Put
\[
 \mathcal J=\sum_{\beta>0}J_\beta(H;H^2)Q^{|\beta|},\qquad
 \mathcal K=\sum_{\beta>0}K_\beta(H^2)Q^{|\beta|},\qquad
 V=1-\theta_Q^2\Dcal_0/20 .
\]
Aggregate~\eqref{ext:eq:fiveEE} with $E=E'=H$ and $\nu=H^2$,
and~\eqref{ext:eq:JK} with $\lambda=H$ and $\nu=H^2$, using
$A_\beta(H^2)=2N_{1,\beta}$ and $P_\beta(H^3)=6N_{0,\beta}$. The intermediate
equations are
$V\mathcal T=2\theta_Q^2\Dcal_1+12\theta_Q \Dcal_0$,
$\mathcal J=-2\theta_Q \Dcal_1-6\Dcal_0-(\theta_Q \Dcal_0)\mathcal T/20$,
and $\mathcal K=2\Dcal_1+\Dcal_0\mathcal T/20$.
Equation~\eqref{ext:eq:Q23} with $E=\lambda=H$ gives
$\theta_Q^2\Dcal_2=V\mathcal J-4\theta_Q \Dcal_1-6\Dcal_0$ and
$\theta_Q^2\Dcal_3=V\mathcal K-2\theta_Q \Dcal_2-2\Dcal_1$.
Eliminate $\mathcal T$. Formal $\theta_Q^{-2}$ with zero constants
is unique.
\end{proof}
Thus the diagonal reconstruction closes on two inputs. We do not identify $\mathcal T$ with a Yukawa coupling, nor do we show that a zero of $V$ is an actual singularity.

\section{Endpoint values and physical threshold behavior}\label{ext:lowend}
We use the primitives $C_j$ and their normalization of Section~\Iref{ext:endpoints}
of the companion paper, and put $g_3=2\zeta(3)/3$,
$g_5=2\zeta(5)/5$. The endpoint theorem there gives, for $L=3$,
\[
 (C_1,C_2,C_3,C_4)_*=(R_*,0,g_3,g_3R_*),\qquad
 R_*=-7\zeta(3)/\pi^2,\quad Q_*=e^{R_*}.
\]
Indeed the torus logarithmic formula gives
$R_*=-2\,\mathrm m_{\rm Mah}(1+z_1+z_2+z_3)$, where
$\mathrm m_{\rm Mah}(P)=\int_{[0,1]^n}\log\abs{P(e^{2\pi i t_1},\dots,e^{2\pi i t_n})} dt_1\cdots dt_n$ denotes the logarithmic Mahler measure of the 
multivariate polynomial $P$.
Here $\mathrm m_{\rm Mah}(1+z_1+z_2+z_3)=7\zeta(3)/(2\pi^2)$
by~\cite[\S6.2, equation (66)]{ext:BSWZ}. Consequently
\[
 (K_2,K_3,K_4)_*=
 (-R_*^2/2,\ g_3+R_*^3/3,\ -R_*^4/8).
\]
At $L=4$, write $R_\infty=C_1(\infty)$ and $S_\infty=C_2(\infty)$.
The same theorem gives
\[
 (C_1,C_2,C_3,C_4,C_5)_\infty=
 (R_\infty,S_\infty,g_3,g_3R_\infty,g_3S_\infty+g_5).
\]
Thus all descendant endpoint values follow by the finite formula
$K_k(\infty)=\sum_{j=0}^k(-R_\infty)^{k-j}C_j(\infty)/(k-j)!$.
These reduce everything to two constants, which we do not evaluate in closed form. Explicit convergent subtracted Bessel integrals for
them are supplied in~\Ilemref{ext:endpoint}.

By~\Ipropref{prop:sTaylor} and~\Iremref{rem:pseudothresholds}, the physical four-loop Feynman integral is
holomorphic at the pseudothresholds $s=1,9$, since both lie in its Taylor disk
$|s|<25$. At the physical threshold, with $\delta=25-s\downarrow0$,
\[
 I^{(4)}_\circleddash(25-\delta)=T+\alpha\delta\log\delta+c\delta
       +\frac{\alpha}{20}\delta^2\log\delta+d\delta^2
       +O(\delta^3|\log\delta|),\qquad
 \alpha=\frac{\pi^2}{5\sqrt5}.
\]
Here $T,c,d$ are finite constants. To verify the normalization, the large-$t$
part of $16tI_0(\sqrt{s}t)K_0(t)^5$ is
$2\pi^2s^{-1/4}t^{-2}e^{-(5-\sqrt{s})t}(1+O(t^{-1}))$. Integration gives
$2\pi^2s^{-1/4}\epsilon\log\epsilon$, where
$\epsilon=5-\sqrt{s}=\delta/10+O(\delta^2)$, proving $\alpha$. The threshold
indicial polynomial of the physical fourth-order equation is
$240000z(z-1)^2(z-2)$. Its Frobenius recurrence gives only first powers of the
logarithm and gives the next logarithmic coefficient $\alpha/20$; the local
Frobenius series is convergent at least for $|s-25|<16$, up to the next singular
point $s=9$ of the differential equation. This is compatible with the
continuous but non-differentiable threshold value established by the positive
coefficient asymptotics in~\Ipropref{prop:sTaylor}.

\section{Exact three-loop period connection polynomials}\label{ext:connection3}
A further question is to determine the coefficients of the logarithmic periods at $x=0$. At three loops this can be done exactly.
Set
\begin{equation}\label{ext:P3}
 \mathcal P_3=\theta_y^3+y(20\theta_y^3+30\theta_y^2+18\theta_y+4)
                         +64y^2(\theta_y+1)^3.
\end{equation}
The homogeneous ambient periods $W_0,W_1,W_2$ solve $\mathcal P_3W=0$.
They are fixed by their leading terms $1,\log y,\log^2y/2$, with no
additional constant terms at the maximally unipotent point.

\begin{theorem}\label{ext:periods}
Along the positive $y$ axis, as $x=1/y\downarrow0$,
\[
 W_j(y)=xP_{j,0}(\log x)+O\bigl(x^2(1+|\log x|^2)\bigr),\qquad 0\leq j\leq2,
\]
where
\begin{equation}\label{ext:periodanswer}
 \quad
 P_{0,0}(t)=\frac3{16\pi^2}(t-6\log2)^2,\qquad
 P_{1,0}(t)=\frac{t-6\log2}{8},\qquad
 P_{2,0}(t)=\frac{\pi^2}{24}.
\end{equation}
In particular the constant coefficient of $P_{0,0}$ is
$27\log^2(2)/(4\pi^2)$.
\end{theorem}
\begin{proof}
Let
\[
 z=\frac{108y^2}{(1+4y)^3},\quad p=\frac1{1+4y},\quad
 a=\frac16,\quad b=\frac13,\quad f(z)={}_2F_1(a,b;1;z).
\]
Let $g_0(z)$ be the second local Gauss solution with
$g_0=f\log z+$ a power series with zero constant term, and put
$h_0=g_0-(\log108)f$. Direct substitution in~\eqref{ext:P3} gives the local
identifications
\begin{equation}\label{ext:square}
 W_0=pf^2,\qquad W_1=\frac12pfh_0,\qquad W_2=\frac18ph_0^2.
\end{equation}
For completeness, the Gauss equation used in that substitution is
\[
 f''+\frac{1-3z/2}{z(1-z)}f'-\frac{1}{18z(1-z)}f=0.
\]
The symmetric products of any two of its solutions, multiplied by $p$,
solve~\eqref{ext:P3}; this is a rational differential identity. It can be checked by collecting the three coefficients of $f^2,ff',f'^2$.
Because $\log z-\log108=2\log y+O(y)$ at zero, the leading logarithms
in~\eqref{ext:square} match the ambient normalization, proving the identities.

These expressions must be continued through $y=1/2$; one cannot substitute the principal branch of $f$ on both sides. Put
\[
 F_1(z)={}_2F_1(a,b;1/2;1-z),\quad
 \kappa=\frac{\Gamma(1/2)}{\Gamma(1/6)\Gamma(1/3)},\quad
 c=\frac{2\pi}{\sqrt3}.
\]
The zero-balanced logarithmic connection formula at $z=0$
\cite[\S15.8(ii)]{ext:DLMF} is
\[
 F_1=\kappa[-g_0+H_0 f],\qquad
 H_0=2\psi(1)-\psi(1/6)-\psi(1/3)=c+\log108.
\]
Consequently $h_0=cf-F_1/\kappa$. At $z=1$, the Gauss connection formula
\cite[\S15.10(ii)]{ext:DLMF} writes
\[
 f=A F_1+B\sqrt{1-z}\,F_2(z),\qquad
 A=\frac{\Gamma(1/2)}{\Gamma(5/6)\Gamma(2/3)},
\]
where $F_2$ is analytic at $1$ and its explicit coefficient $B$ is not
needed. Reflection of gamma functions gives $2A\kappa=1/c$.
The exact factorization
\begin{equation}\label{ext:branch}
 1-z(y)=\frac{(1+16y)(1-2y)^2}{(1+4y)^3}
\end{equation}
shows that continuation through $y=1/2$ changes the sign of the analytic
square root relative to its nonnegative real value. The point $y=1/2$
is ordinary for~\eqref{ext:P3}, so this continuation is unambiguous.
It acts on the local pair by
\begin{equation}\label{ext:reflection}
 f\longmapsto 2AF_1-f=-h_0/c,\qquad h_0\longmapsto-cf.
\end{equation}
As $y\to\infty$ the argument $z$ returns to zero. In the local expressions
on the right of~\eqref{ext:reflection},
\[
 f=1+O(x),\quad h_0=t-6\log2+O(x(1+|t|)),\quad
 p=x/4+O(x^2),\qquad t=\log x.
\]
Substitution in~\eqref{ext:square} proves~\eqref{ext:periodanswer}.
\end{proof}

The branch reflection explains why a direct evaluation of the principal
hypergeometric function at the small positive value $z(y)$ for large $y$
would give an incorrect connection matrix. The logarithm at infinity is
created by the continuation through the pullback point $z=1$.
Once~\eqref{ext:periodanswer} is fixed, the differential recurrence at $x=0$
determines all subsequent homogeneous Frobenius coefficients. We do not have the analogous formula at four loops.

\begin{corollary}
At three loops each of $K_2,K_3,K_4$ has a one-sided $C^1$ extension
at $Q_*$ and an unbounded second derivative. Its Taylor radius as a power series in $Q$ lies
between $1/12$ and $Q_*$. 
\end{corollary}
\begin{proof}
The connection polynomials above give $R_*-R\sim 3x\log^2x/(16\pi^2)$. Moreover
$\partial_RK_2=W_1/A-R$, whose ratio has a nonzero $1/\log x$ term, and
$\partial_RK_3=(W_2-RW_1)/A-K_2$, whose corresponding coefficient is nonzero
because $R_*\ne0$. Their derivatives are unbounded after conversion from $x$ to
$R$. Termwise differentiation is legitimate because the Frobenius series converge. First derivatives have finite
limits. The relation $K_4=-K_2^2/2$ and $K_2(Q_*)\ne0$ give the same assertions
for $K_4$. The lower bound follows from the principal mirror disk, stated in 
\Icorref{cor:Qradius} of the companion paper. A radius greater than $Q_*$
would contradict this nonanalyticity.
\end{proof}
\section{The four-loop connection quantities to certify}\label{ext:cert:setup}
Set $x=1/y$, $t=\log x$, $\theta_y=y\partial_y$. The ambient functions obey
\[
 P_4W_j=\delta_{j4},\quad 0\le j\le4,\qquad
 W_j=(\log y)^j/j!+O(y\log^j y),
\]
where
\begin{align*}
 P_4&=\theta_y^4+yQ_1(\theta_y)+y^2Q_2(\theta_y)+y^3Q_3(\theta_y),\\
 Q_1(z)&=35z^4+70z^3+63z^2+28z+5,\\
 Q_2(z)&=259z^4+1036z^3+1580z^2+1088z+285,\\
 Q_3(z)&=225(z+1)^2(z+2)^2.
\end{align*}
The normalization specifies all degree-zero logarithmic coefficients, with
no extra constant terms. Define $C_0=1$, $\theta_y C_{j+1}=W_j$,
$C_{j+1}\sim(\log y)^{j+1}/(j+1)!$, and
\[
 R=C_1,\quad A=W_0,\quad Q=e^R,\quad
 K_k=\sum_{j=0}^k\frac{(-R)^{k-j}}{(k-j)!}C_j.
\]
The diagonal descendant potentials are $-20K_2,-60K_3,-120K_4,-120K_5$.

The endpoint expansion has the form
\[
 W_j=x(a_j+b_j\log x)+O(x^2|\log x|),\qquad b_0<0.
\]
Write $R_\infty=\lim R$, $Q_\infty=e^{R_\infty}$, and
\begin{equation}\label{ext:cert:Ek}
 d_j=\frac{a_jb_0-a_0b_j}{b_0^2},\qquad
 E_k=\sum_{h=1}^{k-1}\frac{(-R_\infty)^{k-1-h}}{(k-1-h)!}d_h,
 \quad 2\le k\le5.
\end{equation}
All five rows, including the inhomogeneous row $j=4$, must be included.

\begin{theorem}[Certified individual singularities]\label{ext:cert:main}
The quantities in~\eqref{ext:cert:Ek} have signs
\[
 E_2>0,\qquad E_3<0,\qquad E_4>0,\qquad E_5>0.
\]
Each of $K_2,K_3,K_4,K_5$ has a one-sided $C^1$ extension to $Q_\infty$,
but has an unbounded second derivative there. More precisely, if
$\Delta=R_\infty-R\downarrow0$,
\begin{equation}\label{ext:cert:second}
 \partial_Q^2K_k\sim
 \frac{E_k}{Q_\infty^2\Delta|\log\Delta|^2}.
\end{equation}
Consequently their individual $Q$-Taylor radii $\mathcal R_k$ satisfy
\begin{equation}\label{ext:cert:radii}
 \frac1{20}\le\mathcal R_k\le Q_\infty,
 \qquad k=2,3,4,5,
\end{equation}
where the endpoint is enclosed by
\[
 0.336611738699193750820017
 \ \le Q_\infty\le
 0.336611738699193750820018.
\]
\end{theorem}
The sign assertion is a computer-assisted statement. Its computational
part consists of finite ball-arithmetic operations with the analytic tail
bounds proved below.
The rational intervals in Table~\ref{ext:cert:intervals} are intentionally rounded
outward and are sufficient to prove every sign. The proof uses only these
rational enclosures, whose endpoints are exact decimals with denominator
$10^{24}$; truncated decimals elsewhere in the paper, such as
$Q_\infty=0.3366117386\ldots$, are approximations printed for readability and
play no role in any argument.

\begin{table}[h]
\centering\small
\begin{tabular}{crr}
\toprule &Lower endpoint&Upper endpoint\\\midrule
$E_2$&$1.959219916247366530180947$&$1.959219916247366530180948$\\
$E_3$&$-0.535347537401972845058210$&$-0.535347537401972845058209$\\
$E_4$&$0.420379682614378171767897$&$0.420379682614378171767898$\\
$E_5$&$0.748770280780437029430306$&$0.748770280780437029430307$\\
\bottomrule
\end{tabular}
\caption{Certified rational enclosures, with common denominator $10^{24}$.}
\label{ext:cert:intervals}
\end{table}

\begin{proof}[Analytic part of Theorem~\ref{ext:cert:main}]
The convergent endpoint expansion and primitive tail below imply that
$C_j,K_k$ have finite limits. Since $b_0<0$, all ratios $W_j/A$ have finite
limits, so the first $R$- and $Q$-derivatives of $K_k$ have finite limits.
This gives the one-sided $C^1$ extension.
Furthermore,
\[
 \Delta\sim(-b_0)x|t|,\qquad
 \frac{W_j}{A}=\frac{b_j}{b_0}+\frac{d_j}{t}+O(t^{-2})+O(x).
\]
Differentiation of the normalization gives
\[
 \partial_RK_k=
 \sum_{h=0}^{k-1}\frac{(-R)^{k-1-h}}{(k-1-h)!}\frac{W_h}{A}-K_{k-1},
 \qquad K_1=0.
\]
Its inverse-logarithm coefficient is exactly $E_k$: variation of $R$ and
$K_{k-1}$ from their endpoint limits contributes $O(\Delta)$. All expansions
here arise from convergent logarithmic series and may be differentiated termwise.
Using $\dd R/\dd t=-A$ and $\partial_Q^2=Q^{-2}(\partial_R^2-\partial_R)$
gives~\eqref{ext:cert:second}. The certified nonzero signs therefore prove
unboundedness. A Taylor radius larger than $Q_\infty$ would supply a holomorphic
extension across that point, contradicting this conclusion.
\Icorref{cor:Qradius} gives the lower bounds.
\end{proof}

\section{Certificate I: exact normalization and initial tails}
The implementation uses the truncated parameter ring $\mathbb Q[\varepsilon]/
(\varepsilon^5)$, with
\begin{equation}\label{ext:cert:germrec}
 u_0=1,\qquad
 u_D(\varepsilon)=-(D+\varepsilon)^{-4}
       \sum_{l=1}^{\min(3,D)}Q_l(D-l+\varepsilon)u_{D-l}(\varepsilon).
\end{equation}
Then
\[
 W_j(y)=[\varepsilon^j]\,e^{\varepsilon\log y}
                  \sum_{D\ge0}u_D(\varepsilon)y^D.
\]
The degree-zero term gives $\theta_y^4W_4=1$ and zero source for the lower
four rows. The positive-degree recursion gives the stated equations and
fixes the normalized solutions uniquely. As an independent normalization
check, the code compares the 50 rational coefficients in total degrees zero through nine (five parameter coefficients per degree) with the
square-free ambient formula from the square-free ambient seed; all agree exactly.

Use the coefficient $\ell^1$ norm in the truncated parameter ring.
For $D\ge10$,
\[
 \|(D+\varepsilon)^{-4}\|\le(D-1)^{-4},\qquad
 \|Q_l(D-l+\varepsilon)\|\le Q_l(D).
\]
Each $Q_l(D)/(D-1)^4$ decreases with $D\ge10$ because its coefficients
are positive and its degree is at most four. The finite rational check
\[
 \sum_{l=1}^3\frac{Q_l(10)}{9^4\,100^l}
 =\frac{15497623}{21870000}<1
\]
and the ten base cases give
\begin{equation}\label{ext:cert:germmajorant}
 \|u_D\|\le2\,100^D\quad(D\ge0).
\end{equation}
This bound is crude but sufficient at $y_0=1/512$.
Let $q_0=100/512$, and truncate after degree $N_g$.
A common bound for each omitted logarithmic jet $\theta_y^aW_j$, $0\le a\le3$,
and each primitive $C_{j+1}$ is
\begin{equation}\label{ext:cert:initialtail}
 e_g=1024q_0^{N_g+1}(N_g+2)^3
                \frac{1+4q_0+q_0^2}{(1-q_0)^4}.
\end{equation}
Indeed, evaluation of the logarithmic polynomial costs at most
$e^{|\log y_0|}=512$, and each logarithmic differentiation costs at most
$D+1$ in the exponential-factorial basis. The sum is bounded using
$D=N_g+1+m$ and
$\sum_{m\ge0}(m+1)^3q^m=(1+4q+q^2)/(1-q)^4$.
For primitives, the inverse of $D+\partial_{\log y}$ has norm at most
$1/(D-1)$ for $D>1$, so the same bound suffices for the omitted tail.
The degree-zero primitive is integrated exactly. Conversion to ordinary
$y$-derivatives is done only after these bounds have been added to the balls.

\section{Certificate II: finite continuation with Cauchy bounds}
At a positive center $c$, put $z=(y-c)/c$. The augmented state is
\[
 U=(W,\theta_y W,\theta_y^2W,\theta_y^3W,C,1).
\]
It satisfies a homogeneous first-order system in $z$; the last component
encodes the possible source $1$. Write
$P_4=\sum_{a=0}^4p_a(y)\theta_y^a$, so
$p_4(y)=(1+y)(1+9y)(1+25y)$.
Fix $r=1/4$. On $|z|\le r$, an induced maximum-row-sum bound for that system is
\begin{equation}\label{ext:cert:M}
 M(c)=\frac1{1-r}\max\left\{1,
 \frac{1+\sum_{a=0}^3p_a(c(1+r))}
 {(1+c(1-r))(1+9c(1-r))(1+25c(1-r))}\right\}.
\end{equation}
The numerator estimate uses positive polynomial coefficients; the denominator
uses $|1+ac(1+z)|\ge1+ac(1-r)$ for $a=1,9,25$.
Integration along any radial segment gives
\[
 \sup_{|z|\le r}\|U(z)\|_\infty\le
 B:=\|U(0)\|_\infty e^{M(c)r}.
\]
This bound applies to every exact initial vector contained in the current
balls, because $\|U(0)\|_\infty$ is rounded upward.

Each step has $0<h\le c/16$, so $\eta=h/(rc)\le1/4$.
Taylor coefficients through degree $N$ are computed from the scalar ordinary
differential equation by ball arithmetic. With
$S_3(v)=(1+4v+v^2)/(1-v)^4$, a common base remainder is
\[
 e=B4^{-(N+1)}S_3(1/4).
\]
The omitted remainder for the ordinary derivative $W^{(a)}$ at $c+h$,
$0\le a\le3$, is bounded by
\begin{equation}\label{ext:cert:steptail}
 e\left(\frac{4(N+1)}{rc}\right)^a.
\end{equation}
To see this, Cauchy's inequality bounds the $z$-Taylor coefficient by
$Br^{-m}$, and the derivative tail is bounded by the geometric sum with
$(N+1+j)^a\le(N+1)^a(j+1)^a$. The bound for $C$, computed through degree
$N$ using $C'=W/y$, is simply $e$. The code adds these bounds after every
step; subsequent operations propagate the full accumulated balls.
Roundoff and truncation are therefore both included.

\section{Certificate III: convergent endpoint matching}
For $W=xF$ and $D=x\partial_x$, the equation becomes
\[
 [225D^2(D-1)^2+xD^2(259D^2+26)+x^2Q_1(D)+x^3(D+1)^4]F=b x^2.
\]
Write $F=\sum_{k\ge0}x^kP_k(\log x)$, with $P_k$ linear.
The four coefficients of $P_0,P_1$ are free, and for $k\ge2$,
\begin{align*}
 225(k+\partial)^2(k-1+\partial)^2P_k={}&
 -(k-1+\partial)^2[259(k-1+\partial)^2+26]P_{k-1}\\
 &-Q_1(k-2+\partial)P_{k-2}
 -(k-2+\partial)^4P_{k-3}+b\delta_{k2}.
\end{align*}
Here $\partial=\partial_{\log x}$, $P_{-1}=0$.
The inverse on linear polynomials is
\[
 \frac1{225k^2(k-1)^2}
 \left(1-\left(\frac2k+\frac2{k-1}\right)\partial\right).
\]
We use four homogeneous basis vectors with unit initial data and the
particular solution with $P_0=P_1=0$, $b=1$.

For the ordinary coefficient norm of linear polynomials, $\|\partial\|\le1$.
For $k\ge4$ the displayed recurrence has the bounds
\[
 \|P_k\|\le5\|P_{k-1}\|+2\|P_{k-2}\|+\tfrac1{100}\|P_{k-3}\|.
\]
Explicit dominating coefficients are
$1807/405<5$, $871/450<2$, and $13/1350<1/100$.
The first four coefficients of each of the five basis solutions are checked,
and induction then gives
\begin{equation}\label{ext:cert:frobmajorant}
 \|P_k\|\le2\times 8^k,
\end{equation}
since $5/8+2/8^2+1/(100\,8^3)<1$.

At a matching point $x=1/c\le1/32$, put $q=8x$ and
$H=\max(1,|\log x|)$. After degree $N_f$, the omitted tail of $D^aW$,
$0\le a\le3$, is bounded by
\begin{equation}\label{ext:cert:frobjets}
 2xHq^{N_f+1}(N_f+3)^aS_3(q).
\end{equation}
The primitive tail has the exact expansion
\[
 C(\infty)-C(y)=\sum_{k\ge0}\frac{x^{k+1}}{k+1}
 \left(P_k(\log x)-\frac{P'_k(\log x)}{k+1}\right),
\]
and its omitted part is bounded by $4xHq^{N_f+1}/(1-q)$.
These estimates bound the full tails and follow directly from~\eqref{ext:cert:frobmajorant}.

The four transported jets are matched to the four enclosed basis jets by a
validated real-ball linear solve. For the inhomogeneous row, the enclosed
particular solution is subtracted first. The solution balls enclose
$(a_j,b_j,c_j,e_j)$, where $P_{j,1}(t)=c_j+e_jt$. The primitive-tail formula
then encloses every $C_j(\infty)$, including $R_\infty$. Finally ordinary
ball operations evaluate~\eqref{ext:cert:Ek}; the ball for $b_0$ lies strictly below
zero, so the division is valid.

\section{Execution and interpretation of the certificate}\label{ext:cert:run}
The first run uses 1400-bit midpoint precision, $N_g=450$, Taylor degree
$N=400$, endpoint degree $N_f=250$, and matching at $y=32$. It takes 161
positive steps. The independent matching check uses 1500 bits, $N=420$,
and $y=48$, with the same germ and endpoint degrees. Both runs independently
certify all four signs and return the rational enclosures in
Table~\ref{ext:cert:intervals}.

The initial bound~\eqref{ext:cert:initialtail} is less than
$5.4\times10^{-309}$. The largest step base remainder $e$ is below
$1.25\times10^{-218}$ in the first run. Their propagated effect is carried through all subsequent operations.
After endpoint matching, the internal radii of the four $E_k$ balls are
below $6\times10^{-139}$ in the first run. The exported decimal ball strings
are deliberately shorter, and the separate rational intervals have width
$10^{-24}$.

The implementation uses the real-ball and validated matrix-solve facilities
of Arb through python-flint~\cite{ext:Arb,ext:Flint}. It does not use an approximate
matrix-solve mode. The finite verification also checks the ten base cases of
\eqref{ext:cert:germmajorant}, the rational induction constants, the 50 exact
square-free germ coefficients, and every exported rational sign.
This is a reproducible computer-assisted proof with the stated analytic
majorants and the standard correctness assumptions for Python/FLINT/Arb;
it is not a formal proof-assistant verification. Agreement between the two matching runs is an additional cross-check.

\appendix
\section{Reproducible interval certificate}\label{ext:code} The Python scripts
used to perform the calculations are permanently archived on
Zenodo~\cite{vanhove_sunset_gw_2026} at
\begin{center}
\href{https://doi.org/10.5281/zenodo.23361560}{\texttt{doi:10.5281/zenodo.23361560}}
\end{center}
Development and maintenance continue in the GitHub repository
\begin{center}
\url{https://github.com/pierrevanhove/Sunset-Gromov-Witten}
\end{center}

\medskip
They require Python 3.10 or later and \texttt{python-flint==0.9.0}. Run without Python's
optimization flag, so assertions remain enabled:
\begin{verbatim}
python3 certify.py
python3 certify.py --bits 1500 --order 420 --match 48
python3 validate.py
\end{verbatim}
All numerical data are generated from the rational differential equation.

\medskip
\noindent{\textbf{The code \texttt{certify.py}}}

\lstset{style=pythonstyle}
\begin{lstlisting}[language=Python,basicstyle=\ttfamily\tiny,breaklines=true]
"""Rigorous real-ball continuation for four-loop endpoint connections.
Every truncation has an explicit majorant; no floating-point input data.
"""
from flint import arb,arb_mat,ctx
from fractions import Fraction as F
from math import factorial,comb
from pathlib import Path
import json,time,argparse
Q=[[0,0,0,0,1],[5,28,63,70,35],[285,1088,1580,1036,259],[900,2700,2925,1350,225]]
S=[[int(k==r==0) for r in range(5)] for k in range(5)]
for k in range(1,5):
    for r in range(1,k+1):S[k][r]=S[k-1][r-1]+r*S[k-1][r]
ODE=[[sum(Q[j][k]*S[k][r] for k in range(r,5) for j in range(4) if d==j+r) for d in range(r+4)] for r in range(5)]
def A(v):return arb(v.numerator)/v.denominator if isinstance(v,F) else arb(v)
def zero(n):return [arb(0) for _ in range(n)]
def add(a,b):return [x+y for x,y in zip(a,b)]
def scale(a,c):return [x*c for x in a]
def deriv(p):return [(k+1)*p[k+1] for k in range(len(p)-1)]+[arb(0)]
def op(p,k,r):
    for _ in range(r):p=add(scale(p,k),deriv(p))
    return p

def pv(p,x):
    v=arb(0)
    for a in reversed(p):v=v*x+a
    return v

def mul(a,b):return [sum((a[j]*b[k-j] for j in range(k+1)),arb(0)) for k in range(5)]
def qshift(q,d):return [sum(q[a]*comb(a,k)*d**(a-k) for a in range(k,5)) for k in range(5)]
def norm(p):return sum((x.abs_upper() for x in p),arb(0)).upper()
def maxabs(p):return max(x.abs_upper() for x in p)
def widen(x,e):return x+arb(0,e.abs_upper())

def germ_coeff(N):
    us=[[arb(1)]+zero(4)]
    for d in range(1,N+1):
        rhs=zero(5)
        for l in range(1,min(3,d)+1):rhs=add(rhs,mul(list(map(arb,qshift(Q[l],d-l))),us[d-l]))
        inv=[arb((-1)**r*comb(r+3,3))/arb(d)**(4+r) for r in range(5)]
        us.append(scale(mul(inv,rhs),-1))
    # Base cases of the coefficient majorant. Strict bound 2*100^d allows rounding.
    for d in range(10):assert norm(us[d]) < 2*arb(100)**d
    assert sum(F(sum(Q[l][a]*10**a for a in range(5)),9**4*100**l) for l in range(1,4))<1
    return us

def initial(N):
    y=arb(1)/512;t=y.log();us=germ_coeff(N);W=[];C=[]
    q=arb(100)/512
    # Uniform tail bound for theta jets of order <=3 and primitive.
    e=1024*q**(N+1)*(N+2)**3*(1+4*q+q*q)/(1-q)**4
    for j in range(5):
        jets=zero(4);c=arb(0)
        for d,row in enumerate(us):
            p=[row[j-r]/factorial(r) for r in range(j+1)]
            for r in range(4):jets[r]+=y**d*pv(op(p,d,r),t)
            if d:
                prim=zero(len(p));dp=p
                for r in range(len(p)):
                    prim=add(prim,scale(dp,arb((-1)**r)/arb(d)**(r+1)));dp=deriv(dp)
                c+=y**d*pv(prim,t)
            else:c+=t**(j+1)/factorial(j+1)
        jets=[widen(v,e) for v in jets];c=widen(c,e)
        W.append([jets[0],jets[1]/y,(jets[2]-jets[1])/y**2,(jets[3]-3*jets[2]+2*jets[1])/y**3]);C.append(c)
    return y,W,C,e

def shifted(poly,y):return [sum((arb(poly[j])*comb(j,k)*y**(j-k) for j in range(k,len(poly))),arb(0)) for k in range(len(poly))]
def step(y,W,C,h,N):
    aa=[shifted(p,y) for p in ODE];r=arb(1)/4;q=arb(1)/4
    # Supremum of augmented logarithmic-jet first-order system on |z|<=r.
    ymax=y*(1+r);den=(1+y*(1-r))*(1+9*y*(1-r))*(1+25*y*(1-r))
    numer=arb(1)+sum((sum(Q[l][:4])*ymax**l for l in range(1,4)),arb(0))
    M=max(arb(1),(numer/den).upper())/(1-r)
    nw=[];nc=[];largest=arb(0)
    for j,jet in enumerate(W):
        theta=[jet[0],y*jet[1],y*y*jet[2]+y*jet[1],y**3*jet[3]+3*y*y*jet[2]+y*jet[1]]
        B=maxabs(theta+[C[j],arb(1)])*(M*r).exp()
        e=B*q**(N+1)*(1+4*q+q*q)/(1-q)**4
        largest=max(largest,e.upper())
        f=[jet[k]/factorial(k) for k in range(4)]
        for n in range(N-3):
            rhs=arb(int(j==4 and n==0))
            for rr in range(5):
                for k,a in enumerate(aa[rr]):
                    if k>n or rr==4 and k==0:continue
                    ix=n-k+rr
                    factor=1
                    for v in range(n-k+1,ix+1):factor*=v
                    rhs-=a*f[ix]*factor
            f.append(rhs/(aa[4][0]*(n+1)*(n+2)*(n+3)*(n+4)))
        vals=[];df=f
        for rr in range(4):
            vals.append(widen(pv(df,h),e*(4*(N+1)/(r*y))**rr));df=deriv(df)
        cf=[C[j]];prev=arb(0)
        for n in range(N):
            qq=(f[n]-prev)/y;cf.append(qq/(n+1));prev=qq
        nw.append(vals);nc.append(widen(pv(cf,h),e))
    return nw,nc,largest

def frob(v,source,N):
    ps=[list(map(arb,v[:2])),list(map(arb,v[2:]))]
    for k in range(2,N+1):
        p=ps[k-1];rhs=add(scale(op(p,k-1,4),259),scale(op(p,k-1,2),26))
        for a,c in enumerate(Q[1]):rhs=add(rhs,scale(op(ps[k-2],k-2,a),c))
        if k>=3:rhs=add(rhs,op(ps[k-3],k-2,4))
        if k==2:rhs[0]-=source
        rhs=add(rhs,scale(deriv(rhs),-arb(2)/k-arb(2)/(k-1)))
        ps.append(scale(rhs,-arb(1)/(225*k*k*(k-1)**2)))
    for k in range(4):assert norm(ps[k]) < 2*arb(8)**k
    return ps

def eval_frob(ps,x):
    N=len(ps)-1;t=x.log();q=8*x;jets=zero(4);tail=arb(0)
    for k,p in enumerate(ps):
        for r in range(4):jets[r]+=x**(k+1)*pv(op(p,k+1,r),t)
        tail+=x**(k+1)*pv(add(p,scale(deriv(p),-arb(1)/(k+1))),t)/(k+1)
    logbound=max(arb(1),t.abs_upper())
    base=2*x*logbound*q**(N+1)*(1+4*q+q*q)/(1-q)**4
    jets=[widen(v,base*(N+3)**r) for r,v in enumerate(jets)]
    tail=widen(tail,4*x*logbound*q**(N+1)/(1-q))
    return jets,tail

def box(v):
    den=10**24
    lo=int((v.lower()*den).lower().floor().unique_fmpz())
    hi=int((v.upper()*den).upper().ceil().unique_fmpz())
    assert v >= arb(lo)/den and v <= arb(hi)/den
    return {'ball':v.str(45),'outer_lower_numerator':str(lo),'outer_upper_numerator':str(hi),'outer_denominator':str(den),'radius':v.rad().str(8)}
def run(bits=1400,N=400,Ng=450,Nf=250,match=32):
    assert match>=32 and N>=4 and Ng>=10 and Nf>=4
    ctx.prec=bits;start=time.time();y,W,C,initial_tail=initial(Ng);steps=0;maxrem=arb(0)
    while y<match:
        h=min(y/16,arb(match)-y);W,C,e=step(y,W,C,h,N);y+=h;steps+=1;maxrem=max(maxrem,e)
        if steps%40==0:print('step',steps,'y',y.str(8),'radius',maxabs([a.rad() for w in W for a in w]).str(5),flush=True)
    x=1/y;bases=[frob([int(a==j) for a in range(4)],0,Nf) for j in range(4)];part=frob([0]*4,1,Nf)
    be=[eval_frob(b,x) for b in bases];pe=eval_frob(part,x)
    matrix=arb_mat([[be[j][0][r] for j in range(4)] for r in range(4)])
    connections=[];limits=[]
    for j,w in enumerate(W):
        target=[w[0],-y*w[1],y*w[1]+y*y*w[2],-y*w[1]-3*y*y*w[2]-y**3*w[3]]
        if j==4:target=[a-b for a,b in zip(target,pe[0])]
        sol=matrix.solve(arb_mat([[v] for v in target]));weights=[sol[k,0] for k in range(4)]
        connections.append(weights);limits.append(C[j]+sum((weights[k]*be[k][1] for k in range(4)),arb(0))+(pe[1] if j==4 else 0))
    R=limits[0];a0,b0=connections[0][:2]
    assert b0<0
    ds=[(v[0]*b0-a0*v[1])/b0**2 for v in connections]
    Es=[sum(((-R)**(k-1-h)/factorial(k-1-h)*ds[h] for h in range(1,k)),arb(0)) for k in range(2,6)]
    for e,sign in zip(Es,[1,-1,1,1]):assert sign*e>0
    out={'bits':bits,'taylor_order':N,'germ_order':Ng,'frobenius_order':Nf,'matching_y':match,'steps':steps,'initial_tail_bound':box(initial_tail),'largest_step_tail_W_bound':box(maxrem),'R_infinity':box(R),'Q_infinity':box(R.exp()),'connections':[[box(v) for v in row] for row in connections],'C_limits':[box(v) for v in limits],'E':[box(v) for v in Es],'seconds':time.time()-start,'certificate_passed':True}
    Path('certificate_'+str(match)+'.json').write_text(json.dumps(out,indent=2)+'\n')
    print('R',R.str(40),'E',[v.str(30) for v in Es],'seconds',out['seconds'],flush=True)
    return out
if __name__=='__main__':
    p=argparse.ArgumentParser();p.add_argument('--bits',type=int,default=1400);p.add_argument('--order',type=int,default=400);p.add_argument('--germ',type=int,default=450);p.add_argument('--frob',type=int,default=250);p.add_argument('--match',type=int,default=32);a=p.parse_args();run(a.bits,a.order,a.germ,a.frob,a.match)

\end{lstlisting}

\medskip
\noindent\textbf{The code \texttt{validate.py}}

\begin{lstlisting}[language=Python,basicstyle=\ttfamily\tiny,breaklines=true]
"""Exact normalization checks and rational checks on exported enclosures."""
from fractions import Fraction as F
from math import factorial,comb
from pathlib import Path
import json
from certify import Q

def conv(a,b,N):
    out=[F(0)]*(N+1)
    for i,x in enumerate(a):
        for j,y in enumerate(b[:N+1-i]):out[i+j]+=x*y
    return out

def exact_recurrence(N):
    out=[[F(1)]+[F(0)]*4]
    for d in range(1,N+1):
        rhs=[F(0)]*5
        for l in range(1,min(d,3)+1):
            p=[F(sum(Q[l][a]*comb(a,k)*(d-l)**(a-k) for a in range(k,5))) for k in range(5)]
            v=conv(p,out[d-l],4);rhs=[a+b for a,b in zip(rhs,v)]
        inv=[F((-1)**k*comb(k+3,3),d**(k+4)) for k in range(5)]
        out.append([-v for v in conv(rhs,inv,4)])
    return out

def squarefree(N):
    f=[F(1,factorial(d)**2) for d in range(N+1)];g=[];h=F(0)
    for d in range(N+1):
        if d:h+=F(1,d)
        g.append(-2*h*f[d])
    ps=[]
    for k in range(5):
        a=[F(1)]+[F(0)]*N
        for i in range(5):a=conv(a,g if i<k else f,N)
        ps.append(a)
    out=[];num=[F(1)]+[F(0)]*4
    for d in range(N+1):
        if d:
            for _ in range(2):num=[num[k]+(num[k-1]/d if k else 0) for k in range(5)]
        out.append([(-1)**d*factorial(d)**2*sum((num[a]*ps[k-a][d]/factorial(k-a) for a in range(k+1)),F(0)) for k in range(5)])
    return out

def interval(b):return F(int(b['outer_lower_numerator']),int(b['outer_denominator'])),F(int(b['outer_upper_numerator']),int(b['outer_denominator']))

if __name__=='__main__':
    a=exact_recurrence(9);b=squarefree(9);assert a==b
    base=[sum(map(abs,row)) <=2*100**d for d,row in enumerate(a)];assert all(base)
    germ_factor=sum(F(sum(Q[l][k]*10**k for k in range(5)),9**4*100**l) for l in range(1,4));assert germ_factor<1
    H1=(F(259)+F(26,16))*F(16,9)*F(13,6)/225
    H2=F(201,225)*F(13,6);H3=F(1,225)*F(13,6)
    assert H1<5 and H2<2 and H3<F(1,100)
    certs=[json.loads(Path('certificate_'+str(m)+'.json').read_text()) for m in [32,48]]
    Es=[]
    for c in certs:
        assert c['certificate_passed']
        bounds=[interval(v) for v in c['E']]
        assert bounds[0][0]>0 and bounds[1][1]<0 and bounds[2][0]>0 and bounds[3][0]>0
        Es.append(bounds)
    for left,right in zip(Es[0],Es[1]):assert max(left[0],right[0])<=min(left[1],right[1])
    out={'exact_squarefree_normalization_coefficients':50,'germ_base_cases':10,'germ_induction_factor':str(germ_factor),'frobenius_H1_bound':str(H1),'frobenius_H2_bound':str(H2),'frobenius_H3_bound':str(H3),'two_certificates_pass':True,'exported_rational_E_intervals':[[str(a),str(b)] for a,b in Es[0]],'status':'All assertions passed. The analytic tail proofs in the report are part of the certificate.'}
    Path('validation.json').write_text(json.dumps(out,indent=2)+'\n')
    Path('exact_germ_check.json').write_text(json.dumps([[str(v) for v in row] for row in a],indent=2)+'\n')
    print(json.dumps(out,indent=2))

\end{lstlisting}


\begin{thebibliography}{99}



\bibitem{BKV}
S.~Bloch, M.~Kerr and P.~Vanhove,
\emph{Local mirror symmetry and the sunset Feynman integral},
Adv. Theor. Math. Phys. \textbf{21} (2017), 1373--1453;
\href{https://arxiv.org/abs/1601.08181}{arXiv:1601.08181}.

\bibitem{PaperI}
H.~Maennel and P.~Vanhove,
\emph{Equal-mass sunset integrals and genus-zero local Gromov--Witten theory of Calabi--Yau $n$-folds},
\href{https://arxiv.org/abs/2610.10828}{arXiv:2610.10828}.

\bibitem{Iritani}
H.~Iritani,
\emph{Asymptotics of the banana Feynman amplitudes at the large complex structure limit},
Adv. Theor. Math. Phys. \textbf{26} (2022), 1239--1245;
\href{https://arxiv.org/abs/2011.05901}{arXiv:2011.05901}.

\bibitem{IP}
E.-N.~Ionel and T.~H.~Parker,
\emph{The Gopakumar--Vafa formula for symplectic manifolds},
Ann. of Math. \textbf{187} (2018), 1--64,
\href{https://doi.org/10.4007/annals.2018.187.1.1}{doi:10.4007/annals.2018.187.1.1}.

\bibitem{Chou}
Y.-C.~Chou,
\emph{Integrality of genus zero Gopakumar--Vafa type invariants of semi-positive varieties},
Int. J. Mod. Phys. A \textbf{39} (2024), 2446014;
\href{https://arxiv.org/abs/2312.01413}{arXiv:2312.01413}.

\bibitem{vGGR}
M.~van Garrel, T.~Graber and H.~Ruddat,
\emph{Local Gromov--Witten invariants are log invariants},
Adv. Math. \textbf{350} (2019), 860--876;
\href{https://arxiv.org/abs/1712.05210}{arXiv:1712.05210}.

\bibitem{ZhouSurvey}
Y.~Zhou,
\emph{Some algebraic and arithmetic properties of Feynman diagrams},
\href{https://arxiv.org/abs/1801.05555}{arXiv:1801.05555}.

\bibitem{Vanhove:2014wqa}
P.~Vanhove,
\emph{The physics and the mixed Hodge structure of Feynman integrals},
Proc. Sympos. Pure Math. \textbf{88} (2014), 161--194;
\href{https://arxiv.org/abs/1401.6438}{arXiv:1401.6438}.

\bibitem{Vanhove:2018mto}
P.~Vanhove, \emph{Feynman integrals, toric geometry and mirror symmetry}, in
\emph{Elliptic Integrals, Elliptic Functions and Modular Forms in Quantum Field
Theory}, Springer, Cham, 2019, 
\href{https://doi.org/10.1007/978-3-030-04480-0_17}{doi:10.1007/978-3-030-04480-0\_17};
\href{https://arxiv.org/abs/1807.11466}{arXiv:1807.11466}.

\bibitem{LairezVanhove}
P.~Lairez and P.~Vanhove,
\emph{Algorithms for minimal Picard--Fuchs operators of Feynman integrals},
Lett. Math. Phys. \textbf{113} (2023), no.~2, 37;
\href{https://arxiv.org/abs/2209.10962}{arXiv:2209.10962}. 

\bibitem{KlemmEtAl}
A.~Klemm, C.~Nega and R.~Safari,
\emph{The $l$-loop banana amplitude from GKZ systems and relative Calabi--Yau periods},
JHEP \textbf{04} (2020), 088,
\href{https://doi.org/10.1007/JHEP04(2020)088}{doi:10.1007/JHEP04(2020)088};
\href{https://arxiv.org/abs/1912.06201}{arXiv:1912.06201}.

\bibitem{BFKNS}
K.~B\"onisch, F.~Fischbach, A.~Klemm, C.~Nega and R.~Safari,
\emph{Analytic structure of all loop banana integrals},
JHEP \textbf{05} (2021), 066,
\href{https://doi.org/10.1007/JHEP05(2021)066}{doi:10.1007/JHEP05(2021)066};
\href{https://arxiv.org/abs/2008.10574}{arXiv:2008.10574}.

\bibitem{CoMath}
D.~Zheng \textit{et al.},
\emph{AI co-mathematician: accelerating mathematicians with agentic AI},
\href{https://arxiv.org/abs/2605.06651}{arXiv:2605.06651}.

\bibitem{CoatesGivental}
T.~Coates and A.~Givental,
\emph{Quantum Riemann--Roch, Lefschetz and Serre},
Ann. of Math. \textbf{165} (2007), 15--53,
\href{https://doi.org/10.4007/annals.2007.165.15}{doi:10.4007/annals.2007.165.15};
\href{https://arxiv.org/abs/math/0110142}{arXiv:math/0110142}. 

\bibitem{CCITtwisted}
T.~Coates, A.~Corti, H.~Iritani and H.-H.~Tseng,
\emph{Computing genus-zero twisted Gromov--Witten invariants},
Duke Math. J. \textbf{147} (2009), 377--438;
\href{https://arxiv.org/abs/math/0702234}{arXiv:math/0702234}.

\bibitem{Verrill}
H.~Verrill,
\emph{Root lattices and pencils of varieties},
J. Math. Kyoto Univ. \textbf{36} (1996), no.~2, 423--446.

\bibitem{BBBG}
D.~H.~Bailey, J.~M.~Borwein, D.~Broadhurst and M.~L.~Glasser,
\emph{Elliptic integral evaluations of Bessel moments},
J. Phys. A \textbf{41} (2008), 205203,
\href{https://doi.org/10.1088/1751-8113/41/20/205203}{doi:10.1088/1751-8113/41/20/205203};
\href{https://arxiv.org/abs/0801.0891}{arXiv:0801.0891}.

\bibitem{BKV3}
S.~Bloch, M.~Kerr and P.~Vanhove,
\emph{A Feynman integral via higher normal functions},
Compos. Math. \textbf{151} (2015), no.~12, 2329--2375,
\href{https://doi.org/10.1112/S0010437X15007472}{doi:10.1112/S0010437X15007472};
\href{https://arxiv.org/abs/1406.2664}{arXiv:1406.2664}.

\bibitem{Vanhove2026}
P.~Vanhove,
\emph{The multiloop sunset to all orders},
Ann. Henri Poincar\'e (2026),
\href{https://doi.org/10.1007/s00023-026-01734-8}{doi:10.1007/s00023-026-01734-8};
\href{https://arxiv.org/abs/2603.03183}{arXiv:2603.03183}.

\bibitem{BorweinSalvy}
J.~M.~Borwein and B.~Salvy,
\emph{A proof of a recurrence for Bessel moments},
Exp. Math. \textbf{17} (2008), 223--230;
\href{https://arxiv.org/abs/0706.1409}{arXiv:0706.1409}.

\bibitem{KP}
A.~Klemm and R.~Pandharipande,
\emph{Enumerative geometry of Calabi--Yau 4-folds},
Commun. Math. Phys. \textbf{281} (2008), 621--653;
\href{https://arxiv.org/abs/math/0702189}{arXiv:math/0702189}.

\bibitem{Lee}
Y.-P.~Lee,
\emph{Quantum $K$-theory~I: Foundations},
Duke Math. J. \textbf{121} (2004), 389--424;
\href{https://arxiv.org/abs/math/0105014}{arXiv:math/0105014}.

\bibitem{ext:LP}
Y.-P.~Lee and R.~Pandharipande,
\emph{A reconstruction theorem in quantum cohomology and quantum $K$-theory},
Amer. J. Math. \textbf{126} (2004), 1367--1379;
\href{https://arxiv.org/abs/math/0104084}{arXiv:math/0104084}.

\bibitem{PZ}
R.~Pandharipande and A.~Zinger,
\emph{Enumerative geometry of Calabi--Yau 5-folds},
in \emph{New Developments in Algebraic Geometry, Integrable Systems and Mirror Symmetry},
Adv. Stud. Pure Math. \textbf{59} (2010), 239--288;
\href{https://arxiv.org/abs/0802.1640}{arXiv:0802.1640}.

\bibitem{ext:BSWZ}
J.~M.~Borwein, A.~Straub, J.~Wan and W.~Zudilin,
\emph{Densities of short uniform random walks}, with an appendix by D.~Zagier,
Canad. J. Math. \textbf{64} (2012), 961--990;
\href{https://arxiv.org/abs/1103.2995}{arXiv:1103.2995}.

\bibitem{ext:DLMF}
NIST Digital Library of Mathematical Functions,
\url{https://dlmf.nist.gov}.

\bibitem{ext:Arb}
F.~Johansson,
\emph{Arb: efficient arbitrary-precision midpoint-radius interval arithmetic},
IEEE Trans. Comput. \textbf{66} (2017), 1281--1292,
\href{https://doi.org/10.1109/TC.2017.2690633}{doi:10.1109/TC.2017.2690633}; 
\href{https://arxiv.org/abs/1611.02831}{arXiv:1611.02831}. 

\bibitem{ext:Flint}
The FLINT and python-flint contributors,
\emph{python-flint}, version~0.9.0,
\url{https://python-flint.readthedocs.io}.

\bibitem{vanhove_sunset_gw_2026}
P.~Vanhove,
\emph{Sunset-Gromov-Witten: Code release for Sunset Gromov-Witten papers},
v1.0.0, Zenodo (2026),
\href{https://doi.org/10.5281/zenodo.23361560}{doi:10.5281/zenodo.23361560}.
\end{thebibliography}
\end{document}